\documentclass[aps,pra,showpacs,notitlepage,twocolumn,superscriptaddress,nofootinbib]{revtex4-1}
\usepackage{bbm}
\usepackage{mathrsfs}

\usepackage{epsfig}
\usepackage{soul,xcolor}
\usepackage{graphicx}

\usepackage{amsfonts}
\usepackage{amsthm}
\usepackage[figuresright]{rotating}
\usepackage{amssymb}
\usepackage{amsmath}
\usepackage{dcolumn}
\usepackage{enumitem}
\usepackage{physics}
\usepackage{float}
\usepackage{bm}
\usepackage{verbatim}
\usepackage{braket}
\usepackage{kbordermatrix}
\usepackage[normalem]{ulem}
\usepackage[ruled,vlined,linesnumbered]{algorithm2e}
\usepackage{setspace}
\usepackage[colorlinks,linkcolor=blue,anchorcolor=blue,citecolor=blue,urlcolor=blue]{hyperref}
\newtheorem{theorem}{Theorem}
\newtheorem{lemma}[theorem]{Lemma}

\def\>{\rangle}
\def\<{\langle}

\newcommand{\bea}{\begin{eqnarray}}
\newcommand{\eea}{\end{eqnarray}}
\newcommand{\be}{\begin{equation}}
\newcommand{\ee}{\end{equation}}
\newcommand{\bfg}{\begin{figure}[htbp]}
\newcommand{\efg}{\end{figure}}

\newcommand{\mattwocb}[4]{\left[
	\begin{array}{cc}{#1}&{#2}\\{#3}&{#4}\end{array}\right]}
\def\lbL{\lb\rule{0pt}{2.4ex}}
\def\lpL{\left(\rule{0pt}{2.4ex}}

\def\lb{\left[}			\def\rb{\right]}
			\def\rp{\right)}

\usepackage{xcolor}
\definecolor{crimson}{RGB}{140,41,53}

\begin{document}

    \title{A Grand Unification of Quantum Algorithms}
    \author{John M. Martyn}\affiliation{Center for Theoretical Physics, Massachusetts Institute of Technology, Cambridge, Massachusetts 02139, USA}\affiliation{Department of Physics, Massachusetts Institute of Technology, Cambridge, Massachusetts 02139, USA}
    \author{Zane M. Rossi}\affiliation{Department of Physics, Massachusetts Institute of Technology, Cambridge, Massachusetts 02139, USA}
    \author{Andrew K. Tan}\affiliation{Department of Physics, Co-Design Center for Quantum Advantage, Massachusetts Institute of Technology, Cambridge, Massachusetts 02139, USA}
    \author{Isaac L. Chuang}\affiliation{Department of Physics, Co-Design Center for Quantum Advantage, Massachusetts Institute of Technology, Cambridge, Massachusetts 02139, USA} \affiliation{Center for Ultracold Atoms, and Research Laboratory of Electronics, Massachusetts Institute of Technology, Cambridge, Massachusetts 02139, USA}

\begin{abstract}
    Quantum algorithms offer significant speedups over their classical counterparts for a variety of problems. The strongest arguments for this advantage are borne by algorithms for quantum search, quantum phase estimation, and Hamiltonian simulation, which appear as subroutines for large families of composite quantum algorithms. A number of these quantum algorithms were recently tied together by a novel technique known as the quantum singular value transformation (QSVT), which enables one to perform a polynomial transformation of the singular values of a linear operator embedded in a unitary matrix. In the seminal GSLW'19 paper on QSVT [Gily\'en, Su, Low, and Wiebe, ACM STOC 2019], many algorithms are encompassed, including amplitude amplification, methods for the quantum linear systems problem, and quantum simulation. Here, we provide a pedagogical tutorial through these developments, first illustrating how quantum signal processing may be generalized to the quantum eigenvalue transform, from which QSVT naturally emerges. Paralleling GSLW'19, we then employ QSVT to construct intuitive quantum algorithms for search, phase estimation, and Hamiltonian simulation, and also showcase algorithms for the eigenvalue threshold problem and matrix inversion. This overview illustrates how QSVT is a single framework comprising the three major quantum algorithms, suggesting a \emph{grand unification} of quantum algorithms.
\end{abstract}

\maketitle

\section{Introduction}\label{sec:Introduction}
    Algorithms solve problems by presenting a process or set of rules to be followed, utilizing a basic set of building blocks provided. Computer science traditionally employs Boolean circuit components as the basic blocks, from which standard arithmetic operations may be composed, as Boolean functions.  Quantum computation employs a different set of basic blocks, typically unitary operations on one or more two-state systems (qubits), to realize quantum circuits.  A fundamental challenge arises, however, when seeking to unite the world of quantum circuits with that of Boolean functions: in general, Boolean functions need not be reversible, whereas quantum circuits are manifestly unitary transforms, and must thus be invertible.
    
    Early in the history of quantum computation, this barrier was transcended by seminal work~\cite{bennett1989time} showing that all Boolean functions can be made reversible, with only a small overhead in space and time.  Toffoli and Fredkin famously illustrated this idea by showing how the ideal Newtonian dynamics of finite-radii spheres (billiard balls) can be used to simulate reversible Boolean circuits, via their collisions~\cite{fredkin1982conservative}.  Following this concept, simulation of arbitrary Boolean functions can also be accomplished using quantum circuits, by first embedding the desired function into a reversible Boolean circuit, then constructing a quantum circuit realizing this invertible transform.  Such an embedding is a core part of Shor's quantum factoring algorithm~\cite{shor1994algorithms}, for example, as used in the modular exponentiation of an input number.
    
    Intriguingly, however, the two other major ``primordial'' quantum algorithms, Grover's quantum search algorithm~\cite{grover1996fast}, and the Hamiltonian simulation algorithm~\cite{feynman1982simulating,lloyd1996universal}, do not employ an embedding of a reversible Boolean function.  In fact, a key part of the quantum factoring algorithm is its use of the quantum Fourier transform, which has no direct classical analogue, in the sense that it is not at all like a quantum embedding of a reversible Boolean function for the Fourier transform. And yet, all three of these algorithms provide solutions to problems with clear classical counterparts, and attain known speedups over the comparable classical ``Boolean function'' approaches.  So wherein lies the ability of quantum algorithms to address and speed up the solution to a classically specified problem?
    
    As illustrated in this tutorial, a key idea in uniting quantum and classical computation is {\em not} to first make classical computation reversible.  Instead, observe that the dynamical behavior of a {\em subsystem} of a quantum system can be {\em non-unitary}, and thus can directly realize irreversible, non-linear functions.  An extreme case of this is projective measurement: the billiard ball model can realize non-invertible gates simply by discarding balls, but this would be inefficient.  More constructively, the recently developed framework of {\em quantum signal processing} (QSP)~\cite{Low_2016,Low_2019} provides a systematic method to make a quantum subsystem transform under nearly arbitrary polynomial functions of degree $d$, using $\mathcal{O}(d)$ elementary unitary quantum operations.  Crucially, the polynomial describes not the output of the full quantum system, but only a very specific and clearly identified subsystem.  And remarkably, the essential ideas behind QSP originate from the early days of practical control of two-level quantum systems, with nuclear magnetic resonance~\cite{freeman1998spin,levitt1986composite,wimperis1994broadband}.
    
    With this framework, we present in this tutorial a pedagogical overview of the modern approach to quantum search, factoring, and simulation, focusing on how all three of these central quantum algorithms may be unified as instances of the recently developed {\em quantum singular value transform (QSVT) algorithm}~\cite{Gily_n_2019}. The QSVT algorithm generalizes QSP and efficiently applies a polynomial transformation to the singular values of a linear operator (governing a particular subsystem) embedded in a larger unitary operator. And more recently, such a singular value transformation has been generalized to apply to an operator embedded in a block of a Hamiltonian~\cite{lloyd2021hamiltonian}.
    
    Singular values naturally arise in this context from the fact that the input and output spaces of the embedded linear operator may be of different sizes.
    The polynomial transformation of the singular values is achieved by applying a specific sequence of SU(2) rotations to the embedded subspace, where each rotation is parameterized by an angle $\phi_k \in \mathbb{R}$. The QSVT algorithm is parametric in that the polynomial transformation is completely characterized by the choice of phase angles $\{\phi_k\}$. Moreover, given the desired polynomial transformation, the QSP phase angles which generate it may be classically efficiently and stably computed~\cite{Low_2016, Gily_n_2019, efficient_qsp_phase_factors}.  
        
    This seemingly simple parameterization endows QSVT with immense flexibility and power. Using QSVT as a subroutine, we present quantum algorithms for the search problem and phase estimation, and give simplified arguments for optimal Hamiltonian simulation by QSVT and matrix inversion by QSVT. We also present a QSVT-based algorithm for \emph{the eigenvalue threshold problem}, wherein one wishes to know if a (normal) matrix has an eigenvalue above some threshold, and use this to simplify the presentation of the use of QSVT for phase estimation. Each algorithm is realized simply as an instance of QSVT with a natural oracle, adaptively repeated and interspersed in a common pattern with simple quantum gates and measurements. 
    
    This common pattern is particularly fascinating. At face value, the algorithms for the search problem, factoring (aka phase estimation), and Hamiltonian simulation appear to share no similar structures, owing their quantum speedups to different sources, and yet they can all be derived from a single algorithmic primitive, and interpolated between by a simple change of parameters. In addition, these three central algorithms form the foundation for most quantum algorithms currently known. For instance, the Harrow-Hassidim-Lloyd algorithm for linear systems~\cite{Harrow_2009} incorporates Hamiltonian simulation and phase estimation in order to invert a Hermitian matrix, and similarly the quantum counting algorithm integrates quantum search with phase estimation to count the number of marked elements in an unstructured set~\cite{Brassard_1998}. As shown in this tutorial, by simply adjusting the parameters of QSVT, one can construct nearly all known quantum algorithms. It is in this sense that QSVT provides a \emph{grand unification} of quantum algorithms. 
    
    While some of these applications have been covered in recent works on QSVT~\cite{Gily_n_2019, rall2021faster}, here we aim to present these constructions as pedagogically as possible, providing detailed procedures and intuition for each, and including explicit examples to support abstract ideas where appropriate. We also supply performance bounds and resource requirements for each of the algorithms presented here, which we anticipate will be helpful.  It is our hope that this presentation will make QSVT and QSP more accessible, catalyzing future developments in quantum algorithms.
    
    Throughout this tutorial, we will assume familiarity with basic concepts in quantum computing, such as unitary dynamics and measurement, as well as the conventional quantum algorithms for search, phase estimation, and Hamiltonian simulation. For a comprehensive review of these subjects, see~\cite{nielsen2010quantum, preskill1998lecture, kitaev2002classical}. Further, we aim to present this work in a manner accessible to readers without prior knowledge of QSP and QSVT, but if more information on these topics is needed,~\cite{low2017quantum, Low_2016, gilyen2019quantum, Gily_n_2019} may serve as helpful references.

    \subsection{Road Map}
        
    We share the story of this grand unification of quantum algorithms by first surveying the development of QSVT in Section~\ref{sec:QSP_To_QSVT}, beginning with quantum signal processing and then demonstrating how this technology leads to quantum eigenvalue transforms and ultimately the quantum singular value transformation. Thereafter, we detail a statement of grand unification for currently known quantum algorithms, presenting QSVT-based algorithms for the search problem and the eigenvalue threshold problem in Secs.~\ref{sec:Search} and~\ref{sec:Threshold}, respectively. Further fulfilling the unification promise, we introduce an algorithm for phase estimation by QSVT in Section~\ref{sec:PhaseEstimation}, and show how the quantum Fourier transform emerges from cascading QSVT sequences. We then highlight in Section~\ref{sec:FuncEval} how QSVT can yield intuitive algorithms for Hamiltonian simulation and matrix inversion. Lastly, in Section~\ref{sec:Discussion}, we explore the implications of these results and discuss areas of future research in the utility of QSVT.

\section{From QSP to QSVT} \label{sec:QSP_To_QSVT}

    In this section, we develop the framework of QSVT, starting with quantum signal processing and then providing a concrete application of this technology to amplitude amplification. Building off this example, we establish quantum eigenvalue transforms and finally quantum singular value transforms. Overall, this section is meant to be pedagogical and easily accessible to readers with a background in basic quantum circuits and linear algebra, as the constructions of this section underlie the remainder of the work. For more details on QSP and QSVT, including alternate introductions to the topics, rigorous proofs, and applications not covered here, see \cite{low2017quantum, Low_2016, gilyen2019quantum, Gily_n_2019}.

\subsection{Quantum Signal Processing}
\label{sec:qsp}

Quantum signal processing (QSP) generalizes the results of composite pulse sequences~\cite{Low_2016, Low_2017, Low_2019}, and is built on the idea of interleaving two kinds of single qubit rotations: a {\em signal} rotation operator $W$, and a {\em signal processing} rotation operator $S$.  These rotation operations are about different axes through the Bloch sphere, e.g. commonly $W$ is an $x$-rotation, while $S$ is a $z$-rotation.  Moreover, the signal rotation always rotates through the same angle $\theta$, whereas the signal processing rotation rotates through a variable angle according to some pre-determined sequence.

For example, let the signal rotation operator be
    \bea
        W(a) = \mattwocb{a}{i \sqrt{1-a^2}}{i \sqrt{1-a^2}}{a}\,,
        \label{eq:wa_qsp}
    \eea
which is an $x$-rotation by angle $\theta = -2\cos^{-1} a$. And let the signal processing rotation operator be
    \bea \label{eq:qsp_rotation}
        S(\phi) = e^{i \phi Z}\,,
    \eea
which is a $z$-rotation by angle $-2\phi$.  For a tuple of phases $\vec{\phi} = (\phi_0, \phi_1, ... \phi_d) \in \mathbb{R}^{d+1}$, and using these conventions\footnote{The convention presented in Eq.~(\ref{eq:qsp_uphi}) may be called the Wx convention. See Appendix~\ref{sec:QSPConventions} for other conventions.}, the QSP operation sequence $U_{\vec\phi}$ is defined as
    \be
        U_{\vec\phi} = e^{i \phi_0 Z} \prod_{k=1}^d  W(a)  e^{i\phi_k Z}
        \label{eq:qsp_uphi}
        \,.
    \ee
 What is interesting is how QSP can modify the incoming signal. Suppose $\vec{\phi} = (0,0)$, i.e. there is no processing, such that $U_{\vec{\phi}} = W(a)$ is just the unchanged signal.  If we plot the probability of a $|0\>$ qubit input staying unchanged under this operation, i.e. $p = |\<0|U_{\vec{\phi}}|0\>|^2$, as a function of $\theta$, we find a nice cosinusoid (the lower plot in Figure~\ref{fig:qsp_response_functions}), because for this case $p = \cos^2\left(\frac{\theta}{2}\right)$.  Now if we do some signal processing, by letting $\vec{\phi} = (\pi/2, -\eta, 2\eta, 0, -2\eta, \eta)$, where $\eta = \frac{1}{2} \cos^{-1}\left(-1/4\right)$, then for the new $U_{\vec{\phi}}$ using these phases, we find the new probability $p = |\<0|U_{\vec{\phi}}|0\>|^2 = \frac{1}{8} \cos^2\left(\frac{\theta}{2}\right) \Big(3 \cos^8\left(\frac{\theta}{2}\right)$ $- 15 \cos^6\left(\frac{\theta}{2}\right) +35 \cos^4\left(\frac{\theta}{2}\right) - 45 \cos^2\left(\frac{\theta}{2}\right) + 30\Big)$ (the upper dotted dashed line in Figure~\ref{fig:qsp_response_functions}), which for small $\theta$ is approximately $1-\frac{5}{8}\left( \frac{\theta}{2}\right)^6$. This has the nice property that the qubit remains unflipped for a wide range of signal angles, but then a sharp transition happens around $\theta\approx 2\pi/3$. This increases sensitivity to specific values of the signal! 

    \begin{figure}[htpb]
        \centering
        \includegraphics[width=\columnwidth]{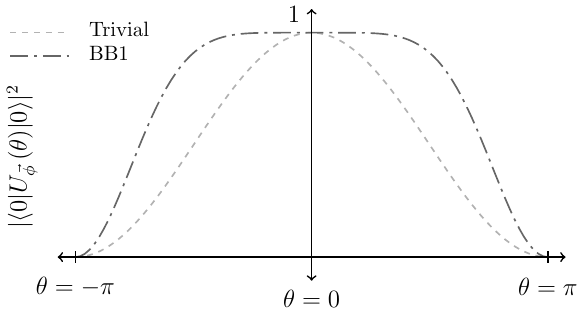}
        \caption{The transition probabilities for $\lvert 0 \rangle \rightarrow \lvert 0 \rangle$ after the application of $\exp(i\theta X/2)$ (uniform dashed) and the QSP sequence for the BB1 protocol (dotted dashed). The latter transition probability is seen to be near-one for a wider range of $\theta$.}
        \label{fig:qsp_response_functions}
    \end{figure}

Thus, sequences like this are widely employed in magnetic resonance imaging to increase image contrast.  This particular sequence is famous in the field of NMR, and is known as the ``BB1'' pulse sequence~\cite{wimperis1994broadband}. Many such ``composite pulse'' sequences are known~\cite{freeman1998spin,RevModPhys.76.1037,wolfowicz2007pulse}, with a wide range of variations, and in experimental quantum computation, they serve to suppress specific kinds of errors and enhance sensitivity to specific kinds of signals, across implementations ranging from quantum dots~\cite{wang2012composite} and nitrogen vacancy centers in diamonds~\cite{aiello2013composite}, to trapped ions~\cite{PhysRevA.92.060301,Tomita_2010,Ivanov:11} and superconducting qubits~\cite{chow2013microwave}.  

That QSP sequences can have such signal transformation properties is well known, because in general the matrix element $P(a) = \<0|U_{\vec{\phi}}|0\>$ becomes a polynomial in $a$, with the order of the polynomial being at most the length of the sequence of QSP phases $\vec{\phi}$.  Specifically, for example, for $\vec{\phi} = (0,0)$, $P(a) = a$; for $\vec{\phi} = (0,0,0)$, $P(a) = 2a^2-1$; and for $\vec{\phi} = (0,0,0,0)$, $P(a) = 4a^3-3a$. These are the Chebyshev polynomials of the first kind, $T_{d}(a)$.

Perhaps the most remarkable property of the QSP sequence of Eq.~(\ref{eq:qsp_uphi}), however, is the reverse of this statement: it turns out that for a given polynomial $P(a)$ (subject to some reasonable constraints), there exists a set of QSP phase angles $\vec{\phi}$ such that $P(a) = \<0|U_{\vec{\phi}}|0\>$. Specifically~\cite{Low_2016}:

    \begin{theorem}[Quantum Signal Processing]
    The QSP sequence $U_{\vec{\phi}}$ produces a matrix which may be expressed as polynomial function of $a$:
    \begin{equation}\label{eq:qsp_thm}
        \begin{gathered}
            e^{i \phi_0 Z} \prod_{k=1}^d  W(a)  e^{i\phi_k Z} = \\
            \mattwocb{P(a)}{iQ(a)\sqrt{1-a^2}}{iQ^*(a)\sqrt{1-a^2}}{P^*(a)}
        \end{gathered}
    \end{equation}
    for $a\in[-1, 1]$, and a $\vec{\phi}$ exists for any polynomials $P$, $Q$ in $a$ such that:
    \begin{enumerate}[label=(\roman*)]
      \item ${\rm deg}(P) \leq d$, ${\rm deg}(Q) \leq d-1$,
      \item $P$ has parity $d~ {\rm mod}~ 2$ and $Q$ has parity $(d-1)~ {\rm mod}~ 2$,
      \item $|P|^2 + (1-a^2) |Q|^2 = 1$.
    \end{enumerate}
    \label{thm:qsp}
    \end{theorem}
The forward part of this theorem is easily proven by induction, starting from the $d=0$ case, for which $P=e^{i\phi_0}$ and $Q=0$. 
The reverse direction of this theorem is more involved, and can be proven in a number of ways, including an elegant interpretation involving Laurent polynomial algebras~\cite{Haah_2019, chao2020finding}. 

Often however, we are interested not in the unitaries that can be constructed with QSP but rather the achievable polynomial transformations of the input, ${\rm Poly}(a)$, in a subsystem.
If as above, we choose ${\rm Poly}(a) = \braket{0 | U_{\vec\phi} | 0} = P(a)$ we are limited to $P$ for which there exists a polynomial $Q$ satisfying the conditions of Theorem~\ref{thm:qsp}.
This can be quite a limiting condition for some applications.
For example, for $a = \pm 1$, $W(\pm 1)$ is proportional to the identity matrix and the entire sequence QSP sequence collapses to a single $z$-rotation limiting us to polynomials ${\rm Poly}(a)$ such that $|{\rm Poly}(\pm 1)| = 1$.
This limitation can be overcome by instead defining ${\rm Poly}(a) = \braket{+ | U_{\vec\phi} | +} = \text{Re}(P(a)) + i\text{Re}(Q(a))\sqrt{1-a^2}$.
In this case, it can be shown that we can accurately approximate any real polynomial with parity $d~ {\rm mod}~ 2$ such that $\deg({\rm Poly}) \le d$, and $|{\rm Poly}(a)| \le 1 ~ \forall ~a \in [-1, 1]$. 
This can be achieved by selecting an appropriate $P$ whose real part approximates the desired function and a $Q$ with small real component.
With this convention, the set of achievable polynomials is sufficiently expressive for all of the applications described in this tutorial.  In general, the basis employed for the desired polynomial may be called the {\em signal basis}, and unless otherwise specified, below we take this basis to be $|+\>, |-\>$.
Appendix~\ref{sec:QSPConventions} elaborates on this, and discusses the different conventional statements of the QSP Theorem~\ref{thm:qsp}.

While Theorem~\ref{thm:qsp} guarantees the existence of such a $\vec{\phi}$, it does not provide a method for determining it. Fortunately, Remez-type exchange algorithms can efficiently compute a $\vec{\phi}$ that produces a good approximation to any feasible polynomials $P$ and $Q$~\cite{Low_2016}. Further, more efficient and numerically stable algorithms have also been found~\cite{chao2020finding, Haah_2019, arute2019quantum}, and novel optimization techniques are currently under development~\cite{lin2021optimization}. Appendix~\ref{sec:qsp_angle_examples} gives explicit examples of $\vec{\phi}$ for a wide family of polynomials used in realizing quantum algorithms, and illustrates an open-source code package accompanying this tutorial for generating QSP phase angle coefficients.

\subsection{An Application to Amplitude Amplification and Search}\label{sec:AA_and_search}

    Toward motivating QSVT from QSP, we introduce an illustrative example of how multi-qubit problems may be simplified by identifying qubit-like subsystems, to which the ideas of QSP may be applied. The concept demonstrated in this subsection is essentially that of \emph{qubitization}~\cite{Low_2019}, which forms a core tenet of QSVT. 
    
    Specifically, we discuss the problem of amplitude amplification; a similar construction will be discussed in Section~\ref{sec:Search} employing the major theorems of QSVT, but here we begin from basic principles. It is hoped that the geometric intuition behind the argument which follows, and the expediency of the fully-developed construction in Section~\ref{sec:Search}, will complement each other.

    The utility of Theorem~\ref{thm:qsp} is nicely demonstrated in solving the following problem. Suppose you are given a unitary $U$ (which may act on some Hilbert space of large dimension; i.e., larger than just a qubit) and its inverse, $U^\dagger$, as well as two operators $A_\phi$ and $B_\phi$, each of which rotates the phases of a specific, privileged state, namely:
        \bea
            A_\phi &=& e^{i\phi |A_0\>\<A_0|} 
            \label{eq:aa_aphi}
            \\
            B_\phi &=& e^{i\phi |B_0\>\<B_0|}
            \,.
        \eea
    The challenge is to construct a circuit $Q$ using $U$, $U^\dagger$, $A_\phi$, and $B_\phi$ such that 
        \bea
             |\< A_0 | Q |B_0\>| \longrightarrow 1
        \eea
    in the limit of a sufficiently large circuit, and assuming that the original matrix element $\< A_0 | U |B_0\>$ of $U$ is nonzero.
    
    This problem is known as {\em amplitude amplification}, and remarkably it can be solved without knowing the specific initial value $\<
    A_0 | U |B_0\>$, by using the {\em oblivious fixed-point amplitude amplification} quantum algorithm. We will show that such an algorithm arises only from Theorem~\ref{thm:qsp}, even in the multiple-qubit setting, by recognizing that there are two concentric Bloch spheres (qubit-like spaces) in this problem.
    
    Specifically, one can recognize that $U |B_0\>$ is a quantum state which has a non-zero component along $|A_0\>$, and another component perpendicular to $|A_0\>$.  So we define
        \be
            |A_\perp\> = \frac{1}{\mathcal{N}} \Big( I-|A_0\>\<A_0| \Big) U|B_0\>
            \,,
        \ee
    where $\mathcal{N}$ is the normalization factor needed to make $|A_\perp\>$ a unit vector. Then
        \bea
            U |B_0\> = a |A_0\>  + \sqrt{1-a^2} |A_\perp\>
        \eea
    for $a=\<A_0|U|B_0\>$ (we may assume that $a$ is real because a possible phase may be absorbed into $|B_0\rangle$)). Similarly, we may define some $|B_\perp\>$ such that
        \bea
            U |B_\perp\> = -a|A_\perp\> + \sqrt{1-a^2} |A_0\>
        \,.
        \eea
    These ideas are illustrated in the diagram of Figure~\ref{fig:amplitude-amplification-bloch-sphere}, using the familiar Bloch sphere.
    
        \begin{figure}[htpb]
            \centering
            \includegraphics[width=\columnwidth]{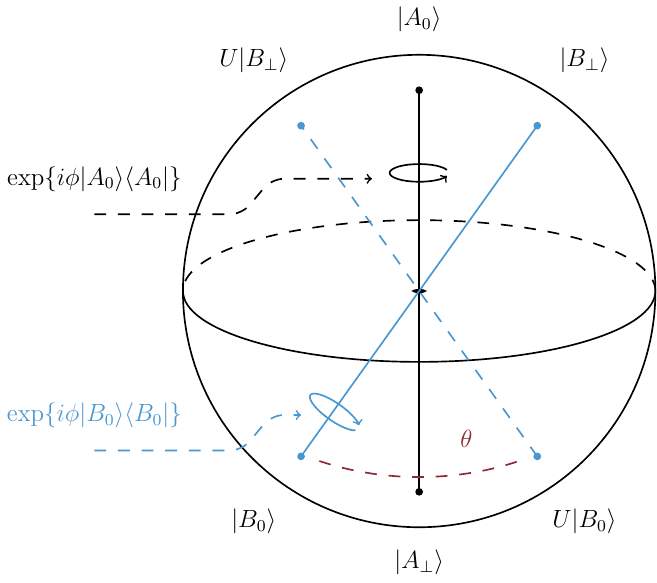}
            \caption{An illustration of amplitude amplification, where one desires to prepare the state $\lvert A_0\rangle$ (here the north pole of some Bloch sphere) obliviously to this state, given only $\lvert B_0 \rangle$, an operator $U$ whose $\langle A_0 \lvert U \rvert B_0\rangle = a$ component is non-zero, and the ability to rotate about the states $\lvert B_0 \rangle$ (blue) and $\lvert A_0 \rangle$ (black) through chosen angles. The standard Grover iterate can be recovered in this model if one desires, by producing a simple rotation by $\theta = 2\cos^{-1}\left(\sqrt{1-a^2} \right) = 2\sin^{-1}a$ (red) toward $\lvert A_0 \rangle$.}
            \label{fig:amplitude-amplification-bloch-sphere}
        \end{figure}
    
    Thus, on the two-dimensional Hilbert space spanned by $\{|A_0\>, |A_\perp\>\}$, the action of $U$ is that of a $2$$\times$$2$ unitary:
        \begin{align}
            \begin{aligned}
                U =\  &a \left (\rule{0pt}{2.4ex} |A_0\>\<B_0| - |A_\perp\>\<B_\perp| \right) \\ 
                &+\sqrt{1-a^2} \left( \rule{0pt}{2.4ex} |A_\perp\>\<B_0| + |A_0\>\<B_\perp| \right),
            \end{aligned}
        \end{align}
    which is convenient to represent in matrix form, as
        \bea
            U = \kbordermatrix{\mbox{} & |B_0\rangle & |B_\perp \rangle \\
              |A_0\rangle & a    & \sqrt{1-a^2} \\ 
              |A_\perp \rangle & \sqrt{1-a^2} & -a 
            }
            \,,
            \label{eq:ampamp_rot}
        \eea
    where the labels on the matrix indicate that columns act on $|B_0\>, |B_\perp\>$ (from left to right), and the rows act on $|A_0\>, |A_\perp\>$ (from top to bottom), such that $U$ brings a state in the $B$ basis into a state in the $A$ basis.  These two bases are the two Bloch spheres (or qubit bases) encoded in the problem. Moreover, $U$ is a reflection operation, which we may represent in this qubit-like abstraction as $R(a) = X R_y(\theta)$, with $\theta = 2\cos^{-1}\left(\sqrt{1-a^2} \right) = 2\sin^{-1}a$.
    
    This two-Bloch sphere picture provides the intuitive basis for the following theorem:
        \begin{theorem}[Amplitude Amplification]
            Given unitary $U$, its inverse $U^\dagger$, and operators $A_\phi = e^{i\phi |A_0\>\<A_0|}$, $B_\phi = e^{i\phi |B_0\>\<B_0|}$, 
              \bea
                    \<A_0 | A_{\phi_0} \left[ \prod_{k=1}^{\frac{d-1}{2}} U B_{\phi_{2k-1}} U^\dagger A_{\phi_{2k}} \right] U |B_0\> = {\rm Poly}(a)\,,
                    \label{eq:thm_aa_qsp}
              \eea
            where $d$ is odd, and ${\rm Poly}(a)$ is a polynomial in $a = \<A_0|U|B_0\>$ of degree at most $d$ that obeys the conditions on $P$ from Theorem~\ref{thm:qsp}.
            \label{thm:aa}
        \end{theorem}
    
    Why does this work? Note first that $U|B_0\>$ lives in the $A$-qubit space, spanned by $|A_0\rangle$ and $|A_\perp \rangle$. This qubit then gets rotated around its ``$z$'' axis by $A_\phi$. The $U^\dagger$ then rotates this around the $y$ axis (and also does a reflection around $Z$, but we can disregard that for this intuition).  $U^\dagger$ also maps the state from the $A$-qubit space back into the $B$-qubit space, spanned by $|B_0\rangle$ and $|B_\perp \rangle$. Next, $B_\phi$ rotates the state around the $B$-qubit space's $z$-axis.  Then the leftmost $U$ does another $y$-rotation (and reflection), and maps the state back into the $A$-qubit space.  The sequence in the square brackets maps the state back and forth between the two qubit bases, sandwiching $y$-axis rotations with $z$-axis rotations. This sandwich of rotations is thus just doing quantum signal processing as in Sec. \ref{sec:qsp}; we understand this behavior well!
    
    The formal proof of this theorem begins by noting that on the subspaces defined by the two Bloch spheres defined above, $U = U^\dagger$.  Moreover, note that $R(a)$, the $ 2\times 2$ unitary representation of $U$ (in the qubit bases relevant to our problem), is related to the $W(a)$ of Eq.~(\ref{eq:wa_qsp}) by 
        \bea
        	R(a) = -i e^{i\frac{\pi}{4} Z} W(a)  e^{i\frac{\pi}{4} Z}\,,
        	\label{eq:wa_to_ra}
        \eea
    Substituting this into Eq.~(\ref{eq:thm_aa_qsp}), and recognizing that $A_\phi$ and $B_\phi$ simply become $z$-axis rotations, we obtain
        \bea
            \<A_0 | \left( e^{i\phi'_0 Z} \left[ \prod_{k=1}^{d} W(a) e^{i\phi'_{k} Z} \right] \right) U |B_0\> = {\rm Poly}(a), \quad 
            \label{eq:aa_qsp_seq}
        \eea
    where $\{ \phi'_k\}$ are linear combinations of the original phases $\{ \phi_k\}$. By Theorem~\ref{thm:qsp}, the term in parentheses is a matrix of related polynomials evaluated at $a$. The above matrix element thus evaluates to a polynomial $\text{Poly}(a)$ of degree at most $d$ (which can be seen by counting the instances of $U$ and $U^\dag$ in Eq.~(\ref{eq:thm_aa_qsp})), completing the proof of Theorem~\ref{thm:aa}.
    
    Theorem~\ref{thm:aa} takes on the meaning of performing ``oblivious amplitude amplification'' when the phases $\{\phi_k\}$ are chosen such that the polynomial constructed approximates the sign function for small values of $a$. The technical aspects of generating the proper polynomials are further discussed in Sec. \ref{sec:Search} and Appendix \ref{sec:appx_fpsearch}.
    
    Grover's celebrated quantum search algorithm~\cite{grover1996fast}, which is similar in character to the above, can easily be derived from the construction of amplitude amplification. In the search problem, some computational basis state $|A_0\>$ is unknown, but an oracle is given that implements
        \be
            A_\pi = e^{i\pi |A_0\>\<A_0|}\,,
        \ee
    and the goal is to create a quantum state close to $|A_0\>$ using as few queries to the oracle as possible.  The search algorithm solves this problem by choosing $U = H^{\otimes n}$ (Hadamards on all the qubits in the search space), and starting with $|B_0\> = |0\>$. Note that since $|A_0\>$ is a computational basis state, $\<A_0 |U|B_0\> = 2^{-n/2} = 1/\sqrt{N}$ because $U|B_0\> = |\psi_0\>$ is an equal superposition over all $N$ basis states in the search space. This means that the amplitude amplification algorithm can be applied, with the specific case of $a=1/\sqrt{N}$ being known, meaning we choose all the $\phi_k = \pi$ for all $k$. This choice implies that 
        \bea
            U B_\phi U^\dagger = e^{i\pi H^{\otimes n} |0\>\<0| H^{\otimes n}} = I - 2|\psi_0\>\<\psi_0|\,,
        \eea
    which can be recognized as Grover's ``inversion about mean'' iterate. This choice also produces an oscillating polynomial that monotonically increases the matrix element up to $d \approx \lceil \pi/(2\sin^{-1} a) \rceil$ steps~\cite{PhysRevLett.113.210501}. This may be seen by noting that each application of the signal rotation operator rotates through an angle $2\sin^{-1} a$, and the target state is at most an angle $\pi$ away, similar to the argument used to derive the query complexity of Grover's algorithm in ~\cite{nielsen2010quantum}. The amplitude amplification algorithm thus gives Grover's quantum search algorithm in this limit, and the number of oracle calls required is $d \approx \pi/(2 \sin^{-1} a) = \pi/(2\sin^{-1}(1/\sqrt{N})) = O(\sqrt{N})$, the known performance of Grover's algorithm. Once more, a similar result will be discussed in Section~\ref{sec:Search}, where the flow from qubitization to QSVT to fixed point amplitude amplification will be applied more generally.

\subsection{Quantum Eigenvalue Transforms}\label{sec:eigenvalue_transforms}

Theorem~\ref{thm:aa} is known as amplitude amplification because it accomplishes a polynomial transform of one specific amplitude, namely the matrix element of $U$ at $|A_0\>\<B_0|$.  However, this polynomial transform can actually be performed over {\em an entire vector space}, not just a one-dimensional matrix element.  In particular, we now show that this technology can be used to polynomially transform all the {\em eigenvalues} of a Hamiltonian $\mathcal{H}$ that has been embedded into a block of a unitary matrix $U$.

Specifically, suppose we have the unitary
    \bea
        U = \kbordermatrix{\mbox{} &0 & 1 \\
          0 & \mathcal{H}     & \cdot \\ 
          1 & \cdot & \cdot 
        }
        \,,
        \label{eq:block_u_h}
    \eea
where $\mathcal{H}$ is some $N$$\times$$N$ (possibly very large) Hamiltonian operator, located in the upper left block of $U$, labeled by an index qubit being in the state $|0\>$ (and it is said that $\mathcal{H}$ has thus been ``qubitized''). We have included the indices $0$ and $1$ adjacent to the matrix representation of $U$ to schematically indicate how $\mathcal{H}$ is encoded in $U$. At the cost of some generality let us make some assumptions for simplicity of exposition.  In particular, assume the operator norm $\|\mathcal{H}\|$ is sufficiently small that this block embedding can be achieved, i.e. $\|\mathcal{H}\| \leq 1$ (if not, then one can instead embed some rescaled version of the Hamiltonian, $\mathcal{H}/\alpha$, but for now we will neglect this case for the sake of expository clarity). In particular, suppose that the eigenvectors and eigenvalues of $\mathcal{H}$ are given as
    \begin{equation}
        \mathcal{H} = \sum_\lambda \lambda |\lambda\>\<\lambda|
    \,.
    \end{equation}
Then, specializing to a specific block encoding for clarity, the missing blocks of $U$ may be completed as
    \bea
         U = \begin{bmatrix}
                \mathcal{H} & \sqrt{I-\mathcal{H}^2} \\
                \sqrt{I-\mathcal{H}^2} & -\mathcal{H} 
             \end{bmatrix}
         \,,
         \label{eq:U_block_encoding}
    \eea
where
    \begin{equation}
        \sqrt{I-\mathcal{H}^2} = \sum_\lambda \sqrt{1-\lambda^2} |\lambda\>\<\lambda|
        \,,
    \end{equation}
and it can be seen by inspection that $U^\dagger U = I$ as required for the unitarity of $U$, as long as the eigenvalues $\lambda$ are of reasonable scale. While a general block encoding need not take this specialized form, this choice of encoding will be sufficient for our illustrative purposes. Moreover, the treatment of general block encodings is presented in~\cite{Gily_n_2019, Low_2019}, wherein it is shown that a general block encoding takes a form similar to Eq.~(\ref{eq:U_block_encoding}) in a special basis related to the eigenbasis of $\mathcal{H}$.

Our specialized choice of block encoding means that $U$ may be expressed as a sum of two tensor products
    \be
    	U = Z\otimes \mathcal{H} + X \otimes \sqrt{I-\mathcal{H}^2}
    \,,
    \ee
and therefore acts as
    \bea
        U |0\>|\lambda\> &=& \lambda |0\>|\lambda\> + \sqrt{1-\lambda^2} |1\>|\lambda\> \\
        U |1\>|\lambda\> &=& -\lambda |1\>|\lambda\> + \sqrt{1-\lambda^2} |0\>|\lambda\>
        \,,
    \eea
which indicates that $U$ contains a Bloch sphere (i.e. a qubit basis) for each eigenspace\footnote{Note that even for degenerate eigenspaces the linearity of the QSP sequence ensures qubitization.} of $\mathcal{H}$ corresponding to a certain eigenvalue. In particular, $U$ may be expressed as a sum over $N$ separate Bloch spheres:
    \begin{align}\label{eq:uham_bloch_spheres}
        \begin{aligned}
            U &= \sum_\lambda \mattwocb{\lambda}{\sqrt{1-\lambda^2}}{\sqrt{1-\lambda^2}}{-\lambda} \otimes |\lambda\>\<\lambda| 
            \\
            &= \sum_\lambda \lbL{ \sqrt{1-\lambda^2} X + \lambda Z }\rb  \otimes |\lambda\>\<\lambda| 
            \\
            &=: \sum_\lambda R(\lambda)  \otimes |\lambda\>\<\lambda| = \bigoplus_\lambda R(\lambda) ,
        \end{aligned}
    \end{align}
where $R(\lambda)$ is defined as the operator in square
brackets in the penultimate line above, and may be interpreted as a reflection and rotation about the Bloch sphere's $y$-axis, exactly the same as we found for the $R(a)$ operator that appeared in the amplitude amplification construction as per Eq.~(\ref{eq:ampamp_rot}). We thus have a form for $U$ which parallels that of the amplitude amplification scenario. However, unlike amplitude amplification, we do not have one-dimensional phase rotation operators $A_\phi$ and $B_\phi$, because we have $N$ Bloch spheres, and not just two!

Nevertheless, there are still distinct vector spaces in which the input and output of $\mathcal{H}$ exist: these are, respectively, the column space and row space of the matrix $\mathcal{H}$, within $U$. In the scenario of Eq.~(\ref{eq:block_u_h}), these vector spaces are defined by the projector $\Pi := |0\>\<0|$ acting on the auxiliary qubit. Generalizing the way amplitude amplification employs the phase shift $A_\phi$ of
Eq.~(\ref{eq:aa_aphi}) acting on a single element $|A_0\>\<A_0|$, we may now define a projector controlled phase shift operation $\Pi_\phi$:
    \bea
    	\Pi_\phi := e^{i 2 \phi \Pi}
        \,,
    \eea
which imparts a phase of $e^{i2\phi}$ to the {\em entire subspace} determined by the projector $\Pi$. Note that if we want to be more precise, we may instead define this operation as
    \bea
    	\Pi_\phi := e^{i\phi (2\Pi - I)}
    \eea
which is a proper unitary transform and acts as a $z$-rotation, like $S(\phi)$ from Eq.~(\ref{eq:qsp_rotation}), but these two definitions differ only by a global phase which may be neglected. From a quantum circuit standpoint, $\Pi_\phi$ may be realized by employing two instances of projector controlled-{\sc not} gates (which we will refer to as $\Pi$-controlled-{\sc not}, or $\text{C}_{\Pi}\text{NOT}$ for short)
around a single qubit $z$-rotation by angle $\phi$:

\begin{figure}[htbp]
    \centering
    \includegraphics[width=0.8\columnwidth]{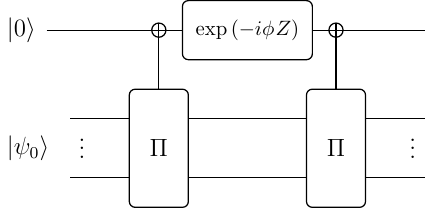}
    \caption{The circuit used to realize the projector controlled phase shift $\Pi_\phi$, also termed the ``phased iterate'', in analogy to the bare $\exp(i\phi Z)$ operation used in QSP, used as the signal processing rotation operator. We note that this circuit is valid for an arbitrary projector $\Pi$, although it simplifies to a single rotation for the simple projector $\Pi = |0\rangle \langle 0|$ that we've used in this section.}
    \label{fig:projector-phase-shift-quantum-circuit}
\end{figure}

\noindent
The main relevant observation is that on the subspace of each of the $N$ Bloch spheres of Eq.~(\ref{eq:uham_bloch_spheres}), $\Pi_\phi$ acts as a $z$-axis rotation,
    \bea
        \Pi_\phi = \sum_\lambda e^{i\phi Z} \otimes |\lambda\>\<\lambda| = \bigoplus_\lambda e^{i\phi Z}
    \eea
with $|\lambda\>$ as an eigenvector.

This picture of $N$ separate Bloch spheres evolving under $z$- and $y$-rotations provides the intuitive basis for the following theorem:
    \begin{theorem}[Eigenvalue transformation]
    Given a block encoding of Hamiltonian $\mathcal{H} = \sum_\lambda \lambda |\lambda\>\<\lambda|$ in a unitary matrix $U$,
        \bea
            U = \kbordermatrix{\mbox{} &\Pi &  \\
            \Pi & \mathcal{H}     & \cdot \\ 
            & \cdot & \cdot 
            }
            \,,
        \eea
    with the location of $\mathcal{H}$ determined by projector $\Pi$, and given the ability to perform $\Pi$-controlled NOT operations to realize projector controlled phase shift operations $\Pi_\phi$, then, for even $d$,
        \begin{align}\label{eq:thm_eigenvalue_transform}
            \begin{aligned}
                U_{\vec\phi} &= 
                \left[ \prod_{k=1}^{d/2} \Pi_{\phi_{2k-1}} U^\dagger \Pi_{\phi_{2k}}  U  \right] 
            &= \kbordermatrix{\mbox{} &\Pi &  \\
            \Pi & {\rm Poly}(\mathcal{H})     & \cdot \\ 
                & \cdot & \cdot 
          }
            \end{aligned}
        \end{align}
    where
        \be
            {\rm Poly}(\mathcal{H}) = \sum_\lambda {\rm Poly}(\lambda) |\lambda\>\<\lambda|
        \ee
    is a polynomial transform of the eigenvalues of $\mathcal{H}$. The polynomial is of degree at most $d$, and obeys the conditions on $P$ from Theorem~\ref{thm:qsp}. 
      
    Similarly, for odd $d$,
        \begin{align}\label{eq:thm_eigenvalue_transform_odd}
            \begin{aligned}
                U_{\vec\phi} &= 
                \Pi_{\phi_{1}}  U \left[ \prod_{k=1}^{(d-1)/2} \Pi_{\phi_{2k}} U^\dagger \Pi_{\phi_{2k+1}}  U  \right] 
                \\ &= \kbordermatrix{\mbox{} &\Pi &  \\
                \Pi & {\rm Poly}(\mathcal{H})     & \cdot \\ 
                & \cdot & \cdot 
                }
            \end{aligned}
        \end{align}
    where ${\rm Poly}(\mathcal{H})$ has an analogous interpretation. 
    \label{thm:eigenvalue_transform}
    \end{theorem}

Why does this work?  For the same reason that Theorem~{\ref{thm:aa}} does! To see this, let us use the specific encoding of Eq.~(\ref{eq:U_block_encoding}) to rewrite Eq.~(\ref{eq:thm_eigenvalue_transform}) as an action on $N$ separate Bloch spheres (the odd $d$ case in Eq.~(\ref{eq:thm_eigenvalue_transform_odd}) being analogous):
    \bea
        U_{\vec\phi} &=& \bigoplus_\lambda \left[ \prod_{k=1}^{d/2} e^{i\phi_{2k-1}Z} R(\lambda) e^{i\phi_{2k}Z} R(\lambda) \right],
        \label{eq:sigproc_tran_a}
    \eea
recognizing that we have chosen conventions for $R(\lambda)$ in Eq.~(\ref{eq:uham_bloch_spheres}) such that $R^\dagger(\lambda) = R(\lambda)$. This allows us to relabel the sum and the phases to put Eq.~(\ref{eq:sigproc_tran_a}) into a standard form of quantum signal processing (similar to how we obtained Eq.~(\ref{eq:aa_qsp_seq})):
    \bea
        U_{\vec\phi} &=& \bigoplus_\lambda \left[ e^{i\phi'_0 Z} \prod_{k=1}^{d} W(\lambda) e^{i\phi'_{k}Z} \right] 
        \label{eq:sigproc_tran_b}
        \,.
    \eea
By Theorem~\ref{thm:qsp}, the term in brackets is a matrix of polynomials in $\lambda$, verifying Theorem~\ref{thm:eigenvalue_transform}. Finally, although we have specialized to the specific block encoding of Eq.~(\ref{eq:U_block_encoding}), the proof for more general block encodings is contained in~\cite{Gily_n_2019, Low_2019}.

One immediately apparent utility of this eigenvalue transformation theorem is its usefulness in filtering eigenvalues. For example, we saw previously in Section~\ref{sec:qsp} that the BB1 sequence of phases can be used to increase sensitivity to the signal. In this eigenvalue transformation scenario, the signal is the eigenvalue $\lambda$, and a sequence of QSP phases can be used to selectively filter a range of eigenvalues, e.g. those below a threshold of $\lambda \approx 2\pi/3$ in the case of BB1 (for now this is an imprecise statement, but a similar question is a major concern of Section~\ref{sec:Threshold}). Measuring the projector $\Pi$ will show that the QSP sequence flips the index qubit with higher probability if the Hamiltonian has an eigenvalue larger than this threshold, versus when all the eigenvalues of $\mathcal{H}$ are below the threshold. This result also assumes the ability to prepare a state which has some overlap with the relevant eigenstates of $\mathcal{H}$, as input into $U$, but the specific amplitudes of the overlap are not crucial, because amplitude amplification can be employed to boost the signal.

\subsection{Quantum Singular Value Transforms}\label{sec:qsvt}

Theorem~\ref{thm:eigenvalue_transform} is known as the eigenvalue transformation because it accomplishes a polynomial transform of the eigenvalues of a matrix embedded within a larger unitary matrix. However, in general the embedded matrix need not have a well defined set of eigenvalues; for example, instead of being a square matrix, it could be rectangular. Remarkably, the same idea of transforming the embedded matrix using quantum signal processing can still apply. Specifically, we now show that QSP sequences can be used to polynomially transform all the {\em singular values} of a (possibly non-square) matrix $A$ which has been encoded into a block of a unitary matrix $U$.

Such a general matrix $A$ can be decomposed as 
    \bea
    	A = W_\Sigma\Sigma V^\dag_\Sigma
    \,,
    \eea
where $W_\Sigma$ and $V_\Sigma$ are unitary and $\Sigma$ is diagonal and contains along its diagonal the set of non-negative real numbers $\{ \sigma_k\}$, known as the singular values of $A$, of which there are $r=\text{rank}(A)$ nonzero values.  All matrices have such a {\em singular value decomposition} (SVD).

As $W_\Sigma$ and $V_\Sigma$ are unitary, their columns form orthonormal bases, which we denote by $\{|w_k\>\}$ and $\{|v_k\>\}$, respectively. $\{|w_k\>\}$ are the left singular vectors, which span the left singular vector space, and $\{|v_k\>\}$ are the right singular vectors, which span the right singular vector space. Using this notation, we may conveniently rewrite the singular value decomposition of $A$ in a form analogous to the eigenvalue decomposition: 
    \bea
    	A = \sum_{k=1}^r  \sigma_k |w_k\>\<v_k|
    \,.
    \eea

Now, suppose that we are given a unitary matrix $U$ such that $A$ is encoded in a block of $U$, i.e.
    \bea
        U = \kbordermatrix{\mbox{} & \Pi &  \\
          \tilde{\Pi} & A     & \cdot \\ 
                      & \cdot & \cdot 
        }
        \,,
        \label{eq:block_u_A}
    \eea
where $\tilde\Pi := \sum_k |w_k\>\<w_k|$ and $\Pi := \sum_k |v_k\>\<v_k|$ are projectors which locate $A$ within $U$, such that $A={\tilde{\Pi}} U \Pi$. We have again included $\Pi$ and $\tilde{\Pi}$ indices adjacent to the matrix representation of $U$ to schematically indicate how $A$ is encoded in $U$: $\Pi$ selects the columns and $\tilde{\Pi}$ the rows in which $A$ is encoded. Moreover, the projectors $\Pi$ and $\tilde{\Pi}$ also identify the left and right singular vector spaces, respectively. 

For pedagogical simplicity, let us assume for now that $A$ is a square matrix (this assumption is dropped in the next section). Again specializing to a specific block encoding, we may complete the missing blocks of the unitary $U$ by writing
    \bea
    \label{eq:A_block_encoding}
        U = \kbordermatrix{\mbox{} & 0 & 1 \\
          0 & A                 & \sqrt{I-A^2} \\
          1 & \sqrt{I-A^2}  & -A
          }
        \,,
    \eea
where the $0$ and $1$ are index qubit labels for the block matrices within $U$, and where $\sqrt{I-A^2}$ is formally defined in terms of the SVD of $A$, as
    \bea
    	\sqrt{I-A^2} := \sum_{k} \sqrt{1-\sigma_k^2} |w_k\>\<v_k|
        \,.
    \eea
It can then be verified straightforwardly that $U^\dagger U = I$, as long as the singular values $\{\sigma_k\}$ are less than or equal to $1$ (which may always be achieved by rescaling $A$ to some $A/\alpha$). Again, while a general block encoding need not take this specialized form, this choice will be sufficient for our illustrative purposes. And moreover, the treatment of general block encodings is presented in~\cite{Gily_n_2019}, wherein it is shown that a general block encoding takes a form similar to Eq.~(\ref{eq:A_block_encoding}) in special basis related to the left and right singular vectors of $A$. 

Just as with the analysis of the block-encoded Hamiltonian $\mathcal{H}$, there are multiple Bloch spheres in $U$. Specifically, note that
    \bea
        U |0\>|v_k\> &=&  \sigma_k |0\>|w_k\> + \sqrt{1-\sigma_k^2} |1\>|w_k\> \\
        U |1\>|v_k\> &=& -\sigma_k |1\>|w_k\> + \sqrt{1-\sigma_k^2} |0\>|w_k\>
        \,,
    \eea
so that $U$ may be expressed as a direct sum over a number of separate Bloch spheres equal to the rank of $A$:
    \begin{align}\label{eq:svd_bloch_spheres}
        \begin{aligned}
            U &= \sum_{k} \mattwocb{\sigma_k}{\sqrt{1-\sigma_k^2}}{\sqrt{1-\sigma_k^2}}{-\sigma_k} \otimes |w_k\>\<v_k| 
            \\
            &= \sum_{k} \lbL{ \sqrt{1-\sigma_k^2} X + \sigma_k Z }\rb  \otimes |w_k\>\<v_k| 
            \\
            &=: \sum_{k}  R(\sigma_k)  \otimes |w_k\>\<v_k| 
            \,,
        \end{aligned}
    \end{align}
where we have defined $R(\sigma_k)$ analogous to $R(\lambda)$ in Eq.~(\ref{eq:uham_bloch_spheres}). We now have a form for $U$ which directly parallels that of the eigenvalue transform scenario, and for exactly the same reasons, the following theorem holds: 

    \begin{theorem}[Singular value transformation]\label{thm:singular_value_transform}
        Given a block encoding of a matrix $A = \sum_k \sigma_k |w_k\>\<v_k|$ in a unitary matrix $U$,
            \bea
                U = \kbordermatrix{\mbox{} & \Pi &  \\
                  \tilde{\Pi} & A     & \cdot \\ 
                              & \cdot & \cdot 
                }
                \,,
            \eea
        with the location of $A$ determined by projectors $\Pi$ and $\tilde\Pi$, and given the ability to perform $\Pi$- and $\tilde\Pi$-controlled {\sc not} operations to realize projector controlled phase shift operations $\Pi_\phi$ and  $\tilde{\Pi}_\phi$ (defined as in Sec. \ref{sec:eigenvalue_transforms}), then, for odd $d$,
            \begin{align}\label{eq:thm_singular_value_transform_d_odd}
                \begin{aligned}
                      U_{\vec\phi} &= \tilde{\Pi}_{\phi_{1}}  U \left[ \prod_{k=1}^{(d-1)/2} \Pi_{\phi_{2k}} U^\dagger \tilde{\Pi}_{\phi_{2k+1}}  U  \right] \\ &= 
              \kbordermatrix{\mbox{} &\Pi &  \\
                \tilde{\Pi} & {\rm Poly}^{({\rm SV})}(A)     & \cdot \\ 
                    & \cdot & \cdot 
              },
                \end{aligned}
            \end{align}
        where $ {\rm Poly}^{({\rm SV})}(A)$ is defined for an odd polynomial as
        \be
            {\rm Poly}^{({\rm SV})}(A) := \sum_{k} {\rm Poly}(\sigma_k) |w_k\>\<v_k|,
        \ee
        which applies a polynomial transform to the singular values of $A$.  The polynomial is of degree at most $d$ and obeys the conditions of $P$ from Theorem~\ref{thm:qsp}.

        Similarly, for $d$ even,
            \begin{align}\label{eq:thm_singular_value_transform}
                \begin{aligned}
                      U_{\vec\phi} &= \left[ \prod_{k=1}^{d/2} \Pi_{\phi_{2k-1}} U^\dagger \tilde{\Pi}_{\phi_{2k}}  U  \right] \\ &= 
              \kbordermatrix{\mbox{} &\Pi &  \\
                \Pi & {\rm Poly}^{({\rm SV})}(A)     & \cdot \\ 
                    & \cdot & \cdot 
              },
                \end{aligned}
            \end{align}
        where $ {\rm Poly}^{({\rm SV})}(A)$ is defined for an even polynomial as
            \be
                {\rm Poly}^{({\rm SV})}(A) := \sum_{k} {\rm Poly}(\sigma_k) |v_k\>\<v_k|.
            \ee
        which is also a polynomial transform of the singular values of $A$, but with the modification that the input and output spaces are both the right singular vector space, spanned by $\{ |v_k \rangle \}$. Analogously, the polynomial is of degree at most $d$ and obeys the conditions of $P$ from Theorem~\ref{thm:qsp}.
        
    \end{theorem}

The main difference between this theorem and the eigenvalue transform of Theorem~\ref{thm:eigenvalue_transform} is that here, the Bloch sphere transformations of $U$ {\em also switch between the $\{|v_k\rangle\}$ and the $\{|w_k\rangle \}$ bases}, similar to how the sequence in Theorem~\ref{thm:aa} (Eq.~(\ref{eq:thm_aa_qsp})) flips between bases. Thus, we should now carefully keep track of which vector space the system is in at each stage of operation. In particular, in each of the terms in square brackets in Eqs.~(\ref{eq:thm_singular_value_transform}) and~(\ref{eq:thm_singular_value_transform_d_odd}), the $U$ on the right moves the system from the $\{|v_k\rangle \}$ basis into the $\{|w_k \rangle\}$ basis, and $\tilde{\Pi}_{\phi}$ then rotates the system around the $z$-axis of each left singular vector space's Bloch sphere. Similarly, the $U^\dagger$ on the left then moves the system back from the $\{| w_k \rangle \}$ basis into the $\{|v_k \rangle \}$ basis, and finally $\tilde{\Pi}_{\phi}$ rotates the system around the $z$-axis of each right singular vector space's Bloch sphere. Therefore, the odd $d$ sequence starts in the right singular vector space and ends in the left singular vector space, such that ${\rm Poly}^{({\rm SV})}(A)$ is accessed with $\Pi$ and $\tilde{\Pi}$, whereas the even $d$ sequence starts and ends in the right singular vector space, such that ${\rm Poly}^{({\rm SV})}(A)$ is accessed with just $\Pi$.

With this crucial modification, Theorem~\ref{thm:singular_value_transform} may be verified by following the same logic as that of Theorem~\ref{thm:eigenvalue_transform}, while the proof for general block encodings is presented in~\cite{Gily_n_2019}. Theorem~\ref{thm:singular_value_transform} is the essence of the quantum singular value transform (QSVT) algorithm, and we elaborate in the next sections upon what can be accomplished with it.

\subsection{Block Encodings}
\label{sec:Block_Encodings}

A key idea behind the generalization of QSP from single-qubit dynamics, to multi-qubit transforms is the use of a block-encoded operator, as we have just seen in the last two subsections (and as is discussed in the following sections).  While the embedding of one matrix inside a larger one is well known in mathematics as a {\em dilation}\cite{horn_johnson_1991}, as used here, block encoded operators are a quantum concept, and their construction comes with some caveats. We elaborate on these constructions in this subsection.

The starting point for many applications of QSP and QSVT is the availability of the desired signal as a linear operator, a matrix $A$, encoded as a block inside a larger unitary matrix $U$, as described in Eq.~(\ref{eq:block_u_A}). What is such a block encoded operator, physically? Some limited intuition comes from the case when $A$ is a unitary matrix, in which case $U$ is simply a controlled-$A$ operator, as depicted in Figure~\ref{fig:block_encoded_A}.

\begin{figure}[htbp]
    \centering
    \includegraphics[width=0.7\columnwidth]{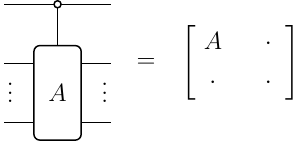}
    \caption{A simple block encoding quantum circuit, showing how such a block encoding of a unitary $A$ is notated. Observe that the empty circle on the control qubit (top line) indicates a $|0\rangle$ control, whereas a filled-in circle would correspond to a $|1\rangle$ control.}
    \label{fig:block_encoded_A}
\end{figure}

\begin{figure}[htbp]
    \centering
    \includegraphics[width=0.88\columnwidth]{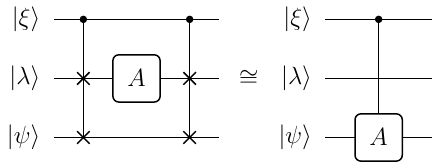}
    \caption{Construction of $\lvert 1\rangle$-controlled-$A$ (for $A$ a single qubit operator) given access to $\lvert \lambda \rangle$ a $(+1)$-eigenvector of $A$, $\lvert \xi \rangle$ the control qubit, and $\lvert \psi \rangle$ the target state. Note that $A$ may in general act on a large Hilbert space, in which case the controlled swap operator is suitably generalized. The swaps can be made $\lvert 0 \rangle$-controlled to emulate Figure~\ref{fig:block_encoded_A}. We note that this is a similar construction to the circuits employed in \cite{Huggins_2020} for matrix element measurement.}
    \label{fig:swap_controlled_A}
\end{figure}

When $A$ is a large operator, it can be challenging to envision how a controlled-$A$ operation may be physically realized. In particular, given $A$, how to construct a controlled-$A$ operation is generally not possible without further information. For example, if a $+1$ eigenstate $|\lambda\>$ of $A$ is provided, then a controlled-$A$ quantum circuit can be constructed using $|\lambda\>$ together with two controlled-SWAP gates; see Figure~\ref{fig:swap_controlled_A}. However, in general, much more work has to be done in achieving a block encoding, generally in ways specific to the relevant physical system \cite{Gily_n_2019, Huggins_2020}.


Related to the issue of the block encoding is accessing the desired polynomial transform of the encoded block.  As described in Section~\ref{sec:qsp}, after performing a QSP/QSVT sequence, a measurement may be done in the signal basis to determine if the desired polynomial transform was realized.  In particular, in many of the algorithms, we wish to apply the transformed operator, ${\rm Poly}^{({\rm SV})}(A)$, to some state $|\psi\rangle$, which may be done using the block encoding, $U_{\vec{\phi}}$. However, this computation will only be successful if we apply to $|\psi\rangle$ the correct block (the one that encodes ${\rm Poly}^{({\rm SV})}(A)$), so we must project the final state into being in this block by performing a projective measurement with the projectors $\Pi$ and $\tilde{\Pi}$, which are assumed accessible. Consequently, the relevant projector for single-auxiliary-qubit QSVT is often $|+\rangle \langle + \rvert \otimes \Pi$, which both (1) isolates the real part of the transforming polynomial $P$ discussed in Section~\ref{sec:qsp}, and (2) determines if we have applied the desired block of the overall unitary. Since this projective measurement is a one time projection that is done at the end of the algorithm, its cost is not large, and moreover, the probability of the projective measurement yielding the desired state can be amplified using amplitude amplification or classical repetition.  Depending on the construction, this probability can also be either the desired output, or an algorithmically useful output.  For more about the signal basis and QSP conventions, see Appendix~\ref{sec:QSPConventions}.


\section{Search by QSVT} \label{sec:Search}
    
    Now that we have an overall perspective on the basic theorems of quantum signal processing and quantum singular value transformations, let us return to explicitly constructing the three ``primordial'' quantum algorithms using QSVT. Specifically, for each algorithm we indicate (a) the signal rotation operator, and (b) the appropriate QSP polynomial. This is not just a problem of reconstruction, but will in certain cases lead to slight improvements in algorithmic performance. We begin with the search algorithm.
    
    A ubiquitous problem that admits a quantum speedup is unstructured search, the goal of which is to determine a single marked element $m \in \{0, 1, \cdots, N - 1\} =: [N]$ from an unstructured database of size $N$. While the search problem can be solved in $\mathcal{O}(N)$ time classically, Grover's quantum algorithm requires only $\mathcal{O}(\sqrt{N})$ time, providing a provably optimal quadratic speedup. However, Grover's algorithm in its traditionally stated form suffers various shortcomings, including divergence from the proper solution if its underlying iterate is repeated too many times. Fortunately, by the parallels between this problem and amplitude amplification as discussed in Section~\ref{sec:AA_and_search}, it is straightforward to recover an improved version of unstructured search using QSVT.
    
    In quantum search, one is given access to a unitary operation $U$ whose application flags some marked element by a phase flip, i.e.,
        \begin{equation} \label{eq:oracle_action}
            U|j\rangle = (-1)^{\delta_{jm}}|j\rangle = \begin{cases}
                |j\rangle & |j\rangle \neq |m\rangle \\
                -|j\rangle & |j\rangle = |m\rangle
            \end{cases},
        \end{equation}
    where $m \in \{0, 1, \cdots, N - 1\}$ is the (single) marked element. 
    
    As is conventional in quantum search, we will take our initial state to be the uniform superposition over all $N$ states, which may be expressed as
    \begin{equation}\label{eq:initial_state}
            \lvert \psi_0 \rangle = \frac{1}{\sqrt{N}} \sum_j^N |j\rangle = \sqrt{\frac{N - 1}{N}}\,\lvert m^\perp\rangle + \frac{1}{\sqrt{N}}\,\lvert m \rangle,
    \end{equation}
    and is easily prepared by applying $n$ Hadamards to $|0\rangle^{\otimes n}$. With this choice of initialization, our goal is to map $|\psi_0\rangle$ to the marked state $|m\rangle$.
    
    In this context, how might we solve the search problem with QSVT? The standard method, which includes the construction of the Grover iterate, will not be necessary. Instead, we rephrase the problem. In the most general setting, we are given access to some initial state $|\psi_0\rangle$ such that its projection onto the desired final state $|m\rangle$ is nonzero. Denoting this projection by $\Pi$, we may express this condition as
        \begin{equation}
            \Pi\, V |\psi_0 \rangle = a |m\rangle,
        \end{equation}
    for some known unitary $V$ and $|a| > 0$. In our case, as illustrated in Eq.~(\ref{eq:initial_state}), we can simply take $V$ as the identity, and $a = 1/\sqrt{N}$ by our choice of initial state. In general for amplitude amplification problems, $V$ need not be trivial, and is often viewed as a preparation oracle.
    
    The search problem as framed above is precisely the statement of amplitude amplification as in Sec.\ref{sec:AA_and_search}, as we desire a circuit $U^\prime$ (which will depend on $U$) such that
        \begin{equation} \label{eq:amplitude_amplification}
            \Big\lVert\,
            |m\rangle\langle m| U^\prime |\psi_0 \rangle\langle \psi_0| 
            - 
            |m \rangle\langle\psi_0|
            \,\Big\rVert 
            < 
            \epsilon,
        \end{equation}
    for some small $\epsilon > 0$. For notational convenience we will refer to $|\psi_0\rangle\langle \psi_0|$ as $\Pi^\prime$. Note that this induces a block encoding of the scalar $a$, that is
        \begin{equation} \label{eq:search_block_encoding}
            \Pi\, V\, \Pi^\prime = \Pi\, \Pi^\prime = a |m\rangle\langle\psi_0|.
        \end{equation}
    We now have a problem goal stated in Eq.~(\ref{eq:amplitude_amplification}), and a block encoding of a single singular value $a = 1/\sqrt{N}$ stated in Eq.~(\ref{eq:search_block_encoding}). What remains to solve unstructured search with QSVT is a choice of QSVT polynomial, and an analysis of the minimum required QSVT sequence length.
    
    Fortunately, given that there is only one relevant singular value, $a$, and one relevant left (and right) singular vector, the entire problem is reduced to mapping $a$ as close to $1$ in magnitude as possible, as under this condition the circuit will induce, with high probability, the desired block-encoded map $|m\rangle\langle \psi_0|$. Luckily, there already exists a theorem which defines the conditions under which such transformations are possible: the major theorem of QSVT!
    
    The requirements to perform QSVT as desired include the application of $V$, $V^\dag$, $C_{\Pi}\text{NOT}$, $C_{\Pi^\prime}\text{NOT}$, and $e^{-i\phi Z}$ gates. Note that here we have been allowed to choose $V = I$, as all Grover oracle information is encoded in the projector-controlled-NOTs. Moreover, these projector-controlled operators are easy to construct, $C_{\Pi^\prime}\text{NOT}$ because $\Pi^\prime$ depends on a fixed, known state. Similarly, $C_{\Pi}\text{NOT}$ may be crafted by first creating a controlled-$U$, which is easily achieved, and conjugating the control qubit by Hadamard gates, the result of which is a $C_{\Pi}\text{NOT}$, where the control qubit now becomes the target of the projector-controlled-NOT.
    
    Having now achieved a suitable block encoding and the necessary conditions for QSVT, we need only choose an appropriate QSVT polynomial. As our goal is to map $a$ to a value close to $1$, a sensible choice for such a polynomial can be derived if we start with the simple sign function $\Theta(x-c)$:
        \begin{equation} \label{eq:sign_function}
            \Theta(x-c) =
            \begin{cases}
                -1 & x < c\\
                0 & x = c\\
                1 & x > c,
            \end{cases}
        \end{equation}
    As discussed in~\cite{low2017quantum, Gily_n_2019}, one can estimate the sign function with arbitrary precision by finding a polynomial approximation to $\text{erf}(k[x-c])$ for large enough $k$. In particular, one can efficiently compute a degree $\mathcal{O}\left(\frac{1}{\Delta}\log(\frac{1}{\epsilon})\right)$ odd polynomial $P^{\Theta}_{\epsilon,\Delta}(x-c)$, where $\epsilon \in \big( 0, \sqrt{2/e\pi} \big]$, such that
        \begin{enumerate}
            \item $|P^{\Theta}_{\epsilon,\Delta}(x-c)|\leq 1 \text{ for } x \in [-1,1]$.
            \item $|\Theta(x-c) - P^{\Theta}_{\epsilon,\Delta}(x-c)| \leq \epsilon$ \\ for $x \in \big[ -1,1 \big] \backslash \left(c-\frac{\Delta}{2}, \ c+\frac{\Delta}{2}\right)$.
        \end{enumerate}
    That is, $P^{\Theta}_{\epsilon,\Delta}(x-c)$ is bounded in magnitude by $1$ and $\epsilon$-approximates the sign function outside of the region $\left( c - \frac{\Delta}{2}, \ c + \frac{\Delta}{2} \right)$. On the other hand, within the range $\left( c - \frac{\Delta}{2}, \ c + \frac{\Delta}{2} \right)$, where $\Theta(x-c)$ suffers a discontinuity, $P^{\Theta}_{\epsilon,\Delta}(x-c)$ is not in general a good approximation to the sign function. In this region, $P^{\Theta}_{\epsilon,\Delta}(x-c)$ is actually an $\epsilon$-approximation to $\text{erf}\big(k[x-c] \big)$ where $k=\frac{\sqrt{2}}{\Delta} \log^{1/2}\left(2/[\pi \epsilon^2]\right)$ (see~\cite{low2017quantum} for the explicit construction of $P^{\Theta}_{\epsilon,\Delta}(x-c)$). For visual intuition, we illustrate this polynomial approximation below in Figure~\ref{fig:SignFunctionApproximation}.
    
    \begin{figure}[htpb]
        \centering
        \includegraphics[width=1.0\columnwidth]{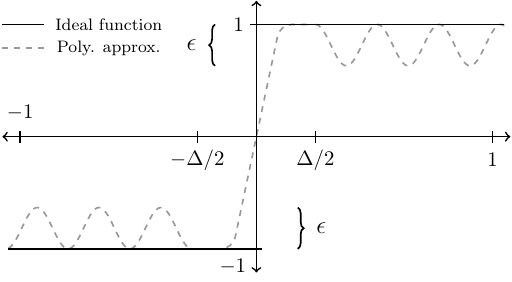}
        \caption{An illustration of $\Theta(x)$ (solid) and its polynomial approximation, $P^{\Theta}_{\epsilon,\Delta}(x)$ (dashed). Note how $P^{\Theta}_{\epsilon,\Delta}(x)$ stays within $\epsilon$ of the sign function outside of the region $\left(\frac{\Delta}{2},\ \frac{\Delta}{2}\right)$, while deviating from $\Theta(x)$ within this region. The shifted version of this function is constructed much the same way.}
        \label{fig:SignFunctionApproximation}
    \end{figure}
    
    Note that the polynomial $P^{\Theta}_{\epsilon,\Delta}(x)$ (i.e., with $c = 0$) has odd parity and is bounded in magnitude by $1$, and so it may be implemented via QSVT (i.e. it obeys the conditions on Poly($a$)$= \langle m | U_{\vec{\phi}} | \psi_0\rangle$ discussed in Section~\ref{sec:qsp} and Appendix \ref{sec:QSPConventions}). We will use this scheme to polynomially approximate the sign function, e.g., $\Theta^{(\text{SV})}(V) \approx (P^{\Theta}_{\epsilon,\Delta})^{(\text{SV})}(V)$, and subsequently apply it to the initial uniform superposition over $N$ quantum states. Then, upon applying the pre-determined QSVT sequence, this yields $|m\rangle$ in the register with probability $p=|P^{\Theta}_{\epsilon,\Delta}(a)|^2$. 

    We can now ask which values of $\Delta$ and $\epsilon$ should we choose for our intended behavior? Because we desire that $a = 1/\sqrt{N}$ be mapped to a value greater than $1 - \epsilon$, we require that $1/\sqrt{N} \geq \Delta/2$, or equivalently $\Delta = \mathcal{O}(1/\sqrt{N})$. Next, note that upon applying $(P^{\Theta}_{\epsilon,\Delta})^{(\text{SV})}(V)$ to the uniform superposition, the amplitude of the desired state, $ |m\rangle$, is at least $1-\epsilon$. Thus, upon measuring these qubits in the computational basis, the probability of the resultant state being $|m\rangle$ is at least $(1-\epsilon)^2 \geq 1 - 2\epsilon$. Therefore, if we want this procedure to succeed with probability at least $1-\delta$, we set $\epsilon = \delta/2$. 
    
    Under these choices, the polynomial $P^{\Theta}_{\epsilon,\Delta}(x)$ is of degree $d = \mathcal{O}\left(\frac{1}{\Delta}\log(1/\epsilon) \right) = \mathcal{O}(\sqrt{N} \log(1/\delta))$, such that the runtime of this algorithm is $\mathcal{O}(\sqrt{N} \log(1/\delta))$. Note that we again have $\mathcal{O}(\sqrt{N})$ scaling and thus maintain the quadratic speedup for the search problem, while circumventing Grover search's convergence issues. We summarize this algorithm for search by QSVT in Algorithm~\ref{alg:UnstructuredSearchQSVT}.
    
    Finally, we remark that this algorithm employs a simple block encoding (i.e. the identity) of $a=1/\sqrt{N}$ with nontrivial singular vectors (i.e. $|\psi_0 \rangle$ and $|m\rangle)$. In this setting, our desired transformation would not have been possible had we used the quantum eigenvalue transformation of Sec.~\ref{sec:eigenvalue_transforms}, as the eigenvalues of $V$ are necessarily $1$ and the input and output spaces identical. Hence, this instance illustrates the advantage of QSVT over the quantum eigenvalue transformation, as QSVT can achieve polynomial transformations between different input and output spaces. 
    
    \begin{algorithm}[] 
        \setstretch{1.05}
        \SetKwInOut{Comp}{Runtime}
        \SetKwInOut{Proc}{Procedure}
        \KwIn{Access to a controlled version of the oracle $U$ which bit-flips an auxiliary qubit when given an unknown target state $\lvert m \rangle \in \{\lvert j \rangle, j \in [N]\}$ and acts trivially otherwise, an error tolerance $\delta$, and a $\Delta \leq 2/\sqrt{N}$.}
        \KwOut{The flagged state $\lvert m \rangle$. }
        \Comp{$\mathcal{O} \left( \sqrt{N} \log{(1/\delta)}\right)$ queries to controlled-$U$, equivalently $\text{C}_{\Pi}\text{NOT}$, $\text{C}_{\Pi^\prime}\text{NOT}$, and $e^{-i\phi Z}$ gates and a single auxiliary qubit. Succeeds with probability at least $1 - \delta$.}
        \Proc{}
        ~Use QSVT to construct the operator $(P^{\Theta}_{\delta/2,\Delta})^{(\text{SV})}(V)$, where $V=I$ block encodes $a=1/\sqrt{N}$ in its $|m\rangle\langle\psi_0|$ element. \\
        ~Apply $(P^{\Theta}_{\delta/2,\Delta})^{(\text{SV})}(V)$ to the uniform superposition. With high probability, $|m\rangle$ remains in the register and can be determined by measuring in the computational basis. 
        \caption{Unstructured Search by QSVT}
        \label{alg:UnstructuredSearchQSVT}
    \end{algorithm}

\section{The Eigenvalue Threshold Problem by QSVT} \label{sec:Threshold}
    Beyond the determination of a desired state, as was the object in unstructured search covered in Algorithm~\ref{alg:UnstructuredSearchQSVT} in Section~\ref{sec:Search}, a separate major intent of many quantum algorithms is determining whether a certain property holds of a given quantum system.  Through the lens of QSVT, we show in this section that these two algorithmic goals are similar: here, rather than wishing to prepare a special state (e.g., the marked state for unstructured search), we now wish to encode some underlying unknown information about a quantum system simply into a state which we can then measure.
        
    This linkage between search and property testing is made apparent by considering the \emph{eigenvalue threshold problem}. In the setup of this problem, we are given access to a Hamiltonian $\mathcal{H}$ and would like to determine if $\mathcal{H}$ has any eigenvalues $\lambda$ that are less than some chosen threshold $\lambda_{\rm th}$. Realistically, this must be determined to within some precision $\Delta_\lambda$, so we are guaranteed that either the Hamiltonian has at least one eigenvalue $\lambda \leq \lambda_{\rm th} - \Delta_\lambda$ or that all of its eigenvalues obey $\lambda \geq \lambda_{\rm th} + \Delta_\lambda$. The eigenvalue threshold problem seeks to distinguish these cases. 
    
    To solve this problem, one is also provided a state $|\psi\rangle$ that has reasonable overlap with the low-energy subspace (spanned by the eigenvectors with eigenvalues $\lambda \leq \lambda_{\rm th} - \Delta_\lambda$) \emph{if it exists}. Mathematically, we represent this condition as
    \begin{equation}
        \| \Pi_{\leq \lambda_{\rm th} - \Delta_\lambda } \ket{\psi} \| \geq \zeta = \Omega(1),
    \end{equation}
    where $\Pi_{\leq \lambda_{\rm th} - \Delta_\lambda}$ is a projector onto the (possibly non-existent) low-eigenvalue subspace mentioned above. 
    
    To solve this problem with QSVT, we will need to construct a block encoding of the Hamiltonian $\mathcal{H}$. However, such a block encoding can only be constructed if $\|\mathcal{H}\| \leq 1$, so we instead must determine an $\alpha \geq \| \mathcal{H}\|$ and construct a unitary block encoding of $\mathcal{H}/\alpha$. In general, doing so may be nontrivial because determining such an $\alpha$ requires prior knowledge about $\mathcal{H}$. We discuss this drawback in Sec. \ref{sec:Discussion}. Fortunately however, for a large class of Hamiltonians, such as sparse Hamiltonians and linear combinations of unitaries, one can determine a sufficient $\alpha$ and construct a block encoding of $\mathcal{H}/\alpha$~\cite{Low_2019, Gily_n_2019}. With this rescaled block encoding, one can equivalently imagine that our goal is to determine if $\mathcal{H}/\alpha$ has any eigenvalues less than $\lambda_{\text{th}}/\alpha$ to within a precision $\Delta_\lambda/\alpha$.
    
    Moreover, as the eigenvalue threshold problem deals with eigenvalues, it is most straightforwardly approached with a quantum eigenvalue transform. However, in accordance with the theme of this paper, we may equivalently solve this problem with QSVT if the eigenvalues of $\mathcal{H}$ are positive, such that the singular values are equal to the eigenvalues. If this is not the case, we may instead use the block encoding of $\mathcal{H}/\alpha$ and the circuit in Figure~\ref{fig:H_PosDef_BlockEncodingCircuit} to construct a block encoding of the positive definite matrix $\frac{1}{2}\left(\mathcal{H}/\alpha+I\right)$. Under this modification, the new goal would be to determine if $\frac{1}{2}\left(\mathcal{H}/\alpha+I\right)$ has eigenvalues less than $\frac{1}{2} \left(\lambda_{\text{th}}/\alpha+1\right)$ to within a precision $\Delta_\lambda/\alpha$, which may be achieved with QSVT. In the remainder of this section, we will assume that this issue has been alleviated and specialize to the case of distinguishing the threshold $\lambda_{\rm th}/\alpha$. 
    
    \begin{figure}[htpb]
        \centering
        \includegraphics[width=0.55\columnwidth]{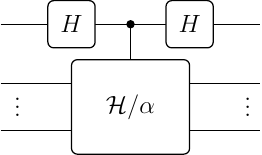}
        \caption{A simple circuit for block encoding $\frac{1}{2} \left(\mathcal{H}/\alpha+I\right)$, the projectors for which are, following QSVT conventions, $|0\rangle\langle0| \otimes \Pi$ and $|0 \rangle\langle 0| \otimes \tilde{\Pi}$.}
        \label{fig:H_PosDef_BlockEncodingCircuit}
    \end{figure}
    
    With those concerns out of the way, how might this problem be solved with QSVT? A fruitful approach, which we will follow, is to apply QSVT to the Hamiltonian with a function that filters out large eigenvalues. While one might suspect that we need a function that locates eigenvalues below the threshold $\lambda_{\rm th}/\alpha$, such as the eigenstate filtering function of~\cite{Lin_2020}, we can actually employ a much simpler construction. All we need is a function that behaves markedly different on eigenvalues greater than and less than $\lambda_{\text{th}}/\alpha$, and such a function may be constructed using the technology from Section~\ref{sec:Search}.
        
    A good choice for such a function is $\Theta(\lambda_{\text{th}}/\alpha-x)$, which maps eigenvalues less than $\lambda_{\text{th}}/\alpha$ to $1$ and larger eigenvalues to $-1$. One could imagine approximating this function by $P^{\Theta}_{\epsilon, \Delta}(\lambda_{\text{th}}/\alpha-x)$, but there is a catch: this polynomial does not have definite parity and thus does not satisfy the constraints of Theorem~\ref{thm:qsp}. Instead, we need to consider a symmetrized version with even parity, the \emph{eigenvalue threshold polynomial}:
        \begin{equation}
        \begin{split}
            P&^{\text{ET}}_{\epsilon,\Delta,\lambda_{\text{th}}/\alpha}(x)  := \\
            &\frac{1}{1+\frac{\epsilon}{4}} \left( -1+\frac{\epsilon}{4} +  P^{\Theta}_{\frac{\epsilon}{2},\Delta}\left(\tfrac{\lambda_{\text{th}}}{\alpha} - x \right) + P^{\Theta}_{\frac{\epsilon}{2},\Delta}\left(\tfrac{\lambda_{\text{th}}}{\alpha} + x \right) \right)
        \end{split}
        \end{equation}
    The factors of $\frac{\epsilon}{4}$ ensure that $P^{\text{ET}}_{\epsilon,\Delta,\lambda_{\text{th}}/\alpha}(x)$ has the following desired properties, which are easily seen via the triangle inequality:
        \begin{enumerate}
            \item $|P^{\text{ET}}_{\epsilon,\Delta,\lambda_{\text{th}}/\alpha}(x)|\leq 1 \text{ for } x \in [0,1]$.
            \item $|\Theta(\lambda_{\text{th}}/\alpha-x) - P^{\text{ET}}_{\epsilon,\Delta,\lambda_{\text{th}}/\alpha}(x)| \leq \epsilon$ \\ for $x \in \big[ 0,1 \big] \backslash \left(\frac{\lambda_{\text{th}}}{\alpha}-\frac{\Delta}{2}, \ \frac{\lambda_{\text{th}}}{\alpha}+\frac{\Delta}{2}\right)$.
        \end{enumerate}
    So, $P^{\text{ET}}_{\epsilon,\Delta,\lambda_{\text{th}}/\alpha}(x)$ behaves as $P^{\Theta}_{\epsilon,\Delta}(\lambda_{\text{th}}/\alpha-x)$ for $x\geq 0$ (the only relevant range for singular values). In addition, $P^{\text{ET}}_{\epsilon,\Delta,\lambda_{\text{th}}/\alpha}(x)$ has definite parity and magnitude bounded by $1$, so it may be implemented through QSVT (again, it obeys the conditions on Poly($a$)$= \langle + | U_{\vec{\phi}} | +\rangle$ discussed in Section~\ref{sec:qsp} and Appendix \ref{sec:QSPConventions}). Thus, this polynomial may be used to distinguish between eigenvalues less than $\lambda_{\text{th}}/\alpha - \Delta/2$ and greater than $\lambda_{\text{th}}/\alpha + \Delta/2$, which naturally indicates a precision $\Delta_\lambda/\alpha = \Delta/2$.
    
    To see how we may employ this technology, note that we may express $|\psi\rangle$ in the eigenbasis as 
    \begin{equation}
        |\psi \rangle = \sum_{\lambda} c_\lambda | \lambda \rangle = \sum_{\lambda \leq \lambda_{\rm th}} c_\lambda | \lambda \rangle + \sum_{\lambda > \lambda_{\rm th}} c_\lambda | \lambda \rangle,
    \end{equation}
    where $\sum_{\lambda \leq \lambda_{\rm th}} |c_\lambda|^2 \geq \zeta^2$ by the guarantee mentioned above. Next, consider using QSVT to develop $(P^{\text{ET}}_{\epsilon,\Delta})^{(\text{SV})}(\mathcal{H}/\alpha)$ and apply it to $|\psi\rangle$ controlled by an ancilla qubit in the state $|+\rangle$, as in the circuit of Figure~\ref{fig:ThresholdMeasurementCircuit}. After applying a Hadamard to this ancilla qubit, the final (unnormalized) state of the system is
    \begin{equation}
    \begin{split}
        \frac{1}{2} \Bigg( |0\rangle \sum_{\lambda}& c_\lambda (1+P^{\text{ET}}_{\epsilon,\Delta}(\lambda)) | \lambda \rangle \\
        & + |1\rangle \sum_{\lambda } c_\lambda (1-P^{\text{ET}}_{\epsilon,\Delta}(\lambda)) | \lambda \rangle \Bigg).
    \end{split}
    \end{equation}

    \begin{figure}[htpb]
        \centering
        \includegraphics[width=.7\columnwidth]{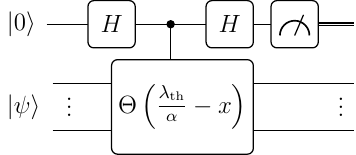}
        \caption{The circuit which, upon measuring the auxiliary qubit, can be used to solve the eigenvalue threshold problem with high probability as in Algorithm \ref{alg:EigenvalueThresholdQSVT}.}
        \label{fig:ThresholdMeasurementCircuit}
    \end{figure}
    
    We can solve the eigenvalue threshold problem by measuring the ancilla qubit in the computational basis. Observe that the probability of measuring the ancilla qubit in $|0\rangle$ is 
    \begin{equation}
    \begin{gathered}
        p(0) = \frac{1}{2} \frac{\sum |c_\lambda|^2 (1+P^{\text{ET}}_{\epsilon,\Delta,\lambda_{\text{th}}/\alpha}(\lambda))^2}{\sum |c_\lambda|^2  (1+(P^{\text{ET}}_{\epsilon,\Delta,\lambda_{\text{th}}/\alpha}(\lambda))^2)}
    \end{gathered}
    \end{equation}
    If all eigenvalues obey $\lambda > \lambda_{\rm th} + \Delta_\lambda$, then $-1 \leq P^{\text{ET}}_{\epsilon,\Delta,\lambda_{\text{th}}/\alpha}(\lambda) \leq -1+\epsilon$, and so
    \begin{equation}
    \begin{gathered}
        p(0 | \  \lambda > \lambda_{\rm th} + \Delta_\lambda) \leq \frac{\epsilon^2}{2(1+(-1+\epsilon)^2)} \leq \frac{1}{2}\epsilon^2.
    \end{gathered}
    \end{equation}
    This implies that it is unlikely to measure $|0\rangle$ if all eigenvalues obey $\lambda > \lambda_{\rm th} + \Delta_\lambda$. On the other hand, if there exists an eigenvalue that obeys $\lambda < \lambda_{\rm th} - \Delta_\lambda$, then $1-\epsilon \leq P^{\text{ET}}_{\epsilon,\Delta,\lambda_{\text{th}}/\alpha}(\lambda < \lambda_{\rm th} - \Delta_\lambda) \leq 1$ and $-1 \leq P^{\text{ET}}_{\epsilon,\Delta,\lambda_{\text{th}}/\alpha}(\lambda > \lambda_{\rm th} + \Delta_\lambda) \leq -1+\epsilon$, which implies that 
    \begin{equation}
    \begin{gathered}
        p(0 | \  \exists \lambda < \lambda_{\rm th} - \Delta_\lambda) \geq \frac{\zeta^2(2-\epsilon)^2}{2(1+1)} \geq \zeta^2(1-\epsilon).
    \end{gathered}
    \end{equation}
    This indicates that, for reasonably sized $\zeta$, we are more likely to measure $|0\rangle$ if there exists an eigenvalue $\lambda < \lambda_{\rm th} - \Delta_\lambda$.
    
    With this construction, the eigenvalue threshold problem is reduced to distinguishing between two Bernoulli distributions: one with mean $\leq \frac{1}{2}\epsilon^2$ and one with mean $\geq \zeta^2(1-\epsilon)$. To ensure that these means are well separated, we desire $\zeta^2(1-\epsilon) > \frac{1}{2}\epsilon^2$. Since $\epsilon \leq \sqrt{2/e\pi} \approx 0.48$ by the construction of $P^{\Theta}_{\epsilon, \Delta}(x)$ in~\cite{low2017quantum}, this may be enforced by choosing an $\epsilon < \zeta$ (i.e. $\epsilon = \mathcal{O}(\zeta)$); for instance, $\epsilon = \frac{1}{4}\zeta$ is a sufficient choice, which we will make use of in the following. We may then distinguish between these distributions with high probability by performing repeated measurements. This is quantified by the following lemma:
    
    \begin{lemma} \label{lemma:multiplicative_chernoff_bound}
        Suppose we have two Bernoulli distributions $X_0, X_1$ which return elements in $\{0, 1\}$ and have expected values $0 \leq a < b \leq 1$ respectively. Given a distribution $X$ that is either $X_0$ or $X_1$, we may determine $X$'s identity with confidence greater than $1 - \delta$ for $\delta > 0$ with $\mathcal{O}\left(\frac{b + a}{(b - a)^2}\log(1/\delta)\right)$ samples of $X$.
        
        \begin{proof}
            To distinguish between these distributions, let us take $n$ samples $X_i$, and guess that $X$ is the distribution whose mean is closest to the empirical mean of the samples. Suppose that $X = X_0$ and thus has mean $a$ (the case $X=X_1$ may be treated analogously). We're interested in computing the probability of mistaking this distribution for $X_1$, namely
            \begin{align}
                &\hphantom{{}={}}P(\text{error}) = P\left(\sum_{i=1}^n X_i \geq n(a+b)/2 \right) \nonumber\\ 
                &= P\left(\sum_{i=1}^n X_i \geq n[a + (b-a)/2] \right).
            \end{align}
            To bound this error probability, we may invoke the well-known multiplicative Chernoff bound, with mean $\mu=\mathbb{E}(\sum_{i=1}^n X_i) = na$ and relative difference from the mean $\varepsilon$,
                \begin{equation}
                    \mu + n(b - a)/2 = \mu(1 + \varepsilon)
                    \;\implies\;
                    \varepsilon = \frac{b - a}{2a}.
                \end{equation}
            Employing the multiplicative Chernoff bound, we may ensure that the error probability be less than some chosen $\delta$ as
                \begin{equation}
                    P\left(\sum_{i=1}^n X_i \geq na(1+\varepsilon) \right) \leq e^{-\varepsilon^2 n a/3} \leq \delta,
                \end{equation}
            which in turn requires $n \geq (3/a\varepsilon^2)\log(1/\delta)$, or equivalently, $n \geq \frac{12a}{(b - a)^2}\log(1/\delta) = \mathcal{O}\big(\frac{a}{(b - a)^2}\log(1/\delta)\big)$. 
            
            Performing a similar calculation for the hypothesis $X=X_1$ yields $n = \mathcal{O}\big(\frac{b}{(b - a)^2}\log(1/\delta)\big)$. Summing these results over equal priors over $X \in \{X_0, X_1\}$ provides the scaling $n \geq \frac{6(a+b)}{(b - a)^2}\log(1/\delta) = \mathcal{O}\big(\frac{a+b}{(b - a)^2}\log(1/\delta)\big)$, as claimed in the lemma statement.
        \end{proof}
    \end{lemma}
    
    Applying this lemma to our scenario in the eigenvalue threshold problem, we have $\frac{1}{b-a} \leq \frac{1}{\zeta^2(1-\epsilon)-\epsilon^2/2}=\mathcal{O}(1/\zeta^2)$. For instance, with our choice $\epsilon = \zeta/4$, we have $\frac{1}{b-a} < \frac{3}{2\zeta^2}$. We then have $\frac{a+b}{(b-a)^2} = \frac{(b-a)+2a}{(b-a)^2} = \mathcal{O}(1/\zeta^2+\zeta^2/\zeta^4) = \mathcal{O}(1/\zeta^2)$. Therefore, if we repeat the aforementioned measurement procedure $\mathcal{O}\left(\frac{1}{\zeta^2} \log(1/\delta)\right)$ times, we can correctly distinguish between the two distributions and solve the eigenvalue threshold problem with probability at least $1-\delta$. As each repetition performs QSVT with a polynomial of degree $\mathcal{O}\left(\frac{1}{\Delta} \log(1/\epsilon)\right) = \mathcal{O}\left(\frac{\alpha}{\Delta_\lambda} \log(1/\zeta)\right)$, this overall algorithm queries the Hamiltonian $\mathcal{O}\left(\frac{\alpha}{\Delta_\lambda \zeta^2}\log(1/\zeta)\log(1/\delta)\right)$ times. For consistency, we note that this $1/\zeta^2$ dependence is suboptimal, and can be improved to $1/\zeta$ using Lemma 7 of \cite{Lin_2020_2}. We summarize the eigenvalue threshold problem by QSVT in Algorithm~\ref{alg:EigenvalueThresholdQSVT}. Here, we note that the upper limit on the index of the for loop comes from the condition $n \geq \frac{6(a+b)}{(b - a)^2}\log(1/\delta)$ according to Lemma~\ref{lemma:multiplicative_chernoff_bound}, with the choice $\epsilon=\zeta/4$.

    \begin{algorithm}[htpb]
    \setstretch{1.05}
     \SetKwInOut{Comp}{Runtime}
     \SetKwInOut{Proc}{Procedure}
     \KwIn{Access to a Hamiltonian, $\mathcal{H}$, and an $\alpha \geq \|\mathcal{H}\|$.}
     \KwOut{Whether $\mathcal{H}$ has at least one eigenvalue $\lambda < \lambda_{\rm th} - \Delta_{\lambda}$, or if all eigenvalues obey $\lambda > \lambda_{\rm th} + \Delta_{\lambda}$.}
     \Comp{$\mathcal{O}\left(\frac{\alpha}{\Delta_\lambda \zeta^2}\log(1/\zeta)\log(1/\delta) \right)$ uses of a unitary block encoding of $\mathcal{H}/\alpha$ and its inverse, and $\text{C}_{\Pi}\text{NOT}$, $\text{C}_{\tilde{\Pi}}\text{NOT}$ and single qubit phase gates. Succeeds with probability at least $1 - \delta$.}
     \Proc{}
     $ $ Prepare a unitary block encoding of $\mathcal{H}/\alpha$. \\
     $ $ \For{$i=1$ to $\lceil \frac{315}{32\zeta^2} \log(1/\delta)\rceil$}{
      Use QSVT to apply $\left(P^{\text{ET}}_{\frac{\zeta}{4},\frac{2\Delta_\lambda}{\alpha},\frac{\lambda_{\text{th}}}{\alpha}} \right)^{(\text{SV})} (\mathcal{H}/\alpha)$ to $|\psi\rangle$, controlled by an ancilla qubit in the state $|+\rangle$ (see Fig. \ref{fig:ThresholdMeasurementCircuit}).\\
      Apply a Hadamard to the ancilla qubit and measure it in the computational basis.
     }
     $ $ If the fraction of $|0\rangle$'s measured is closer to $\zeta^2(1-\zeta/4)$ than to $\frac{\zeta^2}{32}$, then there exists an eigenvalue $\lambda \leq \lambda_{\rm th} - \Delta_\lambda$ with high probability. Otherwise, all eigenvalues obey $\lambda \geq \lambda_{\rm th} + \Delta_{\lambda}$ with high probability.
    \caption{The Eigenvalue Threshold Problem by QSVT}
     \label{alg:EigenvalueThresholdQSVT}
    \end{algorithm}

\section{Phase Estimation by QSVT}
\label{sec:PhaseEstimation}

    Another salient problem that admits a quantum advantage is \emph{phase estimation}, which is as follows: given access to a unitary oracle $U$ and an eigenvector $|u\rangle$ such that $U|u\rangle = e^{2\pi i \varphi}|u\rangle $, determine $\varphi$. Here, we approach this problem by integrating a feedback procedure inspired by Kitaev's algorithm~\cite{kitaev2002classical, kitaev1995quantum} with the technology developed in solving the eigenvalue threshold problem in Sec. \ref{sec:Threshold}. We first provide a sketch of the algorithm, and then address a few caveats in its construction. Finally, we present a concrete formulation of QSVT-based phase estimation, analyze its capabilities, and demonstrate how the quantum Fourier transform naturally emerges from this construction.

\subsection{Intuition}\label{sec:PhaseEstimation_Motivation}
    To motivate iterative phase estimation by QSVT, recall that Kitaev's algorithm for phase estimation~\cite{kitaev2002classical, kitaev1995quantum} may be reinterpreted as a semi-classical feedback process~\cite{Parker_2000, Monz_2016}. One may view this feedback process as incorporating a classical parameter $\theta \in [0,1]$, initialized at $\theta = 0$, such that through a series of controlled-$U^j$ calls and a single-qubit measurements, each of which updates $\theta$, one obtains a value $\theta \approx \varphi$ at the end of the algorithm.
    
    Our approach to phase estimation by QSVT is heavily inspired by this feedback procedure. Just as in Kitaev's algorithm, we introduce a parameter $\theta \in [0,1]$, initialized at $0$, that is updated throughout the algorithm and by the end of the procedure, the final value of $\theta$ will approximate $\varphi$. We sketch our algorithm in the following section. 
    
    Before proceeding, we note that a very similar QSVT-based phase estimation algorithm was independently discovered in \cite{rall2021faster}. A primary goal of their work is maintaining coherence, and accordingly their paper presents a thorough performance analysis of their algorithm including the derivation of constant factors as well as precise use of the diamond norm. Ref.~\cite{rall2021faster} also extends their QSVT-based phase estimation algorithm to coherent energy and amplitude estimation.

    \subsubsection{Sketch of the Algorithm}\label{sec:PhaseEstimation_Sketch}
    To understand our algorithm, let us use the binary fraction representation of $\varphi$ and assume that $\varphi$ is an $m$-bit number: $\varphi = 0.\varphi_1 \varphi_2 ... \varphi_m$ for some finite $m$ (in this representation, $\varphi = \sum_j 2^{-j} \varphi_j$). We will show how to deal with the $m = \infty$ case later. 
    
    In the QSVT phase estimation algorithm, $\varphi$ is determined bit by bit through a series of $m$ iterations, beginning with the least significant bit, $\varphi_m$, and proceeding down to the most significant bit, $\varphi_1$. In each iteration, we construct a matrix whose singular values encode the least significant bits of $\varphi$ as well as the current value of $\theta$. By applying QSVT to this matrix, conditioned on an ancilla qubit, and subsequently measuring the ancilla qubit, we determine a single bit of $\varphi$ and update $\theta$ accordingly. If each iteration succeeds, then $\theta = \varphi$ at the end of the algorithm.
        
    To see how this may be achieved, consider the matrix $A_j(\theta) := \frac{1}{2} (I+e^{-2\pi i \theta} U^{2^j})$, which has a singular value 
        \begin{equation}\label{eq:singular_value_expression}
            \sigma^j := \big|\cos\big(\pi (2^j \varphi - \theta)\big) \big|,
        \end{equation}
    with corresponding right singular vector $|v^j\rangle = |u\rangle$ and left singular vector $|w^j\rangle = e^{i\pi(2^j \varphi-\theta)}|u\rangle$. Because $\text{exp}({\pi i 2^j \varphi}) = \text{exp}(2\pi i 2^{j-1} 0.\varphi_1 ... \varphi_m) = \text{exp}(\pi i 0.\varphi_{j+1}  ... \varphi_m) $, $\sigma^j$ may be re-expressed as 
        \begin{equation}\label{eq:singular_value_reexpression}
            \sigma^j = \big|\cos\big(\pi (0.\varphi_{j+1} \varphi_{j+2}...\varphi_m - \theta)\big) \big|.
        \end{equation}
        
    This singular value is used to extract the bits of $\varphi$ as follows. Consider the first iteration of the QSVT phase estimation algorithm, where we begin with $j=m-1$ (and will decrement $j$ down to $j=0$). Because $\theta$ is initialized at $0$, the singular value of interest is 
    \begin{equation} \label{eq:s_m}
        \sigma^{m-1} = \big|\cos\big(\pi (0.\varphi_{m})\big) \big| = 
        \begin{cases}
            1 & \varphi_m=0 \\
            0 & \varphi_m=1,
        \end{cases}
    \end{equation}
    from which $\varphi_m$ may be evaluated as $\varphi_m = \frac{1}{2}\left( 1 + \Theta\left(\frac{1}{\sqrt{2}} - \sigma^{m-1} \right) \right)$. This expression provides a pathway to extract $\varphi_m$, akin to the procedure used in the eigenvalue threshold problem: apply QSVT to $A_{m-1}(\theta)$ with the target function $\Theta\left(\frac{1}{\sqrt{2}}  -x \right)$\footnote{The sign function is discontinuous and cannot be implemented exactly through QSVT. We will deal with this caveat shortly, but for now we assume it can be implemented exactly.}, and employ the circuit in Figure~\ref{fig:PhaseEstimationCircuit_SingleIteration}, the measurement result of which is $\varphi_m$. After measurement, we update $\theta$ by setting the first bit of $\theta$ to $\theta_1 = \varphi_m$ (again using the binary fraction representation $\theta = 0.\theta_1\theta_2...$), such that $\theta = 0.\varphi_m$, which completes the $j=m-1$ iteration. 
    
    We then move to the next iteration, wherein we decrement $j=m-1 \leftarrow m-2$, map $\theta \leftarrow \theta /2 = 0.0\varphi_m$, and construct $A_j(\theta) = A_{m-2}(0.0\varphi_m)$, which has singular value
    \begin{equation}
    \begin{split}
        \sigma^{m-2} &= \big|\cos\big(\pi (0.\varphi_{m-1} \varphi_m - 0.0\varphi_m)\big) \big| \\
        &= \big|\cos\big(\pi (0.\varphi_{m-1})\big) \big|.
    \end{split}
    \end{equation}
    Note that the updated value of $\theta$ allowed for a convenient cancellation of $\varphi_m$. As this relation is identical to Eq.~(\ref{eq:s_m}), $\varphi_{m-1}$ may be determined by employing the QSVT procedure mentioned above and again executing the circuit in Fig~\ref{fig:PhaseEstimationCircuit_SingleIteration}. Upon obtaining the measurement result of this circuit, we set $\theta_1 = \varphi_{m-1}$. 
    
    It is clear that this procedure may be repeated to determine each bit of $\varphi$, such that at the end of the $j=0$ iteration, one obtains $\theta = \varphi$, which solves the phase estimation problem. Ultimately, this procedure succeeds because $\theta$ encodes the least significant bits of $\varphi$, from which the next bit of $\varphi$ may be extracted by cleverly transforming the singular values of $A_j(\theta)$. Accordingly, this algorithm may be seen as a binary search through the bits of $\varphi$, a key idea that we return to in Section~\ref{sec:EmergentQFT}.
    
    At this stage however, one may object to this algorithm sketch as it requires $m$ iterations, and it could be the case that $m$ is unknown or prohibitively large. This is not an issue, as we prove in Appendix~\ref{sec:PhaseEstimation_m_n_Proofs} that if we start at $j=n-1$ for some positive integer $n$, then we have (with a minor procedural modification needed for $n < m$):
    \begin{theorem} \label{thm:n_geq_m}
    If $n \geq m$, and one can exactly implement the sign function through QSVT, then at the end of this procedure, $\theta = \varphi$. 
    \end{theorem}
    and 
    \begin{theorem} \label{thm:n_l_m}
    If $n < m$ (including the case $m=\infty$),  and one can exactly implement the sign function through QSVT, then at the end of this procedure, $|\theta - \varphi| \leq 2^{-n-1}$. 
    \end{theorem}
    \noindent As evidenced by these theorems, which are easily proven by induction, this procedure necessarily outputs an $n$-bit approximation to $\varphi$ and is the essence of the algorithm for phase estimation by QSVT. However, in spite of the simple appearance of this algorithm, a few issues must be addressed before presenting its concrete formulation.
    
    \begin{figure}[htpb]
        \centering
        \includegraphics[width=0.8\columnwidth]{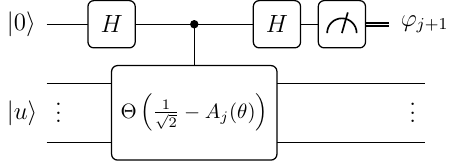}
        \caption{The quantum circuit used to evaluate $\varphi_{j+1}$ at each iteration of the phase estimation by QSVT procedure.}
        \label{fig:PhaseEstimationCircuit_SingleIteration}
    \end{figure}

    \subsubsection{Caveats}\label{sec:PhaseEstimation_Caveats}
    Here we address a few caveats in the above algorithm sketch. First, as described thus far, this procedure applies QSVT to the matrix $A_j(\theta)$, which requires a unitary block encoding of $A_j(\theta)$. Fortunately, this encoding can be easily constructed as the circuit in Figure~\ref{fig:A_j_BlockEncodingCircuit}, which we denote by $W_j(\theta)$. A simple calculation of Figure~\ref{fig:A_j_BlockEncodingCircuit} indicates that $W_j(\theta)$ has the desired form: 
    \begin{equation}\label{eq:A_j_BlockEncoding}
    \begin{split}
        W_j(\theta) :&= \frac{1}{2}
        \begin{pmatrix}
        I+e^{-2\pi i \theta }U^{2^j} & I-e^{-2\pi i \theta }U^{2^j} \\
        I-e^{-2\pi i \theta }U^{2^j} & I+e^{-2\pi i \theta }U^{2^j}
        \end{pmatrix} \\
        &= 
        \begin{pmatrix}
        A_j(\theta) & \cdot \\
        \cdot & \cdot
        \end{pmatrix},
    \end{split}    
    \end{equation}
    with corresponding projection operators $\Pi = \tilde{\Pi} = |0\rangle \langle 0| \otimes I$.
    
    \begin{figure}[htpb]
        \centering
        \includegraphics[width=0.5\columnwidth]{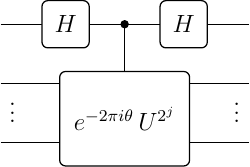}
        \caption{A quantum circuit that achieves a unitary block encoding of $A_j(\theta)$, which we denote by $W_j(\theta)$ and incorporate into Algorithm~\ref{alg:PhaseEstimationByQSVT}.}
        \label{fig:A_j_BlockEncodingCircuit}
    \end{figure}
    
    Another concern with the aforementioned procedure is that the sign function, $\Theta\left(\frac{1}{\sqrt{2}} - x\right)$, may seem over-complicated and unnecessary, as the same results could be attained with a simpler function, such as $1-x$. However, the sign function is crucial when one desires merely an $n$-bit approximation to $\varphi$, which is necessary when $m$ is unknown or prohibitively large, as is often the case in practice. In this situation, the use of the sign function (the exact sign function, rather than an approximation) will ultimately set $\theta$ equal to the $n$-bit number closest to $\varphi$. For example, if $0. \varphi_m \varphi_{m+1} ... = 0.011...$ at the first iteration, then despite the fact that $\varphi_{m} = 0$, the use of the sign function will set $\theta_1 = 1$ because $0.1$ is the best 1-bit approximation to $0.011...$. 
    
    Lastly, the most significant problem with this procedure as it currently stands is that the sign function is not a polynomial and therefore cannot be implemented exactly in QSVT. Instead, $\Theta\big(\frac{1}{\sqrt{2}} - x\big)$ must be approximated by a polynomial, a good choice for which is $P^{\Theta}_{\epsilon, \Delta} \big(\frac{1}{\sqrt{2}} -x\big)$ from Sec. \ref{sec:Search}. However, as in Sec. \ref{sec:Threshold}, in order to ensure that the target polynomial has definite parity, we must use a symmetrized polynomial, the \emph{phase estimation polynomial}:
    \begin{equation}\label{eq:SignFunction_Target}
    \begin{split}
            P&^{\text{PE}}_{\epsilon,\Delta}(x) := \\
            &\frac{1}{1+\frac{\epsilon}{4}} \left( -1+\frac{\epsilon}{4} +  P^{\Theta}_{\frac{\epsilon}{2},\Delta}\left(\tfrac{1}{\sqrt{2}} - x \right) + P^{\Theta}_{\frac{\epsilon}{2},\Delta}\left(\tfrac{1}{\sqrt{2}} + x \right) \right).
    \end{split}
    \end{equation}
    Like the eigenvalue threshold polynomial, $P^{\text{PE}}_{\epsilon,\Delta}(x)$ is even, is bounded in magnitude by $1$, and behaves as $P^{\Theta}_{\epsilon,\Delta}\left(\frac{1}{\sqrt{2}} - x\right)$ for $x \geq 0$, which is the only relevant range for singular values.

    This approximation comes with the potential for error, as an iteration of this procedure may now fail with a nonzero probability, where failure is defined as measuring an incorrect bit of an approximation to $\varphi$ (an approximation that suffers error $\leq 2^{-n-1}$), which will result in an inaccurate $\theta$ such that $|\theta-\varphi|\geq 2^{-n-1}$. Therefore, we reformulate our goal probabilistically as: obtain an $n$-bit approximation to $\varphi$ with probability at least $1-\delta$.
    
    The probability of failure is dictated by the values of $\Delta$ and $\epsilon$. In particular, for a nonzero $\Delta$, note that if $\sigma^j \in \left[\frac{1}{\sqrt{2}} - \frac{\Delta}{2}, \ \frac{1}{\sqrt{2}} + \frac{\Delta}{2}\right]$, where $P^{\text{PE}}_{\epsilon,\Delta}(x)$ is not a good approximation to the sign function, then an error may occur. Fortunately, if $\Delta$ is made sufficiently small, this error mode does not induce significant errors:
    \begin{theorem}\label{thm:Delta_restriction}
    If we choose $\Delta$ such that 
        \begin{equation}
        \begin{gathered}
            \Delta < 2\left(\cos(\frac{3\pi}{16}) - \frac{1}{\sqrt{2}} \right) \approx 0.25,
        \end{gathered}
        \end{equation}
    then $\sigma^j \in \left[\frac{1}{\sqrt{2}} - \frac{\Delta}{2}, \frac{1}{\sqrt{2}} + \frac{\Delta}{2} \right]$ is only possible at the $j=n-1$ iteration, such that an error due to $\Delta$ can only occur at the $j\textsuperscript{th}$ iteration. If an error is made at this iteration, then at the end of the algorithm, then $|\theta - \varphi| < \frac{1}{2^n}$, assuming no errors are made at later iterations.
    \end{theorem}
    \noindent This Theorem, proven in Appendix~\ref{sec:Delta_mitigation}, indicates that the error caused by $\Delta$ is at most $2^{-n}$ when we obey this bound, which we assume is obeyed in the following.

    Next, a nonzero $\epsilon$ can also induce errors. Because $P^{\text{PE}}_{\epsilon,\Delta}(x)$ $\epsilon$-approximates $\Theta\left(\frac{1}{\sqrt{2}} - x\right)$ outside of $\left[\frac{1}{\sqrt{2}} - \frac{\Delta}{2}, \frac{1}{\sqrt{2}} + \frac{\Delta}{2} \right]$, the failure probability of the measurement in Figure~\ref{fig:PhaseEstimationCircuit_SingleIteration} goes as $\mathcal{O}(\epsilon^2)$, and this type of error may occur at any iteration. To study this quantitatively, consider iteration $j < n-1$, such that $\sigma^j \notin \left[\frac{1}{\sqrt{2}} - \frac{\Delta}{2}, \frac{1}{\sqrt{2}} + \frac{\Delta}{2} \right]$, and focus on the scenario $\sigma^j < \frac{1}{\sqrt{2}} - \frac{\Delta}{2}$, in which case $P^{\text{PE}}_{\epsilon,\Delta}(\sigma^j) \in [-1,-1+\epsilon]$, and $\theta_1 = 0$ is the ideal measurement (the case in which $\sigma^j > \frac{1}{\sqrt{2}} + \frac{\Delta}{2}$ is analogous). By evaluating the circuit in Figure~\ref{fig:PhaseEstimationCircuit_SingleIteration}, with $P^{\text{PE}}_{\epsilon,\Delta}(\sigma^j)$ in place of the sign function, we find that the final (unnormalized) state of the ancilla qubit is
    \begin{equation}
    \begin{gathered}
        \frac{1}{2} \Big( \left(1-P^{\text{PE}}_{\epsilon,\Delta}(\sigma^j)\right)|0\rangle + \left(1+P^{\text{PE}}_{\epsilon,\Delta}(\sigma^j) \right)|1\rangle \Big),
    \end{gathered}
    \end{equation}
    such that the probability of failure (measuring $|1\rangle$) at this iteration is 
    \begin{equation}
    \begin{gathered}
        p_{\text{fail}} =  \frac{\big(1+P^{\text{PE}}_{\epsilon,\Delta}(\sigma^j)\big)^2}{2\Big(1+\big(P^{\text{PE}}_{\epsilon,\Delta}(\sigma^j)\big)^2\Big)}
        \leq \frac{\epsilon^2}{2\big( 1+ (-1+\epsilon)^2 \big)} \leq \frac{1}{2}\epsilon^2
    \end{gathered}
    \end{equation}
    for $\epsilon \leq 1$.
    
    What $\epsilon$ is sufficient for our purposes? As our procedure consists of $n$ iterations, we desire that each iteration fails with probability no greater than $\delta /n$, such that the overall algorithm succeeds with probability at least $ 1-\delta$ by the union bound. From the above inequality, this can be enforced by the choice $\epsilon \leq \sqrt{2\delta/n}$. Therefore, with these conditions on $\Delta$ and $\epsilon$, this algorithm sketch succeeds with probability $\geq 1-\delta$ at determining a $\theta$ such that $|\theta - \varphi| < \frac{1}{2^n}$, as desired.

    \subsection{The Complete Algorithm}

    After supplying intuition and addressing potential issues, we now present the complete algorithm for phase estimation by QSVT in Algorithm~\ref{alg:PhaseEstimationByQSVT}. We depict this algorithm's circuit in Figure~\ref{fig:abstract_qsvt_circuit}, which illustrates is resemblance to the inverse quantum Fourier transform, a connection we elaborate on in Section~\ref{sec:EmergentQFT}. We further unpack the details of the phase estimation by QSVT circuit in Figure~\ref{fig:qsvt_phase_estimation}; here, the operator $R_{\ell}$ is defined as
    \begin{equation}\label{eq:R_l}
    \begin{split}
        R_{\ell} := \begin{bmatrix}
            1 & 0 \\
            0 & e^{2\pi i/2^\ell } \\
        \end{bmatrix},
    \end{split}
    \end{equation}
    which is a $z$-rotation, up to a global phase.
    
    \begin{algorithm}[htbp]
        \setstretch{1.05}
        \SetKwInOut{Comp}{Runtime}
        \SetKwInOut{Proc}{Procedure}
        \KwIn{An oracle that performs a controlled-$U^j$ operation, an eigenstate $|u\rangle$ of $U$ with eigenvalue $e^{2\pi i \varphi}$, and $n+1$ ancilla qubits (or 1 ancilla qubit that is reused $(n+1)$ times). Also, an $\epsilon \leq \sqrt{2\delta/(n+1)}$ and a $\Delta < 2 \left( \cos(\frac{3\pi}{16}) - \frac{1}{\sqrt{2}} \right) \approx 0.25$.}
        \KwOut{A $\theta$ such that $|\theta - \varphi| \leq \frac{1}{2^{n}}$.}
        \Comp{$\mathcal{O}\left( n \log(\frac{n}{\delta})\right)$ queries to the controlled-$U^j$ oracle. Succeeds with  with probability $\geq 1-\delta$.}
        \Proc{}
        $ $ Initialize $\theta = 0$.\\
        $ $ Apply a Hadamard gate to each auxiliary qubit.\\
        $ $ \For{$j=n-1$ down to $0$}{
        $\theta \leftarrow \theta/2$ (equivalently, $\theta = 0.\theta_1 \theta_2 ... \leftarrow 0.0\theta_1 \theta_2 ...$)\\
        Construct a block encoding of $A_j(\theta) := \frac{1}{2}(I+e^{-2\pi i \theta }U^{2^j})$ as per Figure~\ref{fig:A_j_BlockEncodingCircuit}.\\
        Use QSVT to apply the operator $(P^{\text{PE}}_{\epsilon,\Delta})^{(\text{SV})}(A_j(\theta))$ to $|u\rangle$, controlled on an ancilla qubit, where $P^{\text{PE}}_{\epsilon,\Delta}(x)$ is defined in Eq.~(\ref{eq:SignFunction_Target}). \\
        Apply a Hadamard gate to the ancilla qubit, and measure it in the computational basis.\\
        Set the first bit of $\theta$ (i.e. $\theta_1$) equal to the result of this measurement.\\
        }
        $ $ With $j=0$, repeat lines 5-8 of the above loop.\\
        $ $ Set the first bit of $\theta$ (i.e. $\theta_0$) equal to the result of the measurement.\\ 
        \caption{Phase Estimation by QSVT}
        \label{alg:PhaseEstimationByQSVT}
    \end{algorithm}

    \begin{figure*}
        \centering
        \includegraphics[width=\textwidth]{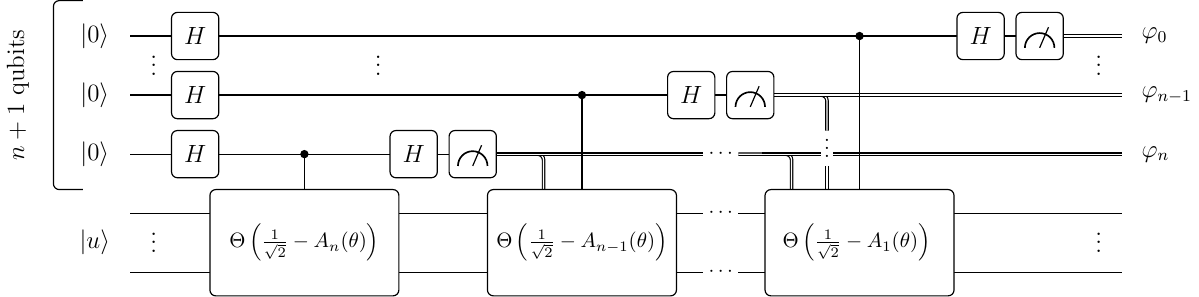}
        \caption{An Abstract overview of the circuit performing phase estimation through QSVT (Algorithm \ref{alg:PhaseEstimationByQSVT}).  The double lines indicate that measurement results are fed forward (via the parameter $\theta$) to all future controlled QSVT operations, where this adaptive process is unpacked in Figure~\ref{fig:qsvt_phase_estimation}. In essence the circuit systematically computes the least significant bit of an unknown quantum phase, adaptively bit shifts this phase, and repeats. Note the similarity of this operation to the inverse quantum Fourier transform, a connection we elaborate on in Section~\ref{sec:EmergentQFT}.}
        \label{fig:abstract_qsvt_circuit}
    \end{figure*}
    
    Note that we have included the additional lines 9-10 in the algorithm to account for the bit in the one's place of $\theta$ (i.e. $\theta_0$), which is needed when $n<m$ such that $\theta$ is an approximation to $\varphi$. Typically, $\theta_0 = 0$, but $\theta_0=1$ is possible in the scenario that $1.0$ is a good approximation to $\varphi$ (for example, rounding $0.11111$ to two binary decimals yields $1.00$). As a result of this additional iteration, we must now require $\epsilon \leq \sqrt{2\delta/(n+1)}$ to ensure that the algorithm succeeds with probability at least $1-\delta$.

    Also observe that Algorithm~\ref{alg:PhaseEstimationByQSVT} consists of $\mathcal{O}(n)$ iterations, each of which performs an instance of QSVT with a degree $\mathcal{O}\left( \frac{1}{\Delta} \log(\frac{1}{\epsilon} ) \right) = \mathcal{O}\left( \log(\frac{n}{\delta}) \right)$ degree polynomial, so phase estimation by QSVT queries the controlled-$U^j$ oracle $\mathcal{O}\big(n \log \big( \frac{n}{\delta} \big) \big)$ times. In addition, the multiplicative factor $1/\Delta$ is not particularly prohibitive because $\Delta = \mathcal{O}(1)$ need only satisfy $\Delta < 2\left(\cos(\frac{3\pi}{16}) - \frac{1}{\sqrt{2}} \right) \approx 0.25$.

    While this $\mathcal{O}(n\log n)$ query complexity may appear to provide a speedup over conventional phase estimation, whose gate count scales as $\mathcal{O}(n^2)$, we note that phase estimation by QSVT requires $\mathcal{O}(n \log n)$ queries to the controlled-$U^j$ oracle, whereas conventional phase estimation requires only $\mathcal{O}(n)$ queries~\cite{nielsen2010quantum}. Nonetheless, we believe that this logarithmic factor could be removed by integrating phase estimation by QSVT with a more sophisticated phase estimation protocol, such as the inference procedures of~\cite{svore2013faster, Wiebe_2016}, which already attain speedups over conventional phase estimation. Indeed, the QSVT-based phase estimation algorithm of \cite{rall2021faster} requires only $\mathcal{O}(n)$ queries to the oracle, which is achieved by using a procedure very similar to Algorithm~\ref{alg:PhaseEstimationByQSVT}, but varying the degree of the QSVT polynomial at each iteration (schematically, by increasing $\Delta$ at each iteration).

    Likewise, phase estimation by QSVT may be modified to be more applicable to specialized scenarios, such as if one is restricted in the polynomials that may be implemented through QSVT, or if one has prior knowledge about $\varphi$. For example, suppose that one cannot implement $\epsilon \leq \sqrt{2 \delta/(n+1)}$ with certainty, and so instead must choose $\epsilon = \mathcal{O}(1)$. To alleviate this difficult, one could repeat the measurement in line 7 $\mathcal{O}(\log( \frac{n}{\delta}))$ times and set $\theta_1$ equal to the majority vote of the measurement results, which will yield an accurate value of $\theta$ with probability $\geq 1-\delta$. Similarly, suppose that the constraint on $\Delta$ cannot be implemented. Then, again using repeated measurements, if $\sigma^{n-1}$ is very close to $\frac{1}{\sqrt{2}}$, the majority of the measurement results may not reflect the correct choice of $\theta_1$. But with high probability, the measurement results will be ambiguous, signaling that $\sigma^{n-1}$ is near $\frac{1}{\sqrt{2}}$. In this event, one could move to an even higher iteration, say some $j>n-1$, where is is likely that $\sigma^j$ is not near $\frac{1}{\sqrt{2}}$ and $\varphi_{j+1}$ is easily determined (and if $\sigma^j$ is again near $\frac{1}{\sqrt{2}}$, one could again move to a higher iteration $j' > j$). With this bit correctly determined, one may proceed through the rest of the iterations and attain a good estimate of $\varphi$.

    \subsection{Applications to Factoring and Beyond}
    Here, we discuss how phase estimation by QSVT may be applied to the factoring problem and used for robust phase estimation. 
        
    \subsubsection{Factoring}
    Phase estimation by QSVT may be straightforwardly applied to the factoring problem. Recall that in the factoring problem, the oracle $U$ behaves as $U|u\rangle = |xu (\text{mod } N)\rangle$ for some $x<N$ and has eigenvalues $e^{2\pi i s/r}$, where $r$ is the order of $U$ (i.e. $x^r = 1 (\text{mod} N)$) and $s \in \{ 0,1, ..., r-1\}$~\cite{nielsen2010quantum}. The goal is to determine some eigenphase $s/r$, from which $r$ may be determined by the continued fractions algorithm and a factor of $N$ may be extracted. In addition, in the factoring problem, one does not have direct access to an eigenstate $|u_s\rangle$, but instead can prepare a uniform superposition of eigenstates $\frac{1}{\sqrt{r}} \sum_s |u_s\rangle$.

    \onecolumngrid
    
    \begin{sidewaysfigure*}        \includegraphics[width=\textwidth]{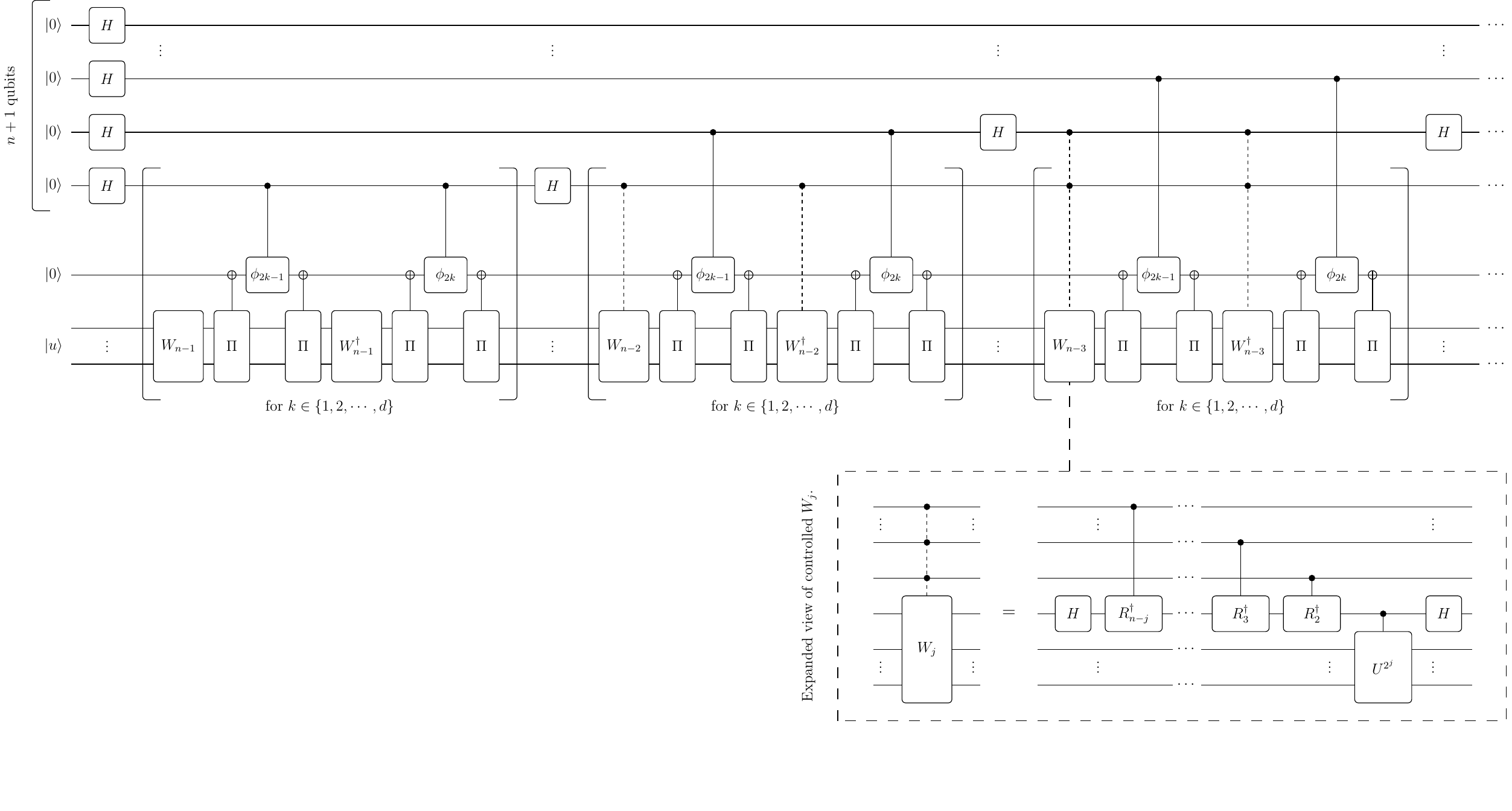}
        \caption{A quantum circuit for performing phase estimation using QSVT, following Algorithm~\ref{alg:PhaseEstimationByQSVT}. The dotted controlled operators, denoted by $W_j$ at step $j$, represent a composite gadget (unwrapped in the inset block) that depends on the states of $(n - j - 2)$ controlling qubits above it. Here controlled-$\phi_j$ gates represent controlled $Z$ rotations by angle $\phi_j$. Additionally, this figure depicts only the even $d$ case of Theorem~\ref{thm:singular_value_transform}. This circuit terminates by measuring all auxiliary qubits in the computational basis.
        Also depicted is an expanded view of a subroutine in the circuit depicted for phase estimation. Note that the $W_j$ operator is the unitary which block-encodes $A_j$ as depicted in Figure~\ref{fig:A_j_BlockEncodingCircuit}. Note also that each $R_{\ell}$ operator is defined as in Eq.~(\ref{eq:R_l}). The controlled $R_{\ell}$ operations serve to adaptively zero out the currently estimated phase to the $(n-j-2)$-th bit.}
        \label{fig:qsvt_phase_estimation}
    \end{sidewaysfigure*}
    
    \twocolumngrid
    \clearpage
    
    \noindent

    First, consider employing the phase estimation by QSVT algorithm, but replacing $|u\rangle$ with the superposition $\frac{1}{\sqrt{r}} \sum_s |u_s\rangle$. With this modification, the phase estimation algorithm begins with the superposition state $\frac{1}{\sqrt{r}} \sum_s |u_s\rangle$ and converges to a single eigenstate at the end of the algorithm. To see this, note that each measurement of an ancilla qubit will restrict the state to be a superposition over eigenstates whose eigenvalues are consistent with the measurement results thus far (i.e. the  eigenphases whose least significant bits agree with the measurement results). Because the eigenvalues are not degenerate, the state at the end of the algorithm will be an eigenstate, say $|u_s\rangle$, whose eigenphase is $\theta \approx s/r$, to which the continued fractions algorithm may be applied to determine $r$.
    
    How accurate must the approximation $\theta \approx s/r$ be? Note that if our estimate $\theta$ obeys
    \begin{equation}
        \left| \frac{s}{r} - \theta \right| \leq \frac{1}{2r^2},
    \end{equation}
    then we can determine the fraction $\frac{s}{r}$ by applying the continued fractions algorithm to $\theta$~\cite{nielsen2010quantum}. We can ensure that this condition is met by selecting $n = 2\log_2(N) + 1$, such that $\left| \frac{s}{r} - \theta \right| \leq 2^{-n} = \frac{1}{2N^2} \leq \frac{1}{2r^2}$, as desired. Thus, subject to these conditions, phase estimation by QSVT may be used to factor $N$ in time $ \mathcal{O}\big(n \log \big( \frac{n}{\delta} \big) \big) = \mathcal{O}\left(\log N \log \left( \frac{\log N}{\delta} \right) \right)$. 
        
    \subsubsection{Robust Phase Estimation}\label{sec:PhaseEstimation_Robust}
    Phase estimation by QSVT comes in handy when $U$ cannot be implemented reliably. For instance, suppose that $U^{2^j}$ can only be implemented approximately, such that the error in $U^{2^j}$ may be interpreted as an additive error in the quantity $0.\varphi_{j+1}\varphi_{j+2} ... \varphi_m$ of Eq.~(\ref{eq:singular_value_reexpression}) (i.e. the erroneous singular value is $\tilde{\sigma}^j = |\cos(\pi(0.\varphi_{j+1} \varphi_{j+2} ... + \varepsilon - \theta))|$, where $\varepsilon$ is the additive error). If $\Delta$ is made sufficiently small, then these errors are not large enough to induce an incorrect measurement result with high probability, and so these errors can be tolerated. 
    
    This robustness to error can be quantified via the analysis in Appendix~\ref{sec:Delta_mitigation}. Formally, if we choose a small enough $\Delta$ such that $\frac{1}{4} - \frac{1}{\pi} \arccos(\frac{1}{\sqrt{2}} +\frac{\Delta}{2}) < \gamma$ for some $\gamma > 0$, then we can tolerate additive errors in $0.\varphi_{j+1}\varphi_{j+2} ...$ of magnitude $< \frac{1}{8} - \frac{3\gamma}{2}$. In this sense, phase estimation by QSVT provides protection against noise of this form, which is a nice improvement over conventional phase estimation.

    \subsection{Emergent Quantum Fourier Transform}\label{sec:EmergentQFT}
    An integral component of Algorithm~\ref{alg:PhaseEstimationByQSVT} is the cascading of multiple QSVT sequences, using the results from one instance to control subsequent instances. In its simplest case, this cascading structure reduces to the celebrated quantum Fourier transform (QFT), and generalizes the QFT in more sophisticated constructions. We elaborate on this key connection in this subsection.

    To spot this connection, observe that circuit of Figure~\ref{fig:abstract_qsvt_circuit} resembles that of the (inverse) quantum Fourier transform, which becomes especially apparent in the limit of a length-one QSVT sequence. From a high level, this arises because Algorithm~\ref{alg:PhaseEstimationByQSVT} may be viewed as a binary search through each bit of $\varphi$, where a bit of $\varphi$ is determined at each iteration. Similarly, the (inverse) quantum Fourier transform may be viewed as a binary search through the bits defining an input state, where each qubit stores a bit of the input state at the end of the computation. Denoting the input state by $|x\rangle = |x_1 x_2 ... x_n\rangle$ where $x=0.x_1x_2 ... x_n$, recall that the action of the inverse quantum Fourier transform is to map 
    \begin{equation}
    \begin{split}
        &\frac{1}{2^{n/2}} \sum_{k=0}^{2^n-1} e^{2\pi i xk}|k\rangle = \frac{1}{2^{n/2}} \left(|0\rangle + e^{2\pi i 0.x_n} |1\rangle\right)\otimes \\ 
        &\left(|0\rangle + e^{2\pi i 0.x_{n-1}x_n} |1\rangle \right) \otimes...\left(|0\rangle + e^{2\pi i 0.x_1 x_2 ...x_n }|1\rangle \right) \ \mapsto \\ & \ \ |x_1 x_2 ... x_n\rangle = |x\rangle. 
    \end{split}
    \end{equation}
    Similar to Algorithm~\ref{alg:PhaseEstimationByQSVT}, this mapping is achieved by applying Hadamards and controlled rotations to extract a bit of $x$ from each qubit, leading to the correspondence seen here.

    Let us explicitly show how the QFT naturally emerges from this construction. For the sake of pedagogical clarity, we will switch to performing QSVT by projecting into the $|0\rangle \langle 0|$ block of the QSVT sequence, as in the $(W_x,S_z,\langle 0| \cdot |0\rangle)$-QSP convention of Appendix~\ref{sec:QSPConventions}, which will make more apparent the connection to the QFT. This change is permissible because the threshold function, an even extension of $\Theta\left( \frac{1}{\sqrt{2}}-x \right)$, can be well approximated by a polynomial in the more restricted class of polynomials of this convention, as presented in Appendix \ref{subsec:phase_est_func}. 
    
    Under these conditions, let us make use of the block encoding of Eq.~(\ref{eq:A_j_BlockEncoding}) (Figure~\ref{fig:A_j_BlockEncodingCircuit}), where $\Pi = \tilde{\Pi} = |0\rangle \langle 0| \otimes I$. Then, a single instance of the signal rotation operator (the block encoding) followed by the signal processing rotation operator (the controlled phase shift), which is iterated in QSVT, may be fully realized as in the circuit on the left side of Figure~\ref{fig:qsp_phaseestonestep}. In this depiction, the top qubit is used to implement the projector controlled phase shift as per Figure~\ref{fig:projector-phase-shift-quantum-circuit}, and the middle qubit is the block encoding qubit used to access the encoding of $A_j(\theta)$ as per Figure~\ref{fig:A_j_BlockEncodingCircuit}. Evaluating this circuit allows simplification to the right side of the figure, where we are now able to ignore the top qubit in favor of just the block encoding qubit, an identity that follows from the simple choice of projector. A block encoding qubit can also double as one of the ancilla qubits used to read out a bit of $\varphi$, as shown below.  Let us refer to this type of qubit as a phase readout qubit.  
    
    \begin{figure}[htbp]
        \centering
        \includegraphics[width=\columnwidth]{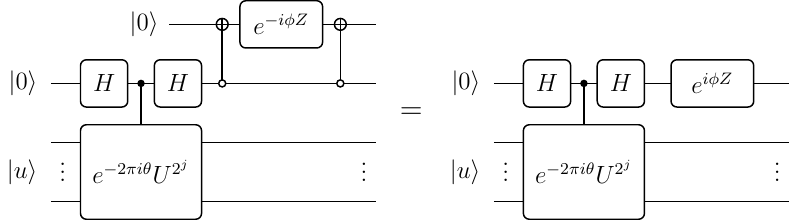}
        \caption{Simplification of QSP phase control for the case when the block-encoded operator has an accessible control qubit.}
        \label{fig:qsp_phaseestonestep}
    \end{figure}
    
    To make this identification explicit and draw a parallel with the QFT, let us also assume that $\varphi$ is an $m=3$ bit number. Then, employing a length $d$ QSVT sequence, the resulting phase estimation circuit may be schematically expressed as in Figure~\ref{fig:qsp_phaseest3}. Here, each phase readout qubit doubles as a block encoding qubit\footnote{Note that we do not include an additional $(m+1)^{\text{th}}$ phase readout qubit for the bit in the ones place of $\varphi$ because no approximation is needed if we know that $m=3$.}, which is crucial to the emergence of the QFT, and is indeed a valid restructuring. In particular, instead of measuring each phase readout qubit, updating $\theta$, and then applying a controlled $\theta$-rotation at the next iteration, as in Algorithm~\ref{alg:PhaseEstimationByQSVT}, here we directly apply rotations controlled by the phase readout qubit, making intermediate QSP measurement steps unnecessary.
    
    We implement this with the rotation operations $R_{\ell}$ defined in Eq.~(\ref{eq:R_l}). Similar to Figure~\ref{fig:qsvt_phase_estimation}, our construction employs controlled $R_{\ell}$'s to implement the following behavior: if a phase readout qubit is $|0\rangle$, then it does not contribute to $\theta$ and no phase shift is applied; alternatively, if the phase readout qubit is $|1\rangle$, then it contributes to $\theta$ and the corresponding phase shift is applied. Moreover, we arrange the $R_{\ell}$'s in a pattern that implicitly performs the $\theta \leftarrow \theta/2$ operation after each iteration; specifically, if ancilla qubit $k$ controls an $R_{\ell}$ at one iteration, then at the next iteration ancilla qubit $k$ controls an $R_{\ell+1}$. It may be easily verified that this procedure correctly encodes $A_j(\theta) = \frac{1}{2}(I+e^{-2\pi i\theta}U^{2^j})$ at each iteration and is entirely equivalent to the formalism of Algorithm \ref{alg:PhaseEstimationByQSVT}, thus justifying the recycling of phase readout qubits as block encoding qubits.

    \begin{figure}[htbp]
        \centering
        \includegraphics[width=\columnwidth]{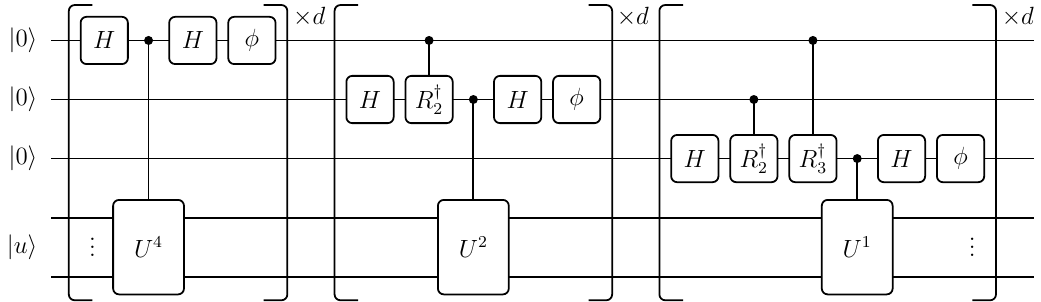}
        \caption{Illustration of three-qubit quantum phase estimation using QSP.}
        \label{fig:qsp_phaseest3}
    \end{figure}

    Further, let the QSVT sequence length be simply $d=1$, giving the quantum circuit depicted in Figure~\ref{fig:qsp_phaseest3_L1}. This allows us to further simplify the circuit, without changing any circuit elements, by observing that the controlled-$U^{2^j}$ operations commute with the $R_{\ell}$ operations, which are $z$-rotations.  Slide the first three Hadamard gates over to the far left, and gather all the remaining gates on the control qubits on the far right, dropping the signal rotation operations $\phi$, which are inconsequential at this level of simplification.
    
    \begin{figure}[htbp]
        \centering
        \includegraphics[width=\columnwidth]{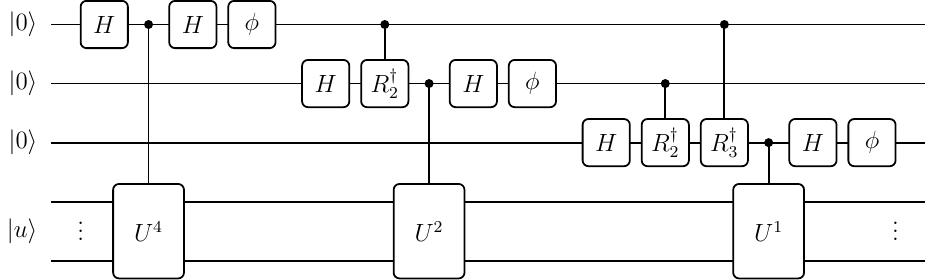}
        \caption{Illustration of three-qubit quantum phase estimation using QSP, in the simplified case when $d=1$.}
        \label{fig:qsp_phaseest3_L1}
    \end{figure}
    
  This systematic simplification results in the quantum circuit shown in Figure~\ref{fig:qsp_phaseest3_L1QFT}, where the inverse QFT circuit emerges naturally as the set of gates following the controlled-$U^{2^j}$ gates.  This is the standard quantum phase estimation circuit.
    
    \begin{figure}[htbp]
        \centering
        \includegraphics[width=\columnwidth]{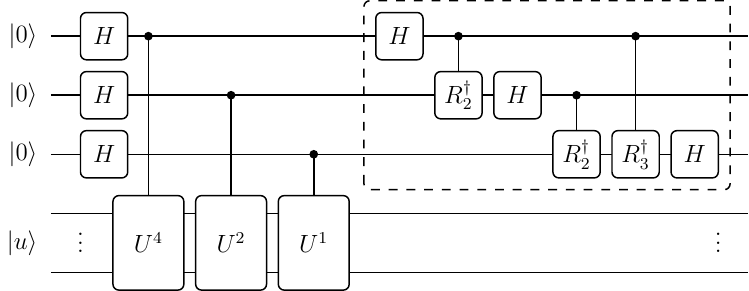}
        \caption{Illustration of three-qubit quantum phase estimation using QSP, in the simplified case when $d=1$, with gates slid along wires to highlight the three-qubit inverse QFT circuit which emerges (gates in dotted rectangular box).}
        \label{fig:qsp_phaseest3_L1QFT}
    \end{figure}
    
    As this pattern persists to higher orders, we see that quantum phase estimation emerges from cascaded QSVT (or QSP) employed to iteratively determine bits of the phase of an eigenphase. The QSVT construction of this transform allows for a richer variety of transformations, however, including a trade-off between the accuracy of each iterative step (by increasing $d$), the number of exponentiated powers $U$ to employ, the number of qubits to use in the transform, and more.

\section{Function Evaluation Problems by QSVT} \label{sec:FuncEval}
    Another useful application of QSVT is to evaluate a function of a matrix, which we term \emph{Function Evaluation Problems}. Schematically, if $f(x)$ is the function of interest, such that we wish to evaluate $f(A)$, then we could imagine solving this problem by employing QSVT with a polynomial $P(x)$ that approximates $f(x)$. While these problems are generally more approachable with a quantum eigenvalue transform, they are still straightforward to solve with QSVT. 
    
    Here we summarize prominent function evaluation problems, most notably Hamiltonian simulation and matrix inversion. Our discussion summarizes results from~\cite{Gily_n_2019}, wherein the full details of these procedures can be found.

\subsection{Hamiltonian Simulation by QSVT}\label{sec:HamiltonianSimulation}
    
    A motivating goal of quantum computation is to simulate the time evolution of a state under a Hamiltonian, a problem known as Hamiltonian simulation. That is, for a Hamiltonian $\mathcal{H}$ and some time $t$, we would like to approximate the time evolution operator $e^{-i\mathcal{H}t}$, which is evidently a function evaluation problem with the function $f(x)=e^{-ixt}$.
    
    In the setup of this problem, we assume access to $\mathcal{H}$, of which we desire a unitary block encoding such that we may solve this problem with QSVT. However, as we discussed in Sec. \ref{sec:Threshold}, such a unitary block encoding is only realizable if $\|\mathcal{H}\| \leq 1$. In general then, we instead determine an $\alpha \geq \| \mathcal{H}\|$ and construct a unitary block encoding of $\mathcal{H}/\alpha$. Again, this requires some prior knowledge about $\mathcal{H}$, a drawback we elaborate on in Sec. \ref{sec:Discussion}, but fortunately such a block encoding can be achieved for a large class of Hamiltonians~\cite{Low_2019, Gily_n_2019}. 
    
    With this rescaled block encoding, one can equivalently imagine that our goal is to simulate the time evolution of a system under the rescaled Hamiltonian $H/\alpha$ for a time $t \alpha$. This equivalence holds because the corresponding time evolution operators are identical: $e^{-i(\mathcal{H}/\alpha) (\alpha t)} = e^{-i\mathcal{H}t}$.
    
    How might this problem be solved with QSVT? Naively, one may try to employ QSVT with a polynomial approximation to $e^{-ixt}$ (here, we view $t$ as a parameter, not a variable). However, because the exponential function does not have definite parity, this function does not satisfy the constraints on Poly($a$)$= \langle + | U_{\vec{\phi}} | +\rangle$ discussed in Section~\ref{sec:qsp} and Appendix \ref{sec:QSPConventions}). To circumvent this issue, one can instead apply QSVT twice - once with an even polynomial approximation to $\cos(xt)$, and once with an odd polynomial approximation to $\sin(xt)$, both of which have definite parities. Then, using the circuit illustrated in Figure~\ref{fig:ComplexExponentialCircuit}, one can sum together the results of these two QSVT executions to obtain $\cos^{(\text{SV})}(\mathcal{H}t) - i\sin^{(\text{SV})}(\mathcal{H}t) = e^{-i\mathcal{H}t}$, as desired. 
    
    However, note that the above relation only holds if the eigenvalues of $\mathcal{H}$ are positive, such that the singular values are equal to the eigenvalues. As we discussed in \ref{sec:Threshold}, if this is not the case, we may instead use the block encoding of $\frac{\mathcal{H}}{\alpha}$ and a circuit analogous to Figure~\ref{fig:A_j_BlockEncodingCircuit} to construct a block encoding of the positive definite matrix $\frac{1}{2}\left(\frac{\mathcal{H}}{\alpha}+I\right)$. Applying the aforementioned QSVT procedure to this matrix for a time $2\alpha t$ produces a block encoding of $e^{-i\mathcal{H}t}$ up to a global phase. In the remainder of this section, we assume that this issue has been alleviated. 
    
    \begin{figure}[htpb]
        \centering
        \includegraphics[width=0.9\columnwidth]{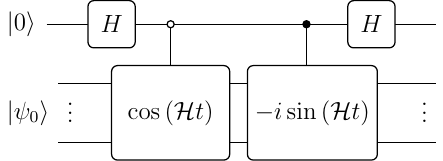}
        \caption{One quantum circuit that can be used to construct the time evolution operator $\cos(\mathcal{H}t) - i\sin(\mathcal{H}t) = e^{-i\mathcal{H}t}$ and apply it to $|\psi_0\rangle$, for use in Hamiltonian simulation as described in Algorithm ~\ref{alg:HamiltonianSimulationByQSVT}. Note that the correct evolution of the input state $| \psi_0 \rangle$ is achieved only upon post-selection of the auxiliary qubit in the $| 0\rangle$ state. One can achieve this with either repetition or fixed-point amplitude amplification, similar to projecting into the desired block of the QSVT sequence operator as discussed at the end of Section~\ref{sec:qsvt}.}
        \label{fig:ComplexExponentialCircuit}
    \end{figure}

    Returning back to the QSVT scheme, we note that in~\cite{Gily_n_2019}, Gily\'en et al. approximate the functions $\cos(xt)$ and $\sin(xt)$ by polynomials using the Jacobi-Anger expansion:
    \begin{align} \label{eq:Jacobi-Anger1}
        \cos(xt) &= J_0(t) + 2\sum_{k=1}^\infty (-1)^k J_{2k}(t) T_{2k}(x) \\
        \label{eq:Jacobi-Anger2}
        \sin(xt) &= 2\sum_{k=0}^\infty (-1)^k J_{2k+1}(t) T_{2k+1}(x)
    \end{align}
    where $J_i(x)$ is a Bessel function of order $i$, and $T_i(x)$ is a Chebyshev polynomial of order $i$. One can attain $\epsilon$-approximations to $\cos(xt)$ and $\sin(xt)$ by truncating these expressions at a sufficiently large index $k'$. The necessary truncation index $k'$ may be determined by a function $r(t,\epsilon)$, which is defined implicitly as
        \begin{equation}
            \epsilon = \left( \frac{|t|}{r(t, \epsilon)} \right)^{r(t, \epsilon)} \text{ s.t. }  r(t, \epsilon) \in (t, \infty),
        \end{equation}
    and scales asymptotically as 
        \begin{equation}
            r(t, \epsilon) = \Theta \Bigg( |t| + \frac{\log(1/\epsilon)}{\log\Big( e+\frac{\log(1/\epsilon)}{|t|} \Big)}\Bigg). 
        \end{equation}
    In particular, truncating Eqs.~(\ref{eq:Jacobi-Anger1}) and~(\ref{eq:Jacobi-Anger2}) at $k'=\big\lfloor \frac{1}{2} r\left(\frac{e}{2} |t|, \frac{5}{4} \epsilon\right) \big\rfloor$ yields $\epsilon$-approximations to $\cos(xt)$ and $\sin(xt)$, respectively, where $0 < \epsilon < 1/e$. Because $T_i(x)$ is a polynomial of degree $i$ with definite parity, these approximations are polynomials of degree $2k'$ and $2k'+1$, respectively, with the correct even/odd parity. Let us denote these polynomials by $P^{\text{cos}}_{\epsilon}(x;t)$ and $P^{\text{sin}}_{\epsilon}(x;t)$.

    Lastly, because cosine and sine are bounded in magnitude by $1$, these $\epsilon$-approximations only obey $|P^{\text{cos}}_{\epsilon}(x;t)|, |P^{\text{sin}}_{\epsilon}(x;t)| \leq 1+\epsilon$. However, we need the target polynomials to necessarily be bounded in magnitude by $1$ in order to be implemented through QSVT. As in~\cite{Low_2017}, we can fix this by rescaling the polynomials by a factor of $\frac{1}{1+\epsilon}$, at the expense of increasing the error of these approximations to $2 \epsilon$. This can be seen with the triangle inequality as $ \left|\frac{1}{1+\epsilon} P^{\text{cos}}_{\epsilon}(x;t) - \cos(xt) \right| \leq \frac{1}{1+\epsilon} \left( \left|P^{\text{cos}}_{\epsilon}(x;t) - \cos(xt)\right| + \left|\epsilon \cos(xt) \right| \right) \leq \frac{1}{1+\epsilon}(\epsilon + \epsilon) \leq 2\epsilon $, and similarly for $P^{\text{sin}}_{\epsilon}(x;t)$. 
    
    To determine the complexity of this Hamiltonian simulation algorithm, first recall that our effective goal is to simulate the rescaled Hamiltonian $\mathcal{H}/\alpha$ for time $\alpha t$. In addition, note that the truncations of the Jacobi-Anger expansion used in this procedure should be $\epsilon/4$-approximate such that, when rescaled by $\frac{1}{1+\epsilon/4}$, they are $\epsilon/2$-approximations to $\cos(xt)$ and $\sin(xt)$. With this choice, the sum of these approximations, which is the approximation to $e^{-ixt}$, is $\epsilon$-approximate by the triangle inequality. Incorporating these conditions, we see that this QSVT-based Hamiltonian simulation algorithm prepares an $\epsilon$-approximate block encoding of $e^{-i\mathcal{H}t}$ and queries $U$ a total number of times
    \begin{equation}
        \begin{split}
            2k' + &2k'+1 = 4 \cdot \Bigg\lfloor \frac{1}{2} r\left(\frac{e}{2} \alpha |t|, \frac{5}{4} \frac{\epsilon}{4} \right) \Bigg\rfloor +1 \\ 
            &= \Theta \Bigg( \alpha|t| + \frac{\log(1/\epsilon)}{\log\Big( e+\frac{\log(1/\epsilon)}{\alpha|t|} \Big)}\Bigg).
        \end{split}
    \end{equation}
    In comparing this query complexity with results quoted in the literature, $\alpha$ may be replaced with $\| \mathcal{H}\|$.
    
    This complexity has state-of-the-art scaling in $t$ and $\epsilon$ for Hamiltonian simulation: it is linear in $t$, logarithmic in $\epsilon$, and additive in these two terms. As such, it provides a significant improvement over other algorithms~\cite{Low_2017, Low_2019, Gily_n_2019}. We summarize Hamiltonian simulation by QSVT in Algorithm~\ref{alg:HamiltonianSimulationByQSVT}.

    $ $ \\
    \begin{algorithm}[htpb] 
    \setstretch{1.05}
     \SetKwInOut{Comp}{Runtime}
     \SetKwInOut{Proc}{Procedure}
     \KwIn{Access to a Hamiltonian $\mathcal{H}$, a desired time $t$, an error tolerance $\epsilon$, and an $\alpha \geq \| \mathcal{H}\|$.}
     \KwOut{A block encoded $\epsilon$-approximation of $e^{-i\mathcal{H}t}$.}
     \Comp{$\Theta \left( \alpha|t| + \frac{\log(1/\epsilon)}{\log\left( e+\frac{\log(1/\epsilon)}{\alpha|t|} \right)}\right)$ queries to (the encoding of) $\mathcal{H}/\alpha$.}
     \Proc{}
     $ $ Prepare a unitary block encoding of $\mathcal{H}/\alpha$.\\
     $ $ Apply QSVT to this encoding twice, using the polynomials $\frac{1}{1+\epsilon/4}P_{\epsilon/4}^{\text{cos}}(x;t)$ and $\frac{1}{1+\epsilon/4}P_{\epsilon/4}^{\text{sin}}(x;t)$, where $P_{\epsilon/4}^{\text{cos}}(x;t)$ and $P_{\epsilon/4}^{\text{sin}}(x;t)$ are obtained by truncating the series in Eqs.~(\ref{eq:Jacobi-Anger1}) and~(\ref{eq:Jacobi-Anger2}), respectively, at index $k'=\big\lfloor \frac{1}{2} r(\frac{e}{2} \alpha |t|, \frac{5}{4} \frac{\epsilon}{4}) \big\rfloor$.\\
     $ $ With the results of the above QSVT executions, which approximate $\cos^{(\text{SV})}(\mathcal{H}t)$ and $\sin^{(\text{SV})}(\mathcal{H}t)$, respectively, run the circuit in Figure~\ref{fig:ComplexExponentialCircuit}. \\  
     \caption{Hamiltonian Simulation by QSVT}
     \label{alg:HamiltonianSimulationByQSVT}
    \end{algorithm}
    $ $ \\

\subsection{Matrix Inversion by QSVT}
\label{sec:matrix_inversion}
    
    Another straightforward, yet widely applicable function evaluation problem is that of \emph{matrix inversion}. That is, given access to a square matrix $A$, one wishes to construct an approximation to $A^{-1}$. Harrow, Hassidim, and Lloyd presented a quantum algorithm for this problem in the case that $A$ is Hermitian~\cite{Harrow_2009}. In their eponymous algorithm, they prepare the state $A^{-1}| b \rangle$, which provides a quantum solution to the linear system $A|x\rangle = | b \rangle$.
    
    Let's now look at this problem through the lens of QSVT. Suppose that we have an $N \times N$ matrix $A$ with singular value decomposition $A = W_\Sigma \Sigma V^\dag_\Sigma$, where $\Sigma$ contains the singular values along its diagonal. As per the setup of HHL algorithm, we also assume that the singular values of $A$ obey $\sigma_i \in [\frac{1}{\kappa}, 1]$ for some finite condition number $\kappa \geq 1$ (if not, $A$ may be rescaled to obey this condition). As the singular values are nonzero, the inverse of $A$ exists and may be obtained as $A^{-1} = V_\Sigma \Sigma^{-1} W^\dag_\Sigma$, where $\Sigma^{-1}$ contains the reciprocals of the singular values along its diagonal. Because $A^\dag = V_\Sigma \Sigma W^\dag_\Sigma$, this can be re-expressed as $A^{-1} = f^{(\text{SV})}(A^\dag)$, where $f(x)=1/x$, indicating that matrix inversion is a function evaluation problem with $f(x) = 1/x$. 
    
    Upon this realization, it is clear how to apply QSVT to matrix inversion: find an odd polynomial $P(x)$ that approximates $f(x) = 1/x$ over the range of singular values of $A$, and employ QSVT to construct $P^{(\text{SV})}(A^\dag)$, which approximates $A^{-1}$. Finding a good polynomial $P(x)$ is tricky because of the discontinuity in $1/x$, but it indeed can be done. In addition, this procedure requires that one can construct a unitary block encoding of $A$, which is feasible because $\|A\| \leq 1$ as per the assumption that $\sigma_i \leq 1$. Such a block encoding can indeed be achieved for a variety of matrices relevant to physics~\cite{Gily_n_2019, Low_2019} (but again, we discuss some caveats of doing so in Sec. \ref{sec:Discussion}).
    
    Moreover, because we require that the polynomial $P(x)$ be bounded in magnitude by $1$ such that it can be implemented through QSVT, we cannot necessarily use $P(x)\approx 1/x$ as our target function since $1/\sigma_i \geq 1$ in general. To overcome this challenge, let us instead seek a polynomial approximation to a function that behaves as $\frac{1}{2\kappa} \frac{1}{x}$ over the range $[-1,1] \setminus [\frac{-1}{\kappa}, \frac{1}{\kappa}]$, which will invert each singular value and is bounded in magnitude by $\frac{1}{2}$ over this range (we use the multiplicative factor $\frac{1}{2\kappa}$ instead of $\frac{1}{\kappa}$ to avoid the need to rescale by $\frac{1}{1+\epsilon}$ when we $\epsilon$-approximate this function, which was done in Section~\ref{sec:HamiltonianSimulation}). This procedure will output an approximation of $\frac{1}{2 \kappa} A^{-1}$, and because $\kappa$ is \emph{a priori} known, this multiplicative factor is not prohibitive to calculations. However, due to this multiplicative factor, we now desire a $\frac{\epsilon}{2 \kappa}$-approximation to $\frac{1}{2 \kappa} A^{-1}$, from which we effectively attain an $\epsilon$-approximation to $A^{-1}$.  
    
    The appropriate polynomial for matrix inversion is thus an $\frac{\epsilon}{2 \kappa}$-approximation to $\frac{1}{2 \kappa} \frac{1}{x}$. In Appendix~\ref{sec:matrix_inversion_polynomial}, we demonstrate how to construct such a polynomial. While the construction is a bit involved, in essence it is a product of a polynomial approximation to $\frac{1}{x}$ over a restricted range (constructed as a sum of Chebyshev polynomials) and a polynomial approximation to a rectangular function (constructed as a sum of the sign function approximations, $P^{\Theta}_{\epsilon, \Delta}(x)$). 
    
    We term this polynomial the \textit{matrix inversion polynomial}, denoted by $P^{\text{MI}}_{\epsilon,\kappa}(x)$, and defer its rigorous definition to in Appendix~\ref{sec:matrix_inversion_polynomial}. In addition to being an $\frac{\epsilon}{2 \kappa}$-approximation to $\frac{1}{2 \kappa} \frac{1}{x}$, $P^{\text{MI}}_{\epsilon,\kappa}(x)$ has odd parity and is bounded in magnitude by $1$ for $x \in [-1,1]$. Hence, the matrix inversion polynomial may be implemented through QSVT.
    
    Moreover, we also show in Appendix~\ref{sec:matrix_inversion_polynomial} that $P^{\text{MI}}_{\epsilon,\kappa}(x)$  has degree 
    \begin{equation}
        d = \mathcal{O}(\kappa \log(\kappa/\epsilon)).
    \end{equation}
    As QSVT requires $\mathcal{O}(d)$ calls to the block encoding, we see that matrix inversion by QSVT has complexity $\mathcal{O}(\kappa \log(\kappa/\epsilon))$. This is an improvement over the conventional HHL algorithm, which has runtime $\mathcal{O}(\kappa^2 \log(N) / \epsilon)$. It is quite impressive that the QSVT algorithm provides an large improvement in the scaling with $\kappa/\epsilon$, although similar results have been achieved with non-QSVT methods \cite{Childs_2017}. In addition, the HHL algorithm uses a sparse Hamiltonian simulation subroutine with target Hamiltonian $A$, resulting in the $\log(N)$ term in its complexity, whereas the QSVT algorithm does not use such a subroutine and thus $N$ dependence is absent from its complexity (however, constructing the necessary block encoding of $A$ may scale with $N$). We summarize matrix inversion by QSVT in Algorithm~\ref{alg:MatrixInversionByQSVT}.
    
    Moreover, like the HHL algorithm, matrix inversion by QSVT may be used to solve the linear system of equations $A|x\rangle = |b\rangle$, by applying the block encoding of $\frac{1}{2\kappa} A^{-1}$ to $|b\rangle$, which yields an $\epsilon$-approximation to $A^{-1} |b\rangle$ upon rescaling by $2\kappa$. As discussed at the end of Section~\ref{sec:qsvt}, this procedure requires that we project into the desired block of the QSVT sequence operator, which may be performed with little overhead. 
    
    Lastly, we note that this result can be further extended. With some minor adjustments, this QSVT-based algorithm can be adapted to prepare the pseudo-inverse of a rectangular matrix, which is useful in various machine learning contexts~\cite{Gily_n_2019}.
    
    \begin{algorithm}[htpb] 
    \setstretch{1.05}
     \SetKwInOut{Comp}{Runtime}
     \SetKwInOut{Proc}{Procedure}
     \KwIn{Access to $A$, an error tolerance $\epsilon$, and a condition number $\kappa \geq 1/(\min_i \sigma_i ) $}
     \KwOut{A block encoded $\frac{\epsilon}{2\kappa}$-approximation of $\frac{\kappa}{2} A^{-1}$, which is effectively equivalent to an $\epsilon$-approximation to $A^{-1}$.}
     \Comp{$\mathcal{O}\left(\kappa \log(\frac{\kappa}{\epsilon}) \right)$ queries to (a block encoding of) $A^\dag$.}
     \Proc{}
     
    $ $ Prepare a unitary block encoding of $A^\dag$.\\
    $ $ Apply QSVT to this block encoding to compute $(P^{\text{MI}}_{\epsilon,\kappa})^{(\text{SV})}(A^\dag)$, where the polynomial $P^{\text{MI}}_{\epsilon,\kappa}(x)$ is defined in Eq.~(\ref{eq:1OverxApprox}) of Appendix~\ref{sec:matrix_inversion_polynomial}.\\
    \caption{Matrix Inversion by QSVT}
    \label{alg:MatrixInversionByQSVT}
    \end{algorithm}

\section{Discussion}\label{sec:Discussion}

    In this paper, we have presented how the quantum singular value transformation encapsulates the three major quantum algorithms. Paralleling~\cite{Gily_n_2019}, we have constructed QSVT-based algorithms for search, Hamiltonian simulation, and the quantum linear systems problem. Toward this end we have also derived QSVT-based algorithms for the eigenvalue threshold problem and phase estimation. Moreover, the utility of QSVT is not entirely enumerated here — further applications of QSVT to the quantum OR lemma, quantum machine learning, quantum walks, fractional query implementation, and Gibbs state preparation appear in the literature~\cite{Gily_n_2019}. 
    
    It is insightful that QSVT encompasses such a broad spectrum of problems. Effectively, QSVT provides a series of dials (i.e., a well-defined parameterization) that can be turned to transform from one algorithm to another. In addition, when there is sufficient structure inherent to the problem of interest, the resulting algorithm often becomes more efficient. Consequently, QSVT provides one lens through which to analyze the source of quantum advantage, and make concrete the somewhat vague tradeoff conjectured between problem structure and quantum algorithmic efficiency (while maintaining a significant gap between optimal classical and quantum performance). In aggregate, these constructions suggest that it is wise to continue to search for new quantum algorithms in the setting of QSVT.
    
    There is ample room for future research in this area. Notably, various quantum algorithms have not yet been constructed from QSVT-based subroutines, such as variational algorithms like the variational quantum eigensolver \cite{peruzzo2014variational} or the quantum approximate optimization algorithm \cite{farhi2014quantum}. It would be fascinating to see if QSVT can encompass, or perhaps even enhance, these hybrid quantum algorithms as well. This work also begets the question of how else one might tweak the parameters of QSVT to create novel algorithms, or extend previously known algorithms to novel noise settings. As QSVT is intuitive and flexible, there is likely a large class of problems that are amenable to analysis by QSVT and admit a significant quantum advantage. 
    
    Moreover, note that there is a significant caveat in the use of QSVT, arising from the requirement of block encodings. In a typical implementation of QSVT on quantum computer, we may imagine that the matrix to which we would like to apply a transform, say $A$, is provided in a quantum random-access memory (``QRAM''), from which we may straightforwardly construct a block encoding of $A/\|A\|_F$~\cite{Gily_n_2019}. One could then apply QSVT to this encoding with a runtime linear in $\|A\|_F$, similar to how the runtimes of the eigenvalue threshold algorithm and the Hamiltonian simulation algorithm of this paper are linear in $\alpha$. However, assuming that one has a classical computer with sampling and query access to $A$, as a classical analog of having $A$ loaded into QRAM, then one can acquire access to a singular value transformation of $A$ by executing a classical (not quantum!) algorithm~\cite{Chia_2020}; this classical algorithm impressively has a runtime polynomial in $\|A\|_F$. This polynomial runtime suggests that QRAM-based QSVT cannot always achieve an exponential speedup over classical algorithms, but can still attain a significant polynomial speedup. Although this challenge could be an Achilles' heel for application of QSVT to general and unstructured problems, clearly QSVT is still of interest for speeding up solution of problems with structure, as illustrated by quantum factoring.  Also, while every QRAM essentially provides a block-encoded matrix, there are ways to block-encode matrices which do not require a QRAM; this is good, given the fact that QRAM implementations generally face a number of practical challenges in their realization, e.g. requiring a number of ancillary qubits which grows linearly with the number of items stored (see, e.g. \cite{HannLeeGirvinJiang2021a} and references therein).
    
    It is also interesting to note that while physics has developed significant insight into the role of eigenvalues, appreciation of singular values has lagged.  For example, eigenvalues are the bread and butter of quantum systems, as energies for eigenstates of the Hamiltonian, and as stability points for stochastic systems.  In contrast, singular values apparently play few starring roles in physics.  One of the few examples arises in the study of entanglement, where the singular value decomposition is the underlying construct behind the Schmidt decomposition of a bipartite quantum state.  
    
    Why are there so few prominent roles for singular values in physics?  Maybe it is because physics is drawn to closed, Hamiltonian systems (think: square matrices), whereas singular value decompositions arise mostly in studies of subsystems (think: non-square matrices), where the input and output dimensions may be different.  As discussed in the introduction, such disparate dimensions also arises naturally in computation.  And indeed, singular value decompositions play a prominent role in modern computation, especially in machine learning, where it is the basis for principal component analysis, model reduction, collaborative filtering, and more.  As quantum information and computer science continue to grow into a unified field, perhaps it is not surprising that singular value decompositions -- and singular value transformations -- are emerging as a unifying bridge.
    
    \textit{Acknowledgements}: The numerical phase computation and algorithms work by AKT was supported by the U.S. Department of Energy, Office of Science, National Quantum Information Science Research Centers, Co-Design Center for Quantum Advantage under contract DE-SC0012704 and the Natural Sciences and Engineering Research Council of Canada (NSERC) [PGSD3-545841-2020].  The algorithm analysis work by ZMR was supported in part by the NSF EPIQC program.
    ILC was supported in part by the NSF Center for Ultracold Atoms.  ZMR and ILC were also supported in part by ARO contract W911NF-17-1-0433. 


\appendix
\section{QSP Conventions} \label{sec:QSPConventions}

A quantum signal processing construction may be entirely determined by four constituents:
\begin{enumerate}
\item The {\em signal operator} $W$ (sometimes also called the signal unitary)
\item The {\em phase angles} $\vec{\phi} = (\phi_0, \phi_1, \ldots, \phi_d)$
\item The {\em signal processing} operator $S(\phi)$, constructing using $\phi\in \vec{\phi}$
\item The {\em signal basis} $M$, in which the desired polynomial is obtained
\end{enumerate}
In its most basic form, QSP is performed by interleaving $W$ with $S(\phi)$ operations followed by a projective measurement in the basis $M$. 

The signal operator $W$ is signal-dependent and constant throughout the sequence;
the signal processing operator $S$ is parameterized by a sequence of phases $\vec\phi \in \mathbb{R}^{d+1}$, which are chosen based on the desired output function.
The exact form of $W$ and $S$, along with the choice of measurement basis $M$, determine the family of achievable output functions.  More specifically, $M$ may be not a basis for a complete vector space, but just the basis for a subspace identifying a specific desired polynomial output function.  For example, we may specify $M=\<+|\cdot |+\>$, when the measurement is to select the $|+\>$ outcome, with the QSP sequence starting off with the control qubit in the $|+\>$ state.  Thus, we may refer to any particular QSP construction as being a ($W$, $S$, $M$)-QSP convention.

In this section, we discuss common ($W$, $S$, $M$)-QSP conventions in the literature for $W$, $S$, and $M$, and present relationships between the QSP phase angles $\vec{\phi}$.  Specifically, we elaborate on the ``Wx'', reflection, and ``Wz'' conventions.

\subsection{Wx convention for QSP}

One common convention is to choose the signal operator $W$ to be an $x$-rotation in the Bloch sphere,
\be
W_x(\theta) := e^{i \frac{\theta}{2} X} = \mattwocb{a}{i \sqrt{1-a^2}}{i \sqrt{1-a^2}}{a}
\,,
\label{eq:wa_qsp_x}
\ee
where, compared with Equation~\ref{eq:wa_qsp}, we introduce the additional $x$ subscript for clarity. Under this convention, one also chooses the signal processing rotation to be a $z$-rotation, $S_z(\phi) = e^{i \phi Z}$.
Theorem~\ref{thm:qsp} describes the family of unitaries achievable in this convention.

Typically, we are not concerned with the achievable unitaries, but rather the achievable functions that can be computed in a subsystem.
If we choose to project out the $\braket{0 | \cdot | 0}$ element, the output is ${\rm Poly}(a) = \braket{0 | U_{\vec\phi} | 0} = P(a)$.
Here, the choice of ${\rm Poly}$ is restricted not only by the conditions on $P$ of Theorem~\ref{thm:qsp}, but also the additional constraints below which ensure that a polynomial $Q$ exists satisfying the conditions of Theorem~\ref{thm:qsp}.

\begin{theorem}[$(W_x, S_z, \braket{0| \cdot | 0})$-QSP]
A QSP phase sequence $\vec\phi \in \mathbb{R}^{d+1}$ exists
\begin{equation}\label{eq:qsp_thm_Wx-Sz-00}
    {\rm Poly}(a) = \braket{0 | e^{i \phi_0 Z} \prod_{k=1}^d  W_x(a)  e^{i\phi_k Z} | 0}
\end{equation}
for $a\in[-1, 1]$, and for any polynomial ${\rm Poly} \in \mathbb{C}[a]$ if and only if the following conditions hold:
\begin{enumerate}[label=(\roman*)]
  \item ${\rm deg}({\rm Poly}) \leq d$
  \item ${\rm Poly}$ has parity $d~ {\rm mod}~ 2$
  \item $\forall a \in [-1, 1]$, $|{\rm Poly}(a)| \le 1$
  \item $\forall a \in (-\infty, -1] \cup [1, \infty)$, $|{\rm Poly}(a)| \ge 1$
  \item if $d$ is even, then $\forall a \in \mathbb{R}$, ${\rm Poly}(ia){\rm Poly}^*(ia) \ge 1$
\end{enumerate}
\label{thm:qsp_Wx-Sz-00}
\end{theorem}
The family of achievable polynomials in this case is significantly limited since the projective measurement is performed in the same basis $M = \{|0\>, |1\>\}$ as the signal processing operations.
For example, this immediately limits us to polynomial functions such that $|{\rm Poly}(\pm 1)|~=~1$.

Properties $(iv)$ and $(v)$, are quite restrictive especially when approximating real functions.
In fact, none of polynomial approximations described in this paper satisfy these last two constraints.
Allowing a small but non-zero imaginary part, it is possible to find QSP phases in the $(W_x, S_z, \braket{0| \cdot | 0})$ convention approximating real functions that satisfy $|f(\pm 1)| = 1$;
an example of this is the phase estimation function plotted in Section~\ref{subsec:phase_est_func}.

Since we are often interested in constructing real polynomials, choosing a different signal basis ends up being much more useful.   In particular, when $M=\{|+\>, |-\>\}$, then we may employ this theorem for real polynomials:

\begin{theorem}[$(W_x, S_z, \braket{+| \cdot | +})$-QSP]
A QSP phase sequence $\vec\phi \in \mathbb{R}^{d+1}$ exists such that
\begin{equation}\label{eq:qsp_thm_Wx-Sz-++}
    {\rm Poly}(a) = \braket{+ | e^{i \phi_0 Z} \prod_{k=1}^d  W_x(a)  e^{i\phi_k Z} | +}
\end{equation}
for $a\in[-1, 1]$, and for any real polynomial ${\rm Poly} \in \mathbb{R}[a]$ if and only if the following conditions hold:
\begin{enumerate}[label=(\roman*)]
  \item ${\rm deg}({\rm Poly}) \leq d$
  \item ${\rm Poly}$ has parity $d~ {\rm mod}~ 2$
  \item $\forall a \in [-1, 1]$, $|{\rm Poly}(a)| \le 1$
\end{enumerate}
\label{thm:qsp_Wx-Sz-++}
\end{theorem}
This QSP convention is expressive enough for all of the polynomials considered in this paper and is used in~\cite{Low_2016, Gily_n_2019}.
The proofs of Theorem~\ref{thm:qsp_Wx-Sz-00} and Theorem~\ref{thm:qsp_Wx-Sz-++} are given in~\cite{Gily_n_2019}.

\subsection{Reflection convention for QSP}

Another common convention is to choose the signal operator $W$ to be a reflection (as in Eq.~(\ref{eq:ampamp_rot})).
\be
R(a) := \mattwocb{a}{\sqrt{1-a^2}}{\sqrt{1-a^2}}{-a}
\,.
\label{eq:ra_qsp_x}
\ee
The reflection operator is preferred in some cases as it has the added benefit of being Hermitian, which can simplify proof constructions, and in particular equations such as Eq.~(\ref{eq:sigproc_tran_b}). Given a phase sequence $\vec{\phi} \in \mathbb{R}^{d+1}$, we can find a $\vec{\phi}' \in \mathbb{R}^{d+1}$ such that
\be
e^{i \phi_0 Z} \prod_{k=1}^d R(a) e^{i \phi_k} = e^{i \phi_0' Z} \prod_{k=1}^d W_x(a) e^{i \phi_k'}
\,.
\ee
Using the relationship of Eq.~(\ref{eq:wa_to_ra}), this can be accomplished by choosing $\phi_0 = \phi_0' + (2d - 1)\frac{\pi}{4}$, $\phi_d = \phi_d' - \frac{\pi}{4}$, and $\phi_k = \phi_k' - \frac{\pi}{2}$ for $k \in {1, \ldots, d-1}$.
Therefore, these two conventions are equivalent regardless of the final measurement basis. 

\subsection{Wz convention for QSP}

Theorem~\ref{thm:qsp_Wx-Sz-++} can also be understood through its relationship to the convention used in~\cite{chao2020finding}.
Here the authors define QSP with the signal operator $W$ being a $z$-rotation,
\bea
W_z(\theta) = e^{i \frac{\theta}{2} Z} = \mattwocb{w}{0}{0}{w^{-1}}
\,,
\label{eq:wa_qsp_z}
\eea
where $w := e^{i\theta / 2}$, and the signal processing operator is an $x$-rotation, $S_x(\phi) = e^{i \phi X}$.  Furthermore, in this convention, it is typical to choose the signal basis as being $M=\{|0\>, |1\>\}$.

In our notation, this convention is written as $(W_z, S_x, \braket{0| \cdot |0})$-QSP.
This convention is equivalent to $(W_x, S_z, \braket{+| \cdot |+})$-QSP and is equally expressive which can easily be seen since
\begin{equation}\label{eq:Wx-Wz-equivalence}
    \begin{gathered}
        \braket{+| e^{i \phi_0 Z} \prod_{k=1}^d  W_x(\theta)  e^{i\phi_k Z} |+} = \\
        \braket{0| e^{i \phi_0 X} \prod_{k=1}^d  W_z(\theta)  e^{i\phi_k X} |0} 
    \end{gathered}
\end{equation}

The Laurent polynomial formulation of $(W_z, S_x, \braket{0| \cdot |0})$-QSP of~\cite{chao2020finding} can be related to our formulation as follows
\begin{eqnarray}
    &&e^{i \phi_0 X} \prod_{k=1}^d  W_z(\theta)  e^{i\phi_k X} \\
    \label{eq:fg_matrix}
    &=&: \mattwocb{F(w)}{i G(w)}{i G(w^{-1})}{F(w^{-1})} \\
    &=& H e^{i \phi_0 Z} \prod_{k=1}^d  W_x(\theta)  e^{i\phi_k Z} H \\
    &=& H \mattwocb{P(a)}{i Q(a) \sqrt{1 - a^2}}{i Q^*(a) \sqrt{1 - a^2}}{P^*(a)} H
\end{eqnarray}
for complex polynomials $P, Q \in \mathbb{C}[a]$ and real Laurent polynomials $F, G \in \mathbb{R}[w, w^{-1}]$ with parity $d \mod 2$.

Explicitly, 
\begin{eqnarray}
    f_0 = \Re[p_0],
    &\quad&
    g_0 = \Im[p_0]
\end{eqnarray}
and for $k > 0$,
\begin{eqnarray}
    f_k = \frac{1}{2} \Re[p_k + q_k],
    &\qquad&
    f_{-k} = \frac{1}{2} \Re[p_k - q_k]
    \\
    g_k = \frac{1}{2} \Im[p_k - q_k],
    &\qquad&
    g_{-k} = \frac{1}{2} \Im[p_k + q_k]
\end{eqnarray}
where the coefficients for $P$ and $Q$ are given in the Chebyshev bases
\begin{eqnarray}
    P(a) &\equiv& \sum_{k=0}^d p_k T_k(a) \\
    Q(a) &\equiv& \sum_{k=0}^d q_k U_{k-1}(a)
\end{eqnarray}
and the coefficients for the $F$ and $G$ are given in the standard basis
\begin{eqnarray}
    F(w) \equiv \sum_{k=-d}^d f_k w^{k},
    &\qquad&
    G(w) \equiv \sum_{k=-d}^d g_k w^{k}
\end{eqnarray}

Requiring unitarity in Eq.~(\ref{eq:fg_matrix}) is equivalent to 
\be
F(w)F(w^{-1}) + G(w)G(w^{-1}) = 1
\,.
\ee
A numerically stable method for computing phases for a given $F(w)$ in the $(W_z, S_x, \braket{0| \cdot |0})$ convention is discussed further in~\cite{chao2020finding} and can also be seen as a constructive proof of Theorem~\ref{thm:qsp_Wx-Sz-++}.


\section{Proofs about Phase Estimation by QSVT} \label{sec:PhaseEstimationProofs}
In this appendix, we prove the theorems used in the development of phase estimation by QSVT in Section~\ref{sec:PhaseEstimation}.

\subsection{Theorems~\ref{thm:n_geq_m} and~\ref{thm:n_l_m}} \label{sec:PhaseEstimation_m_n_Proofs}

Here, we prove the Theorems~\ref{thm:n_geq_m} and~\ref{thm:n_l_m} from Section~\ref{sec:PhaseEstimation_Sketch}:
\begingroup
\def\thetheorem{\ref{thm:n_geq_m}}
    \begin{theorem}
    If $n \geq m$, and one has the ability to implement the sign function exactly through QSVT, then at the end of the algorithm sketch of Section~\ref{sec:PhaseEstimation_Sketch}, $\theta = \varphi$. 
    \end{theorem}
\addtocounter{theorem}{-1}
\endgroup
and
\begingroup
\def\thetheorem{\ref{thm:n_l_m}}
    \begin{theorem}
    If $n < m$ (including the case $m=\infty$), and one has the ability to implement the sign function exactly through QSVT, then at the end of the algorithm sketch of Section~\ref{sec:PhaseEstimation_Sketch} (with a minor modification discussed below), $|\theta - \varphi| \leq 2^{-n-1}$. 
    \end{theorem}
\addtocounter{theorem}{-1}
\endgroup
\noindent The minor modification needed for $n<m$ is an additional $j=0$ iteration used in the complete phase estimation by QSVT algorithm.

In the setup of these proofs, we suppose that $\varphi$ is an $m$-bit number and use the notation $0.\varphi_{[j:]} := 0.\varphi_{j}\varphi_{j+1}...\varphi_{m}$ to denote a contiguous string of binary digits (this string being infinite in the case $m=\infty$). We also assume that we can implement the sign function exactly in QSVT (which is unrealistic, but was addressed in Section~\ref{sec:PhaseEstimation_Caveats}). Hence, the procedure that we discuss in these proofs is effectively Algorithm $\ref{alg:PhaseEstimationByQSVT}$ modulo the error analysis.

\subsubsection{$n \geq m$}\label{sec:n_geq_m}
If $n \geq m$, then append $\varphi$ with $m-n$ 0's after its $m^{\text{th}}$ binary decimal, such that $\varphi = 0.\varphi_1 \varphi_2 ... \varphi_m 0 0... 0$. This effectively makes $\varphi$ an $n$-bit number without changing its numerical value. We now invoke the following lemma, which is proven by induction:

\begin{lemma} \label{lemma:n_geq_m}
If $\varphi = 0.\varphi_1\varphi_2...\varphi_n$ is an n-bit number, then at the end of the iteration $j\geq 0$, $\theta = 0.\varphi_{j+1}\varphi_{j+2}...\varphi_{n} = 0.\varphi_{[j+1:]}$. 
\end{lemma}
\begin{proof}
The proof proceeds by induction. 

\textit{Base Case}: The base case is the $j=n-1$ iteration, at which $\theta = 0$. Using Eq.~(\ref{eq:singular_value_reexpression}), we see that $A_j(\theta) = A_{n-1}(0)$ has singular value
\begin{equation}
\begin{gathered}
    \sigma^{n-1} = \big|\cos\big(\pi (0.\varphi_{n}-0) \big) \big| =
    \begin{cases}
        1 & \text{if } \varphi_n = 0 \\
        0 & \text{if } \varphi_n = 1,
    \end{cases}
\end{gathered}
\end{equation}
and so $\Theta\Big(\frac{1}{\sqrt{2}}-\sigma^{n-1}\Big) = 1-2\varphi_n$. Therefore, after we apply $\Theta^{(\text{SV})}\Big(\frac{1}{\sqrt{2}}-A_{n-1}(0)\Big)$ to $|u\rangle$, controlled on the ancilla qubit, the state of the ancilla qubit is $\frac{1}{\sqrt{2}}\Big(|0\rangle + (1-2\varphi_n)|1\rangle \Big)$. The final Hadamard gate brings the ancilla qubit to the state
\begin{equation}
\begin{gathered}
    (1-\varphi_n)|0\rangle + \varphi_n |1\rangle = 
    \begin{cases}
        |0\rangle & \text{if } \varphi_n = 0 \\
        |1\rangle & \text{if } \varphi_n = 1
    \end{cases},
\end{gathered}
\end{equation}
and therefore the measurement of the ancilla qubit will yield $\varphi_n$. The final step of the iteration will set $\theta_1 = \varphi_n$, such that $\theta = 0.\varphi_n$, as claimed.

\textit{General Case:} Proceeding with induction, assume that this claim is true at the end of iteration $j+1$. That is, $\theta = 0.\varphi_{j+2} \varphi_{j+3} ... \varphi_n = 0.\varphi_{[j+2:]}$ at the end of iteration $j+1$. Then, at the start of iteration $j$, $\theta \leftarrow \theta/2 = 0.0\varphi_{j+2} \varphi_{j+3} ... \varphi_n$. Again using Eq.~(\ref{eq:singular_value_reexpression}), we see that $A_j(\theta)$ has singular value
\begin{equation}
\begin{split}
    \sigma^{j} &= \big|\cos\big(\pi (0.\varphi_{j+1} \varphi_{j+2} ... \varphi_n-0.0\varphi_{j+2} \varphi_{j+3} ... \varphi_n) \big) \big|\\ &=
    \Big| \cos\Big( \frac{\pi}{2}\varphi_{j+1} \Big) \Big|= 
    \begin{cases}
        1 & \text{if } \varphi_{j+1} = 0 \\
        0 & \text{if } \varphi_{j+1} = 1
    \end{cases},
\end{split}
\end{equation}
and so $\Theta\Big(\frac{1}{\sqrt{2}}-\sigma^{j}\Big) = 1-2\varphi_{j+1}$. This is identical to the base case, and we ultimately find that the measurement of the auxiliary qubit will yield $\varphi_{j+1}$, such that we set $\theta_1 = \varphi_{j+1}$. So at the end of this iteration, we have $\theta = 0.\varphi_{j+1}\varphi_{j+2}...\varphi_{n} = \varphi_{[j+1:]}$, as desired. 
\end{proof}

Considering the $j=0$ iteration, we see that the final output of the algorithm is $\theta = 0.\varphi_1 \varphi_2 ... \varphi_n = \varphi$, which completes the proof of Theorem~\ref{thm:n_geq_m}.

\subsubsection{$n<m$}\label{sec:n_l_m}
If $n<m$, this algorithm will produce an $n$-bit approximation to $\varphi$ that suffers error at most $\frac{1}{2^{n+1}}$. To show this, we first prove the following lemma.

\begin{lemma}\label{lemma:n_l_m}
If $\varphi$ is an $m>n$ bit number (including the case $m=\infty$), then at the end of iteration $j\geq 0$, $\theta = 0.\tilde{\varphi}_{j+1}\tilde{\varphi}_{j+2}...\tilde{\varphi}_{n} = 0.\tilde{\varphi}_{[j+1:]}$, such that $|\tilde{r}_j.\tilde{\varphi}_{[j+1:]} - 0.\varphi_{[j+1:]}| \leq 2^{j-n-1}$, where $\tilde{r}_j. \tilde{\varphi}_{[j+1:]}$ is attained by rounding $0.\varphi_{[j+1:]}$ to $n-j$ binary decimals, and $\tilde{r}^j$ is an additional bit carried over by this rounding. 
\end{lemma}
\begin{proof}
As before, the proof proceeds by induction, albeit longer than the $n\geq m$ proof, but no more complicated. One minor difference is that we will still write expressions like $\varphi = 0.\varphi_1 \varphi_2 ... \varphi_m$ to indicate the binary expansion of $\varphi$, where it is to be understood that this will really be an infinite sequence in the case $m=\infty$. 

\textit{Base Case:} The base case is the $j=n-1$ iteration, at which, $\theta = 0$. Using Eq.~(\ref{eq:singular_value_reexpression}), we see that $A_j(\theta) = A_{n-1}(0)$ has singular value
\begin{equation}\label{eq:s_n1_expression}
\begin{split}
    \sigma^{n-1} &= \big|\cos\big(\pi (0.\varphi_{n}\varphi_{n+1}...\varphi_{m}-0) \big) \big|\\ &=
    \Big| \cos\Big( \frac{\pi}{2}\big(\varphi_{n} + \frac{1}{2}\varphi_{n+1} + \sum_{i=2}^{m-n} \frac{1}{2^{i}} \varphi_{n+i} \big) \Big) \Big|.
\end{split}
\end{equation}
We are not concerned with this exact value, but instead with the quantity $\Theta\big(\frac{1}{\sqrt{2}} - \sigma^{n-1} \big)$, which is dictated by the values of $\varphi_n$ and $\varphi_{n+1}$. In particular, the sum in Eq.~(\ref{eq:s_n1_expression}) is bounded above by $\frac{1}{2}$, so it cannot affect $\Theta\big( \frac{1}{\sqrt{2}} - \sigma^{n-1}\big)$ by itself. Let's look at each possible value of $\varphi_n$ and $\varphi_{n+1}$. 

First, if $\varphi_{n+1} = 0$ (regardless of the value of $\varphi_n$), then this scenario is similar to the proof of Theorem~\ref{thm:n_geq_m}, and we ultimately set $\theta_1 = \varphi_{n}$. In addition, rounding $0.\varphi_{n}\varphi_{n+1}...\varphi_{m} = 0.\varphi_{n}0\varphi_{n+2}...\varphi_{m}$ to $n-j = 1$ decimal places yields $\tilde{r}_{n-1}.\tilde{\varphi}_n = 0.\varphi_{n}$, so indeed $\theta = 0.\varphi_{n} = 0.\tilde{\varphi}_n$. We then have that 
\begin{equation}
\begin{split}
    |\tilde{r}_{j-1}.\tilde{\varphi}_j - 0.\varphi_{[j+1:]}| = \sum_{i=2}^{m-n}&\frac{1}{2^{i+1}}\varphi_{n+i} \\ &\leq \frac{1}{4} = 2^{j-n-1},
\end{split}
\end{equation}
as desired.

On the other hand, if $\varphi_{n+1} = 1$ and $\varphi_n=0$, then $\Theta\big(\frac{1}{\sqrt{2}}-\sigma^{n-1}\big) = -1 $, which results in setting $\theta_1 = 1$. Likewise, rounding $0.\varphi_{n}\varphi_{n+1}...\varphi_{m} = 0.01\varphi_{n+2}...\varphi_{m}$ to $n-j = 1$ decimal places yields $\tilde{r}_{n-1}.\tilde{\varphi}_n = 0.1 $, so indeed $\theta = 0.1 = 0.\tilde{\varphi}_n$. We then have that 
\begin{equation}
\begin{split}
    &|\tilde{r}_{j}.\tilde{\varphi}_{j+1} - 0.\varphi_{[j+1:]}| = \Big|\frac{1}{2} - \frac{1}{4} - \sum_{i=2}^{m-n}\frac{1}{2^{i+1}}\varphi_{n+i} \Big|= \\ &\Big|\frac{1}{4} - \sum_{i=2}^{m-n}\frac{1}{2^{i+1}}\varphi_{n+i} \Big| \leq \frac{1}{4} = 2^{j-n-1},
\end{split}
\end{equation}
as desired.

Lastly, if $\varphi_{n+1} = 1$ and $\varphi_n=1$, then $\Theta\big(\frac{1}{\sqrt{2}}-\sigma^{n-1}\big) = 1 $, which results in setting $\theta_1 = 0$. In addition, rounding $0.\varphi_{n}\varphi_{n+1}...\varphi_{m} = 0.11\varphi_{n+2}...\varphi_{m}$ to $n-j = 1$ decimals yields $\tilde{r}_{n-1}.\tilde{\varphi}_n = 1.0$, so indeed $\theta = 0.0 = 0.\tilde{\varphi}_n$. We then have that 
\begin{equation}
\begin{split}
    &|\tilde{r}_{j}.\tilde{\varphi}_{j+1} - 0.\varphi_{[j+1:]}| = \Big|1 - \frac{1}{2} - \frac{1}{4} - \sum_{i=2}^{m-n}\frac{1}{2^{i+1}}\varphi_{n+i} \Big| \\ &\Big|\frac{1}{4} - \sum_{i=2}^{m-n}\frac{1}{2^{i+1}}\varphi_{n+i} \Big| \leq \frac{1}{4} = 2^{j-n-1},
\end{split}
\end{equation}
as desired.

\textit{General Case:} Proceeding now to the general case, assume that this claim is true at the end of iteration $j+1$ and all prior iterations. That is, $\theta = 0.\tilde{\varphi}_{j+2}\tilde{\varphi}_{j+3}...\tilde{\varphi}_n$ at the end of iteration $j+1$. Then, at the start of iteration $j$, $\theta \leftarrow \theta/2 = 0.0\tilde{\varphi}_{j+2} \tilde{\varphi}_{j+3} ... \tilde{\varphi}_n$. Again using Eq.~(\ref{eq:singular_value_reexpression}), we see that $A_j(\theta)$ has singular value
\begin{equation}
\begin{split}
    \sigma^{j} =& \big|\cos\big(\pi (0.\varphi_{j+1}\varphi_{j+2}...\varphi_{m}-\theta) \big) \big|= \\ 
    &\Bigg| \cos\Bigg( \frac{\pi}{2}\Big(\varphi_{j+1} + \frac{1}{2}(\varphi_{j+2}-\tilde{\varphi}_{j+2})+ \\ 
    \sum_{i=3}^{n-j} \frac{1}{2^{i-1}}& (\varphi_{j+i}-\tilde{\varphi}_{j+i}) + \sum_{i=n-j+1}^{m-j} \frac{1}{2^{i-1}} \varphi_{n+i} \Big) \Bigg) \Bigg|.
\end{split}
\end{equation}
As before, we are interested in $\Theta \big(\frac{1}{\sqrt{2}} - \sigma^j \big)$, which, as we will see below, is dictated by the values of $\varphi_{j+1}, \ \varphi_{j+2}$, and $\tilde{\varphi}_{j+2}$. Let's now look at the possible cases of these values.

First, if $\tilde{\varphi}_{j+2} =\varphi_{j+2}$, then $\Theta\big(\frac{1}{\sqrt{2}}- \sigma^j \big)$ is dictated entirely by $\varphi_{j+1}$, and so we ultimately set $\theta_1 = \varphi_{j+1}$. In addition, the condition $\tilde{\varphi}_{j+2} =\varphi_{j+2}$ implies that no rounding occurred in the $j+1$ iteration, so $\tilde{r}_{j+1} = 0$, and $\tilde{r}_j.\tilde{\varphi}_{[j+1:]} = 0.\varphi_{j+1}\varphi_{j+2}\tilde{\varphi}_{j+3} ... \tilde{\varphi}_n$ is the rounded value of $0.\varphi_{[j+1:]}$. Therefore, we indeed have $\theta = 0.\tilde{\varphi}_{j+1} \tilde{\varphi}_{j+2} ... \tilde{\varphi}_{n}$, and also 
\begin{equation}
\begin{split}
    |\tilde{r}_j . \tilde{\varphi}_{j+1} - 0.\varphi_{j+1}| &= \frac{1}{2}|\tilde{r}_{j+1}.\tilde{\varphi}_{j+2} - 0.\varphi_{j+2}| \\ &\leq \frac{1}{2}2^{(j+1)-n-1} = 2^{j-n-1},
\end{split}
\end{equation}
where we have used the inductive hypothesis.

On the other hand, consider $\tilde{\varphi}_{j+2} = 1$ and $\varphi_{j+2} = 0$. This condition implies that, at the $j+1$ iteration, $0.\varphi_{[j+2:]} = 0.0\varphi_{j+3} ... \varphi_m$ was rounded to $\tilde{r}_{j+1}.\tilde{\varphi}_{[j+2:]} = 0.1\tilde{\varphi}_{j+3} ... \tilde{\varphi}_{n} $, which requires that $\varphi_{j+3} = 1 = \varphi_{j+4} = ... \varphi_{n+1}$ and $\tilde{\varphi}_{j+3} = 0 = \tilde{\varphi}_{j+4} = ... = \tilde{\varphi}_{n}$. Using these values, we see that $\tilde{r}_j.\tilde{\varphi}_{[j+1:]} = 0.\varphi_{j+1}\tilde{\varphi}_{j+2} ... \tilde{\varphi}_n = 0.\varphi_{j+1} 100...0$ is the rounded value of $\varphi^j$. In addition, inputting all of these values into $\Theta\big(\frac{1}{\sqrt{2}} - \sigma^j \big)$, we find that we will ultimately set $\theta_1 = \varphi_{j+1}$. Therefore, we indeed have $\theta = 0.\tilde{\varphi}_{j+1} \tilde{\varphi}_{j+2} ... \tilde{\varphi}_{n}$, and again 
\begin{equation}
\begin{split}
    |\tilde{r}_j . \tilde{\varphi}_{j+1} - 0.\varphi_{j+1}| &= \frac{1}{2}|\tilde{r}_{j+1}.\tilde{\varphi}_{j+2} - 0.\varphi_{j+2}| \\ &\leq \frac{1}{2}2^{(j+1)-n-1} = 2^{j-n-1},
\end{split}
\end{equation}
where we again used the inductive hypothesis.

Finally, consider $\tilde{\varphi}_{j+2} = 0$ and $\varphi_{j+2} = 1$. This condition implies that, at the $j+1$ iteration, $0.\varphi_{[j+2:]} = 0.1\varphi_{j+3} ... \varphi_m$ was rounded to $\tilde{r}_{j+1}.\tilde{\varphi}_{[j+2:]} = 1.0\tilde{\varphi}_{j+3} ... \tilde{\varphi}_{n} $, which requires that $\varphi_{j+3} = 1 = \varphi_{j+4} = ... \varphi_{n+1}$ and $\tilde{\varphi}_{j+3} = 0 = \tilde{\varphi}_{j+4} = ... \tilde{\varphi}_{n}$. Thus, if $\varphi_{j+1} = 0$, then $\tilde{r}_j. \tilde{\varphi}_{[j+1:]} = 0.100....0$ is the rounded value of $0.\varphi_{[j+1:]} = 0.011...1\varphi_{n+2}...\varphi_m$. Inputting these values into $\Theta\big(\frac{1}{\sqrt{2}} - \sigma^j \big)$, we find that we will ultimately set $\theta_1 = 1 =\tilde{\varphi}_{j+1} =  1-\varphi_{j+1}$. Therefore, we indeed have $\theta = 0.\tilde{\varphi}_{j+1} \tilde{\varphi}_{j+2} ... \tilde{\varphi}_{n}$, and also 
\begin{equation}
\begin{split}
    \big|\tilde{r}_j . \tilde{\varphi}_{j+1} &- 0.\varphi_{j+1}\big| = \big| 0.100...0 - 0.011...1\varphi_{n+2} ... \varphi_{m} \big| \\ 
    &= \Big|\frac{1}{2} - \sum_{i=2}^{n-j+1}\frac{1}{2^i} - \sum_{i=n-j+2}^{m-j}\frac{1}{2^{i}}\varphi_{j+i} \Big| \\
    &= \Big|\frac{1}{2^{n-j+1}} - \sum_{i=n-j+2}^{m-j}\frac{1}{2^{i}}\varphi_{j+i} \Big| \leq 2^{j-n-1},
\end{split}
\end{equation}
Similarly, if $\varphi_{j+1} = 1$, then $\tilde{r}_j . \tilde{\varphi}_{[j+1:]} = 1.00...0$ is the rounded value of $0.\varphi_{[j+1:]} = 0.11...1\varphi_{n+2}...\varphi_m$. Inputting these values into $\Theta\big(\frac{1}{\sqrt{2}} - \sigma^j \big)$, we will ultimately set $\theta_1 = 0 = \tilde{\varphi}_{j+1} = 1-\varphi_{j+1}$. Therefore, we indeed have $\theta = 0.\tilde{\varphi}_{j+1} \tilde{\varphi}_{j+2} ... \tilde{\varphi}_{n}$, and also 
\begin{equation}
\begin{split}
    \big|\tilde{r}_j . \tilde{\varphi}_{j+1} &- 0.\varphi_{j+1}\big| = \big| 1.00...0 - 0.11...1\varphi_{n+2} ... \varphi_{m} \big| \\
    &= \Big|1 - \sum_{i=1}^{n-j+1}\frac{1}{2^i} - \sum_{i=n-j+2}^{m-j}\frac{1}{2^{i}}\varphi_{j+i} \Big| \\
    &= \Big|\frac{1}{2^{n-j+1}} - \sum_{i=n-j+2}^{m-j}\frac{1}{2^{i}}\varphi_{j+i} \Big| \leq 2^{j-n-1},
\end{split}
\end{equation}

This proves the general case, and the entire proof is complete.

\end{proof} 

Furthermore, we must discuss the additional $j=0$ iteration, which is just like that of the phase estimation by QSVT algorithm. The analysis of this iteration is identical to the case above, just with the modification that $\varphi_0 = 0$ because $\varphi <1$. Ultimately, we find that this step usually will output $\theta_0 = 0$, such that $\theta = 0.\tilde{\varphi}_{[1:]}$. However, if $\varphi_1 = 1 = \varphi_2 = ... = \varphi_{n+1}$ and $\tilde{\varphi}_1 = 0 = \tilde{\varphi}_2 = ... = \tilde{\varphi}_{n}$, then this step will output $\theta_0 = 1$, such that $\theta = 1.00...0$. This accounts for the possibility that $1.0$ is the best approximation to $\varphi$. For instance, the best 2-decimal approximation to $\varphi = 0.1110101$ is $\theta = 1.00$. In either case, $\theta$ still satisfies $|\theta - \varphi| \leq 2^{-n-1}$, as Lemma~\ref{lemma:n_l_m} dictates at $j=0$.  

Finally, we mention one last caveat. It is possible that $\sigma^j = \frac{1}{\sqrt{2}}$, in which case the sign function is $0$, and we are equally likely to measure $0$ or $1$. However, this is not a problem, as the condition $\sigma^j = \frac{1}{\sqrt{2}}$ implies that $|0.\varphi_{[j+1:]} - 0.0\tilde{\varphi}_{[j+2:]}| = 0.1 = \frac{1}{2}$, such that both $\varphi_{j+1} = 0$ and $\varphi_{j+1} = 1$ are equally accurate approximations, and so the inequality in Theorem~\ref{thm:n_l_m} is still obeyed. This is analogous to conventional rounding, wherein one could round 0.5 to either 0 or 1 without changing the accuracy of the rounding. With these concerns alleviated, the proof of Theorem~\ref{thm:n_l_m} is complete.

\subsection{Theorem~\ref{thm:Delta_restriction}} \label{sec:Delta_mitigation}
Here, we prove the Theorem~\ref{thm:Delta_restriction} from Section~\ref{sec:PhaseEstimation_Caveats}:
\begingroup
\def\thetheorem{\ref{thm:Delta_restriction}}
    \begin{theorem}
    If we choose $\Delta$ such that 
        \begin{equation}
        \begin{gathered}
            \Delta < 2\left(\cos(\frac{3\pi}{16}) - \frac{1}{\sqrt{2}} \right) \approx 0.25.
        \end{gathered}
        \end{equation}
    then an error due to $\Delta$ can only occur at the $j=n-1$ iteration. If an error is made at this iteration, then at the end of the algorithm, $|\theta - \varphi| < \frac{1}{2^n}$, assuming no errors are made at later iterations.
    \end{theorem}
\addtocounter{theorem}{-1}
\endgroup

As discussed in Section~\ref{sec:PhaseEstimation_Caveats}, we assume that we perform QSVT with a function that behaves as $P^{\Theta}_{\epsilon, \Delta}\big(\frac{1}{\sqrt{2}}-x\big)$ (for $x\geq 0$), which approximates the sign function. Recall that this approximation fails in the region $[\frac{1}{\sqrt{2}} - \frac{\Delta}{2}, \frac{1}{\sqrt{2}}+\frac{\Delta}{2}]$, which we dub ``the $\Delta$-region". We show that the adverse effects of a finite sized $\Delta$-region can be mitigated by choosing a sufficiently small $\Delta$.

Throughout this section, we again use the notation $0.\varphi_{[j:]} := 0.\varphi_{j}\varphi_{j+1}...$ to denote a string of contiguous binary digits. We also use the definition of $\tilde{r}_j$ from Section~\ref{sec:n_l_m}. 

To demonstrate our claim, recall that $\sigma^j = |\cos(\pi (0.\varphi_{[j+1:]} - \theta)  )|$ as per Eq.~(\ref{eq:singular_value_reexpression}). Whether or not $\sigma^j$ is inside of the $\Delta$-region is dictated by the value of $|0.\varphi_{[j+1:]} - \theta|$. In particular, in order for $\sigma^j$ to be inside of the $\Delta$-region, we require $|\cos(\pi (0.\varphi_{[j+1:]} - \theta)  )| \in \left[\frac{1}{\sqrt{2}} -\frac{\Delta}{2}, \frac{1}{\sqrt{2}} +\frac{\Delta}{2}\right]$, or equivalently, 
\begin{equation}
    \begin{split}
        \pi|0.&\varphi_{[j+1:]} - \theta| \in \\
        &\left[\arccos(\frac{1}{\sqrt{2}} +\frac{\Delta}{2}), \arccos(\frac{1}{\sqrt{2}} -\frac{\Delta}{2})\right] \\
        \cup &\left[\arccos(-\frac{1}{\sqrt{2}} +\frac{\Delta}{2}), \arccos(-\frac{1}{\sqrt{2}} -\frac{\Delta}{2})\right].
    \end{split}
\end{equation} 
We depict this condition graphically in Figure~\ref{fig:DeltaRegion}, which is a useful illustration for understanding the proof of this theorem.

\begin{figure}[htpb]
    \centering
    \includegraphics[width=\columnwidth]{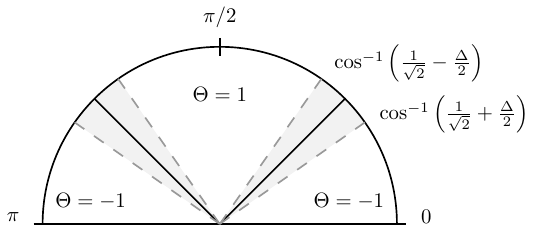}
    \caption{A depiction of a finite sized $\Delta$-region and the resulting values of $\Theta$. Shaded regions indicate where the application of the QSVT sequence returns an indeterminate measurement result (equivalently where the sign function is poorly approximated). }
    \label{fig:DeltaRegion}
\end{figure}

As we show below, if we choose $\Delta$ such that $\sigma^j$ is within the $\Delta$-region only if $|0.\varphi_{[j+1:]} - \theta| \in (\frac{1}{4} - \frac{1}{16},\ \frac{1}{4} + \frac{1}{16}) \cup (\frac{3}{4} - \frac{1}{16},\ \frac{3}{4} + \frac{1}{16})$, then it is only possible for $\sigma^j$ to be within the $\Delta$-region at iteration $j=n-1$. In order to enforce this constraint on the $\Delta$-region, we require that 
\begin{equation}
    \Big| \frac{1}{\pi}\arccos\left( \frac{1}{\sqrt{2}} \pm \frac{\Delta}{2} \right) - \frac{1}{4} \Big| < \frac{1}{16}.
\end{equation}
As it turns out, the $+$ condition is actually more stringent, so we have the following lemma:
\begin{lemma}\label{lemma:Delta_restriction}
If we choose $\Delta$ such that 
\begin{equation}
    \left( \frac{1}{4} - \frac{1}{\pi}\arccos( \frac{1}{\sqrt{2}} + \frac{\Delta}{2}) \right) < \frac{1}{16},
\end{equation}
then $\sigma^j$ can only be inside of the $\Delta$-region at iteration $j=n-1$. 
\end{lemma}
\begin{proof}
To prove this, we will begin by analyzing the possible values of $\sigma^j$. We assume that, if $\sigma^j$ is outside of the $\Delta$-region, then we can correctly determine $\theta_1$ with high probability using an appropriate value of $\epsilon$ as in Section~\ref{sec:PhaseEstimation_Caveats}. In addition, note that our restriction on $\Delta$ implies that, in order for $\sigma^j$ to be inside of the $\Delta$-region, we must have $|0.\varphi_{[j+1:]} - \theta| \in (\frac{1}{4} - \frac{1}{16},\ \frac{1}{4} + \frac{1}{16}) \cup (\frac{3}{4} - \frac{1}{16},\ \frac{3}{4} + \frac{1}{16})$

First, let $j=n-1$ and suppose that we can correctly determine $\theta_1$ (either $\sigma^j$ is outside of the $\Delta$-region, or $\sigma^j$ is inside of the $\Delta$-region and we get lucky). Next, proceed to iteration $j=n-2$. If $\tilde{r}_{n-1} = 0$ at the previous iteration, then Lemma~\ref{lemma:n_l_m} implies that, 
\begin{equation}
\begin{split}
   |0.\varphi_{[n-1:]}& - \theta| = |0.\varphi_{[n-1:]} - 0.\tilde{r}_{n-1} \tilde{\varphi}_{n:}| \\
    \in &\left[0.\varphi_{n-1} - \frac{1}{2}2^{n-1-n-1},\ 0.\varphi_{n-1} + \frac{1}{2}2^{n-1-n-1} \right] \\
    = &\left[0.\varphi_{n-1} - \frac{1}{8},\ 0.\varphi_{n-1} + \frac{1}{8} \right].
\end{split}
\end{equation}
For either possible value of $\varphi_{n-1}$, the restriction on $\Delta$ implies that the corresponding value of $\sigma^{n-2}$ is not within the $\Delta$-region, and so we can correctly determine $\theta_1$ with high probability at this iteration. 

On the other hand, if $\tilde{r}_{n-1} = 1$, then $0.\varphi_{[n:]} > 0.11$ and $\theta = 0$. At iteration $j=n-2$, we have
\begin{equation}
\begin{gathered}
    |0.\varphi_{[n-1:]} - \theta| = 0.\varphi_{[n-1:]} > 0.\varphi_{n-1}11.
\end{gathered}
\end{equation}
Again, because $0.011 = \frac{1}{4} + \frac{1}{8} > \frac{1}{4} + \frac{1}{16}$, the corresponding value of $\sigma^{n-2}$ is not within the $\Delta$-region for either possible value of $\varphi_{n-1}$, and so we can correctly determine $\theta_1$ with high probability at this iteration.

By inductively following this logic to further iterations, we see that, if we can correctly determine $\theta_1$ at iteration $j=n-1$, then the subsequent values of $\sigma^j$ will not be in the $\Delta$-region, and the corresponding values of $\theta_1$ can be correctly determined with high probability.

Next, again let $j=n-1$ and now suppose that $\sigma^j$ is inside of the $\Delta$-region, such that we choose an incorrect value for $\theta_1$. First, consider the case in which $\frac{1}{4} < 0.\varphi_{[n:]} <\frac{1}{4} + \frac{1}{16}$, and we make an error by setting $\theta_1 = 0$. Then at iteration $j=n-2$, 
\begin{equation}
\begin{gathered}
    |0.\varphi_{[n-1:]} - \theta| = 0.\varphi_{n-1}\varphi_{[n:]} 
\end{gathered}
\end{equation}
This is bounded below by $0.\varphi_{n-1} + \frac{1}{8}$ and above by $0.\varphi_{n-1} + \frac{1}{8} +  \frac{1}{32}$, so $\sigma^{n-2}$ does not fall within the $\Delta$-region. Hence, despite our initial error, the correct value for $\theta_1$ may be determined with high probability at the next iteration. As per the result of the previous paragraphs, this indicates that, with high probability, $\sigma^j$ will not be inside the $\Delta$-region at later iterations. 

Next, consider the case in which $\frac{1}{4}-\frac{1}{16} < 0.\varphi_{[n:]} <\frac{1}{4}$, and we make an error by setting $\theta_1 = 1$. Then at iteration $j=n-2$, 
\begin{equation}
\begin{gathered}
    0.\varphi_{[n-1:]} - 0.01 = 0.\varphi_{n-1}\varphi_{[n:]} - 0.01.
\end{gathered}
\end{equation}
This is bounded below by $0.\varphi_{n-1} - \frac{1}{8} - \frac{1}{32}$ and above by $0.\varphi_{n-1} - \frac{1}{8}$, so again $\sigma^{n-2}$ does not fall within the $\Delta$-region. 

Finally, we note that the cases in which $\frac{3}{4} < 0.\varphi_{[n:]} <\frac{3}{4} + \frac{1}{16}$ and $\frac{3}{4}-\frac{1}{16} < 0.\varphi_{[n:]} <\frac{3}{4}$ are analogous to the two cases illustrated above, and we ultimately find that $\sigma^{n-2}$ is not within the $\Delta$-region. 

\end{proof}

This theorem tells us that, if $\Delta$ is made sufficiently small, then we are only plagued by the $\Delta$-region at iteration $j=n-1$. At all other iterations, $\sigma^j$ will not be within the $\Delta$-region, and we can correctly determine $\theta_1$ with high probability. Thus, if at iteration $j=n-1$, $\sigma^{n-1}$ is within the $\Delta$-region and we do make an error, then our approximation of $0.\varphi_{[n:]}$ is incorrect by some amount $ < \frac{1}{4} +  \frac{1}{16} < \frac{1}{2}$, so the overall error in our estimate of $\varphi$ will be $< \frac{1}{2^n}$ (with high probability). 

Lastly, if we rearrange the inequality in Lemma~\ref{lemma:Delta_restriction}, then we see that the necessary restriction on $\Delta$ must be 
\begin{equation}
\begin{gathered}
    \Delta < 2\left(\cos(\frac{3\pi}{16}) - \frac{1}{\sqrt{2}} \right) \approx 0.25.
\end{gathered}
\end{equation}
So if we satisfy this constraint, then we can we can guarantee that the algorithm will succeed with high probability and suffer error $< \frac{1}{2^n}$. This proves Theorem~\ref{thm:Delta_restriction}.

As a corollary, suppose more generally that we choose $\Delta$ such that $\left( \frac{1}{4} - \frac{1}{\pi}\arccos( \frac{1}{\sqrt{2}} + \frac{\Delta}{2}) \right) < \gamma$ for some $\gamma>0$. In order for $\sigma^j$ to be inside of the $\Delta$-region under this constraint, we must have $|0.\varphi_{[j+1:]} - \theta| \in (\frac{1}{4} - \gamma,\ \frac{1}{4} + \gamma ) \cup (\frac{3}{4} - \gamma,\ \frac{3}{4} + \gamma )$. Then, suppose that $\sigma^{n-1}$ is inside the $\Delta$-region such that $\frac{1}{4} < 0.\varphi_{[n:]} <\frac{1}{4} + \gamma$, and we make an error by setting $\theta_1 = 0$. Then at iteration $j=n-2$, 
\begin{equation}
\begin{gathered}
    |0.\varphi_{[n-1:]} - \theta| = 0.\varphi_{n-1}\varphi_{[n:]}, 
\end{gathered}
\end{equation}
which is bounded below by $0.\varphi_{n-1} + \frac{1}{8}$ and above by $0.\varphi_{[n-1:]} + \frac{1}{8} +  \frac{\gamma}{2}$. Hence this quantity is necessarily a distance $\frac{1}{4} - \gamma - (\frac{1}{8} +\frac{\gamma}{2}) = \frac{1}{8} - \frac{3\gamma}{2}$ from the $\Delta$-region, and this distance increases at later iterations. Identical bounds hold for the other cases that $\sigma^{n-1}$ is in the $\Delta$-region. In this sense, the quantity $0.\varphi_{[j+1:]}$ can suffer an additive error of magnitude $<\frac{1}{8} - \frac{3\gamma}{2}$, and the phase estimation algorithm will still output a $\theta$ such that $|\varphi - \theta| < \frac{1}{2^n}$ with high probability, as claimed in Section~\ref{sec:PhaseEstimation_Robust}.

\section{Construction of the Matrix Inversion Polynomial}\label{sec:matrix_inversion_polynomial}

As we described in Sec.~\ref{sec:matrix_inversion}, in order to invert a matrix with QSVT, we desire an $\frac{\epsilon}{2 \kappa}$-approximation to $\frac{1}{2 \kappa} \frac{1}{x}$, where $\kappa$ is the condition number of the matrix to be inverted. Gily\'en et al. design such a polynomial by first noting that the function
    \begin{equation}
        g_{\epsilon, \kappa}(x) = \frac{1-(1-x^2)^b}{x}
    \end{equation}
provides a good approximation to $\frac{1}{x}$ over the range $x \in [-1, 1] \setminus [\frac{-1}{\kappa}, \frac{1}{\kappa}]$ for large $b$. In particular, for $0< \epsilon < \frac{1}{2}$, $g_{\epsilon,\kappa}(x)$ $\epsilon$-approximates $\frac{1}{x}$ over the range $x \in [-1, 1] \setminus [\frac{-1}{\kappa}, \frac{1}{\kappa}]$ for $b(\epsilon,\kappa) = \lceil \kappa^2 \log(\kappa/\epsilon) \rceil$~\cite{Gily_n_2019, Childs_2017}.

Next, although $g_{\epsilon, \kappa}(x)$ is not a polynomial, it can be $\epsilon$-approximated over the range $x \in [-1, 1] $ by the polynomial
    \begin{equation}\label{eq:1OverxPoly}
        P^{1/x}_{2\epsilon, \kappa}(x) = 4\sum_{j=0}^D (-1)^j \left[2^{-2b} \sum_{i=j+1}^b {2b \choose b+i}\right] T_{2j+1}(x),
    \end{equation}
where $T_i(x)$ is the Chebyshev polynomial of order $i$ and $D(\epsilon,\kappa) = \Big\lceil \sqrt{b(\epsilon,\kappa) \log(4b(\epsilon,\kappa)/\epsilon)} \Big\rceil = \mathcal{O}(\kappa \log(\kappa/\epsilon))$ is the degree of this polynomial~\cite{Gily_n_2019}. In addition, by the triangle inequality, $P^{1/x}_{2\epsilon, \kappa}(x)$ is a $2\epsilon$-approximation to $\frac{1}{x}$ for $x \in [-1, 1] \setminus [\frac{-1}{\kappa}, \frac{1}{\kappa}]$, hence the subscript $2\epsilon$. To provide intuition, we illustrate this polynomial in Figure~\ref{fig:MatrixInversionPolynomial}.

At this stage, we may suspect that our candidate polynomial is $\frac{1}{2 \kappa} P^{1/x}_{\frac{\epsilon}{2}, 2 \kappa}(x)$, which $\epsilon$-approximates $\frac{1}{2 \kappa} \frac{1}{x}$ for $x \in [-1,1] \setminus [\frac{-1}{2\kappa}, \frac{1}{2\kappa}]$. In addition, for $\epsilon < 1/2$, this polynomial is necessarily bounded in magnitude by $1$ for $x \in [-1,1] \setminus [\frac{-1}{\kappa}, \frac{1}{\kappa}]$, which is easily seen via the triangle inequality. Unfortunately however, this candidate polynomial is not necessarily bounded for $x \in [\frac{-1}{2 \kappa}, \frac{1}{2 \kappa}]$. In particular, the approximation of Eq.~(\ref{eq:1OverxPoly}) obeys~\cite{Childs_2017, Gily_n_2019}
\begin{equation}
    \max_{x\in[-1,1]} |P^{1/x}_{\epsilon, \kappa}(x)| \leq 4D(\epsilon,\kappa) = \mathcal{O}(\kappa \log(\kappa/\epsilon)),
\end{equation}
so our candidate polynomial is only bounded in magnitude by $\frac{1}{2\kappa} 4 D(\frac{\epsilon}{2}, 2\kappa) = \mathcal{O}(\log(\kappa/ \epsilon))$, which is not necessarily $1$. 

Therefore, to enforce that the magnitude of the candidate polynomial be bounded, we may multiply it by an even function that is close to $1$ for $x \in [-1, 1] \setminus [\frac{-1}{\kappa}, \frac{1}{\kappa}]$, and close to $0$ for $x \in [\frac{-1}{2\kappa}, \frac{1}{2\kappa}]$, letting the range $x \in [\frac{-1}{\kappa}, \frac{-1}{2\kappa}] \cup [\frac{1}{2\kappa}, \frac{1}{\kappa}]$ be a transition region between these two values. Such a rectangular function may be polynomially approximated by a linear combination of the step function approximations of Section~\ref{sec:Search}:
\begin{equation}
\begin{split}
    P^{\text{rect}}_{\epsilon, \kappa}(x)& := \\
    \frac{1}{1+\frac{\epsilon}{2}}&\left(1 + \frac{1}{2}\left(  P^{\Theta}_{\epsilon, \frac{1}{4 \kappa}} \left( x - \tfrac{3}{4 \kappa} \right) + P^{\Theta}_{\epsilon, \frac{1}{4 \kappa}} \left(-x - \tfrac{3}{4 \kappa} \right) \right) \right),
\end{split}
\end{equation}
which is easily seen to obey
\begin{equation}
\begin{gathered}
        P^{\text{rect}}_{\epsilon, \kappa}(x) \in [1-\epsilon,1] \ \ \forall x \in [-1, 1] \setminus \Big[\frac{-1}{\kappa}, \frac{1}{\kappa} \Big] \\
        P^{\text{rect}}_{\epsilon, \kappa}(x) \in [0,\epsilon] \ \ \forall x \in \Big[\frac{-1}{2\kappa}, \frac{1}{2\kappa} \Big]
\end{gathered}
\end{equation}
and has even degree $\mathcal{O}(\kappa \log(1/\epsilon))$. In the in-between region $x \in [\frac{-1}{\kappa}, \frac{-1}{2\kappa}] \cup [\frac{1}{2\kappa}, \frac{1}{\kappa}]$, $P^{\text{rect}}_{\epsilon, \kappa}(x)$ transitions between values close to 0 and close to 1, remaining bounded in magnitude by $1$ throughout. We illustrate the behavior of this polynomial in Fig~\ref{fig:MatrixInversionPolynomial}.

Therefore, our target polynomial is the \emph{matrix inversion polynomial}: 
\begin{equation}\label{eq:1OverxApprox}
\begin{gathered}
     P^{\text{MI}}_{\epsilon,\kappa}(x) := \frac{1}{2\kappa} P^{1/x}_{\frac{1}{2} \epsilon, 2 \kappa}(x) P^{\text{rect}}_{\epsilon', \kappa}(x)   
\end{gathered}
\end{equation}
where $\epsilon' = \text{min}\left( \frac{2\epsilon}{5\kappa}, \frac{\kappa}{2D(\epsilon/4, 2\kappa)} \right) = \mathcal{O}\left(\frac{\epsilon}{\kappa}\right)$. We illustrate this polynomial in Fig~\ref{fig:MatrixInversionPolynomial}. In defining $\epsilon'$, the term $\frac{2\epsilon}{5\kappa}$ ensures that this polynomial is an $\frac{\epsilon}{2\kappa}$-approximation to $\frac{1}{2 \kappa}\frac{1}{x}$ over the range of possible singular values, and the term $\frac{\kappa}{2D(\epsilon/4, 2\kappa)}$ ensures that $P^{\text{MI}}_{\epsilon,\kappa}(x)$ is bounded in magnitude by 1 for $x\in [\frac{-1}{\kappa}, \frac{1}{\kappa}]$. Indeed, over the range $x \in [-1, 1] \setminus [\frac{-1}{\kappa}, \frac{1}{\kappa}]$, $P^{\text{MI}}_{\epsilon,\kappa}(x)$ is an $\frac{\epsilon}{2\kappa}$-approximation to $\frac{1}{2 \kappa}\frac{1}{x}$:
\begin{equation}
    \begin{split}
        &\left|P^{\text{MI}}_{\epsilon,\kappa}(x) - \frac{1}{2\kappa}\frac{1}{x}\right| 
        \leq \frac{1}{2\kappa} \left|P^{1/x}_{\frac{1}{2} \epsilon, 2 \kappa}(x)(1-\epsilon') - \frac{1}{x} \right| \\
        &\leq \frac{1}{2\kappa}\left( \left|P^{1/x}_{\frac{1}{2} \epsilon , 2 \kappa}(x) - \frac{1}{x} \right|  + \epsilon' \left| P^{1/x}_{\frac{1}{2} \epsilon, 2 \kappa}(x)\right| \right) \\
        &\leq \frac{1}{2\kappa}\left( \frac{\epsilon}{2} + \frac{2\epsilon}{5\kappa} \left(\kappa + \frac{\epsilon}{2} \right) \right) \\ 
        &\leq \frac{1}{2\kappa}\left( \frac{\epsilon}{2} + \frac{2\epsilon}{5} + \frac{\epsilon}{10} \right) 
        = \frac{\epsilon}{2\kappa}
    \end{split}
\end{equation}
where we have used the (loose) bound $\epsilon/\kappa < \frac{1}{2}$. Likewise, $P^{\text{MI}}_{\epsilon,\kappa}(x)$ is bounded above by $1$ for $x \in [\frac{-1}{2\kappa}, \frac{1}{2\kappa}]$:
\begin{equation}
    \begin{gathered}
        \left|P^{\text{MI}}_{\epsilon,\kappa}(x)\right| \leq \frac{1}{2\kappa} 4D\left(\epsilon/4, 2\kappa \right) \epsilon' \leq  1.
    \end{gathered}
\end{equation}
Similarly, for the in-between region $x \in [\frac{-1}{\kappa}, \frac{-1}{2\kappa}] \cup [\frac{1}{2\kappa}, \frac{1}{\kappa}]$, both components of $P^{\text{MI}}_{\epsilon,\kappa}(x)$ are bounded in magnitude by $1$, so $|P^{\text{MI}}_{\epsilon,\kappa}(x)| \leq 1$.

\begin{figure}[htpb]
    \centering
    \includegraphics[width=\columnwidth]{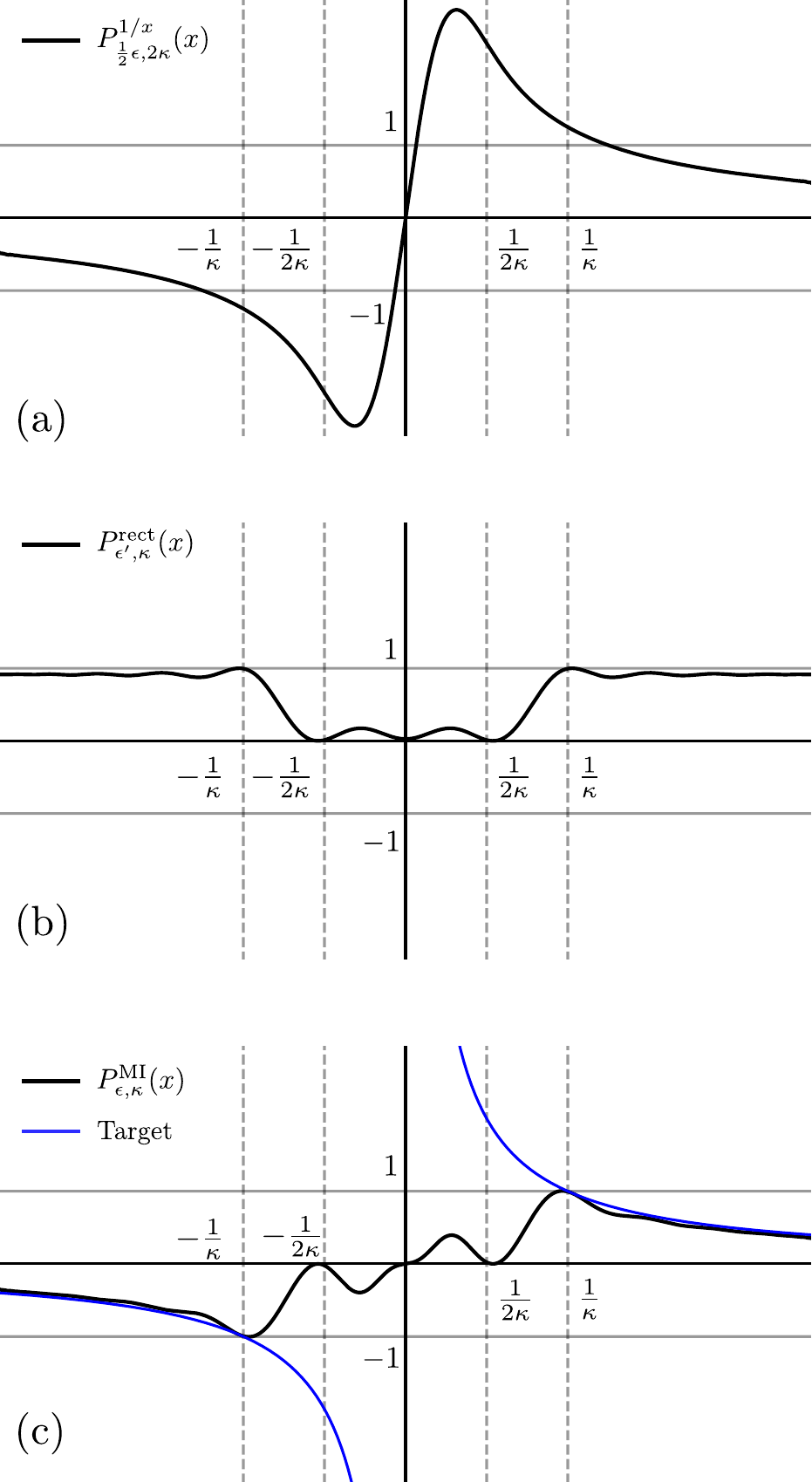}
    \caption{A sketch of (a) $P^{1/x}_{\frac{1}{2} \epsilon, 2\kappa}(x)$, (b) $P^{\text{rect}}_{\epsilon',\kappa}(x)$, and (c) $P^{\text{MI}}_{\epsilon,\kappa}(x)$ for $\kappa=2.5$. The resulting polynomial, $P^{\text{MI}}_{\epsilon,\kappa}(x)$, is a degree $77$ approximation to the inverse function.}
    \label{fig:MatrixInversionPolynomial}
\end{figure}

Finally, it is easy to compute the degree of $P^{\text{MI}}_{\epsilon,\kappa}(x)$, which is the sum of the degrees of $P^{1/x}_{\frac{1}{2} \epsilon, 2 \kappa}(x)$ and $P^{\text{rect}}_{\epsilon', \kappa}(x)$:
\begin{equation}
    d = \mathcal{O}\big(\kappa \log(\kappa/\epsilon) + \kappa \log(1/\epsilon') \big)  = \mathcal{O}\left(\kappa \log(\kappa/\epsilon) \right).
\end{equation}

\section{QSP Phase Angle Sequence Examples}
\label{sec:qsp_angle_examples}

Presented in this appendix are some explicit polynomials and
corresponding QSP phase angles, for functions which are useful in
quantum signal processing applications.  These are not necessarily the
optimal polynomials, nor the best QSP phase angles, but they are
pedagogically clear starting points. Unless otherwise specified, all QSP 
phases are given in the $(W_x, S_z, \braket{+| \cdot | +})$-QSP convention.
All the code for generating these phase angles is available in the {\tt pyqsp} 
open source repository on GitHub\footnote{{\tt \url{https://github.com/ichuang/pyqsp}}}.

\subsection{Oblivious amplitude amplification}
\label{sec:appx_fpsearch}

For fixed point search or oblivious amplitude amplification, it is
desired to make a polynomial which maps $a$ as close as possible to
$1$, for a wide range of small values of $a$, starting as close to
$a=0$ as possible.

One sequence of phases which accomplishes this optimally, with error $1-\delta^2$, is given for $k = 0, 1, \ldots, d-1$
by
\bea
	\phi_{2k} &=& \alpha_{d-k-1}
\\
	\phi_{2k+1} &=& \alpha_{k+1}
\eea
where
\bea
	\alpha_k = -\cot^{-1}\lpL \sqrt{1-\gamma^2} \tan \frac{2\pi (k+1)}{L} \right)
\eea
and $L=2d+1$ and $\gamma^{-1} = T_{1/L}(1/\delta) = \cosh^{-1}(1/\delta) / \cosh(1/L)$ \cite{PhysRevLett.113.210501}.

For example, for $d=10$ and $\delta=0.5$, the QSP phase angles are:
\begin{verbatim}
pyqsp --plot-positive-only --plot-probability 
--plot-tight-y --plot-npts=400 
--seqargs=10,0.5 --plot fpsearch

[-1.58023603 -1.55147987 -1.6009483  -1.52812171 
 -1.62884337 -1.49242141 -1.67885248 -1.41255145 
 -1.8386054  -0.87463828 -0.87463828 -1.8386054
 -1.41255145 -1.67885248 -1.49242141 -1.62884337 
 -1.52812171 -1.6009483  -1.55147987 -1.58023603]

\end{verbatim}
The corresponding response function is shown in Figure~\ref{fig:example_amplification}.
\begin{figure}[htpb]
    \centering
    \includegraphics[width=\columnwidth]{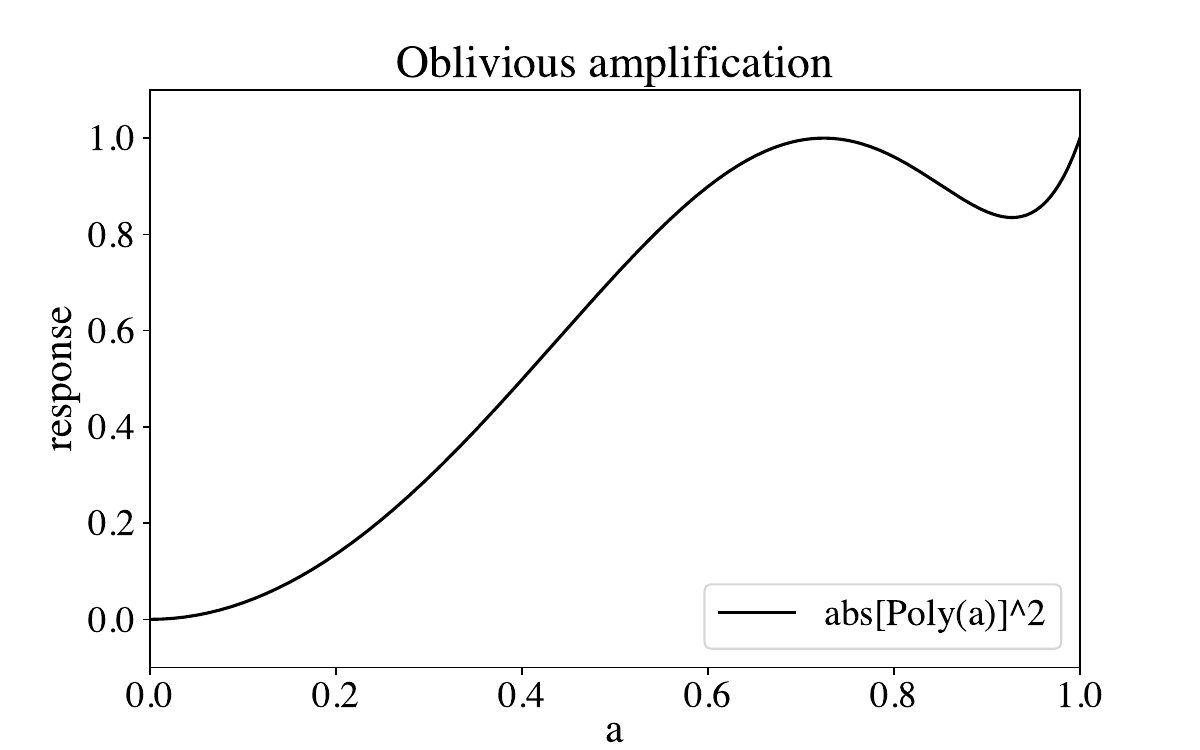}
    \caption{Transition probability using the polynomial for oblivious amplification for $d=10$ and $\delta=0.5$.}
    \label{fig:example_amplification}
\end{figure}

\subsection{Sign function}
\label{sec:appx_sign_func}

The sign function $\Theta(a)$ has a number of applications for quantum signal processing.
As discussed in Section~\ref{sec:Search}, a robust polynomial approximation of it can be
obtained using the error function
\be
   \Theta(a) \approx {\rm erf}(k a)
\,,
\ee
where $k$ is a (large) scaling factor.  

For example, the QSP phase angles obtained using an $d=19$ order approximation with $k=10$ are
\begin{verbatim}
pyqsp --plot-real-only --plot-npts=400 
--seqargs=19,10 --plot poly_sign

[0.01558127 -0.01805798  0.05705643 -0.01661832  
 0.16163773  0.09379074 -2.62342885  0.49168481  
 0.92403822 -0.09696846 -0.09696846  0.92403822
 0.49168481 -2.62342885  0.09379074  0.16163773 
-0.01661832  0.05705643 -0.01805798  1.5863776]

\end{verbatim}
The corresponding response function is shown in Figure~\ref{fig:example_sign}. For $a>0$, this sign function approximation may be employed as a step function, e.g. for oblivious amplitude amplification.
\begin{figure}[htpb]
    \centering
    \includegraphics[width=\columnwidth]{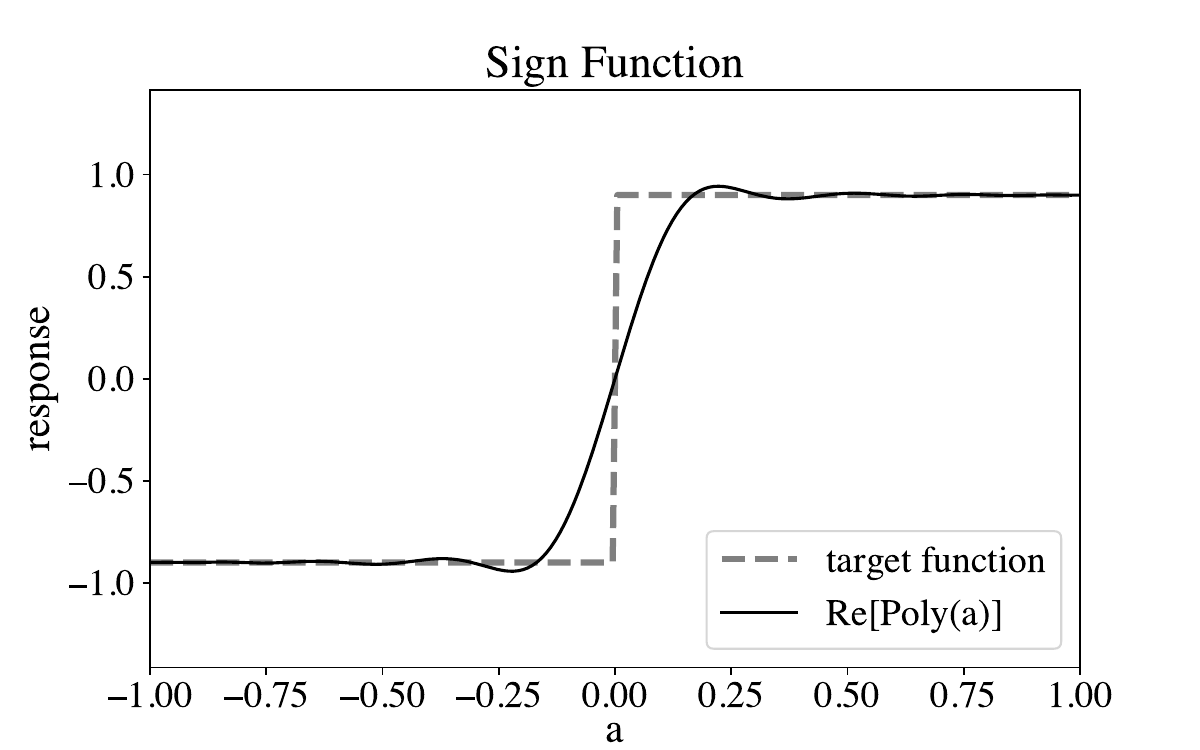}
    \caption{Response function for the polynomial approximation to the sign function with $d=19$ and $k=10$.}
    \label{fig:example_sign}
\end{figure}

\subsection{Matrix inversion using $1/a$}

The $1/a$ function is useful for computing Moore-Penrose
pseudoinverses of matrices using the quantum singular value transform.
As discussed in Section~\ref{sec:matrix_inversion}, Chebyshev polynomials may be employed to approximate this function \cite{Childs_2017}.

For example, with $\kappa=3$ and $\epsilon=0.3$, a set of QSP phase
angles for this polynomial is:
\begin{verbatim}
pyqsp --plot-real-only --plot-npts=400 
--seqargs=3,0.3 --plot invert

[-0.27237279 -1.8808697   2.19755533 -0.860515
  0.84659086  0.62794236 -0.69688032 -0.62874403  
  0.7406656   0.44483992 -0.60489363 -0.60489363
  0.44483992  0.7406656  -0.62874403 -0.69688032 
  0.62794236  0.84659086 -0.860515   -0.94403733  
  1.26072295  1.29842354]
\end{verbatim}
The corresponding response function is shown in Figure~\ref{fig:example_inverse}.
\begin{figure}[htpb]
    \centering
    \includegraphics[width=\columnwidth]{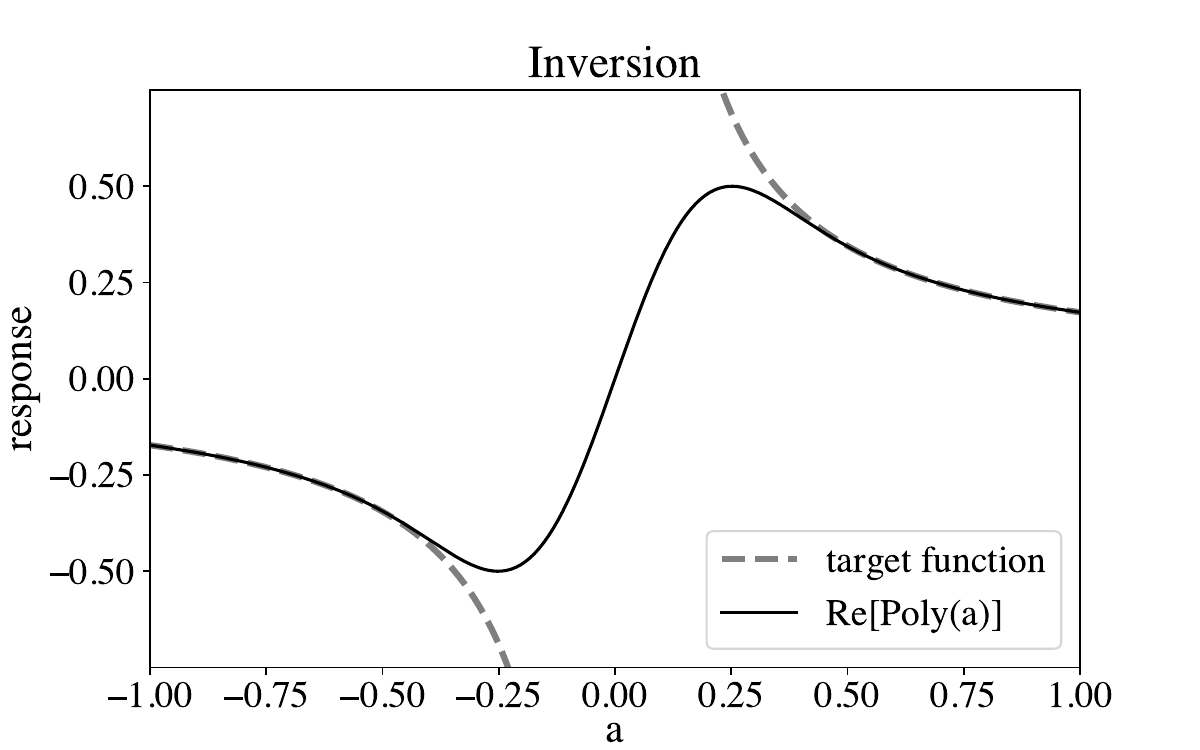}
    \caption{Response function for the polynomial approximation to the inverse function with $\kappa=3$ and $\epsilon=0.3$.}
    \label{fig:example_inverse}
\end{figure}

\subsection{Cosine and sine functions for Hamiltonian simulation}
\label{subsec:hamsim_func}

For Hamiltonian simulation, we seek an approximation to $e^{-i a t}$.
As discussed in section \ref{sec:HamiltonianSimulation}, this can be accomplished using the Jacobi-Anger approximations of $\cos(at)$ and $\sin(at)$ of Eqs.~(\ref{eq:Jacobi-Anger1}) and~(\ref{eq:Jacobi-Anger2}).
The approximation is chosen to sufficient degree so as to bound the error to $\epsilon > 0$ in the region $a \in [-1, 1]$.

For example, with $t=5$ and $\epsilon=0.1$, a set of QSP phase angles for the approximation to $\cos(at)$ are
\begin{verbatim}
pyqsp  --plot-real-only --plot-npts=400 
--seqargs=10,0.1 --plot hamsim

[-1.70932079 -0.05312746  2.12066859 -0.83307065 
 -0.50074601  0.40728859  0.32838472  0.9142489  
 -2.81320793  0.40728859 -0.50074601  2.30852201
 -1.02092406 -0.05312746  3.00306819]
\end{verbatim}
The corresponding response function is shown in Figure~\ref{fig:example_hamsim-cosine}.
\begin{figure}[htpb]
    \centering
    \includegraphics[width=\columnwidth]{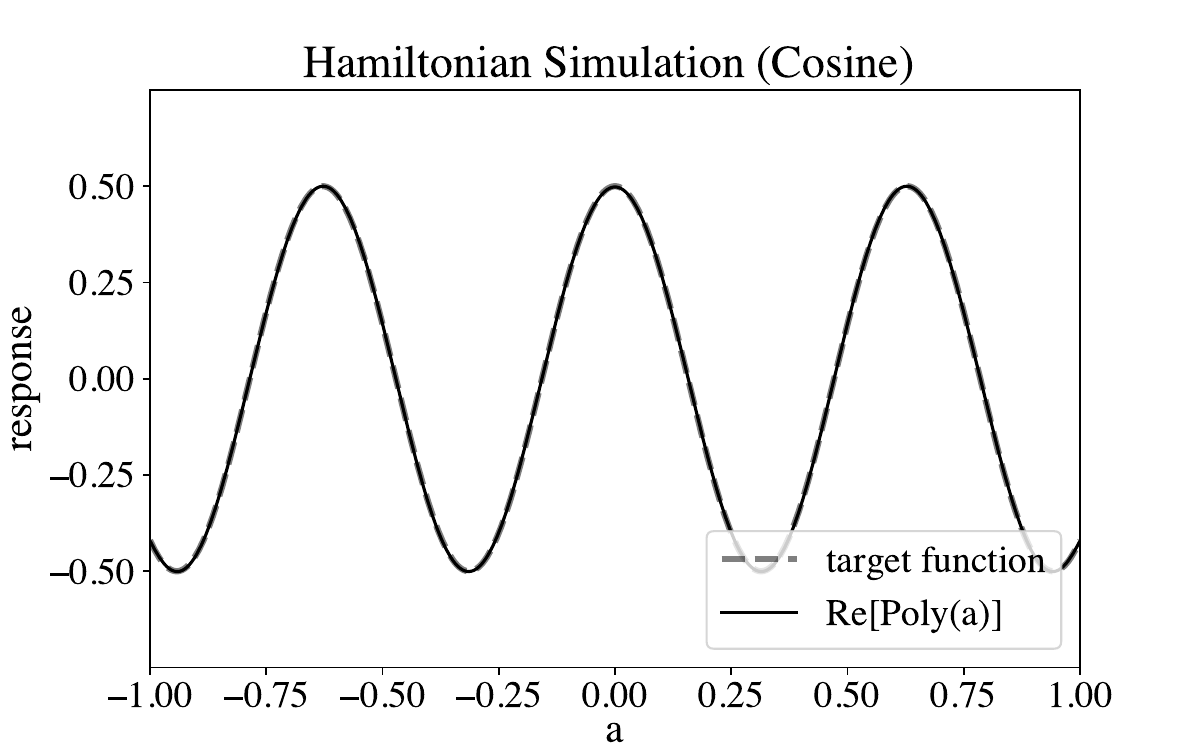}
    \caption{Response function for the polynomial approximation to the cosine function with $t=5$ and $\epsilon=0.1$.}
    \label{fig:example_hamsim-cosine}
\end{figure}

\noindent Similarly, the QSP phases for the approximation to $\sin( at)$ are:
\begin{verbatim}
[-1.63276817  0.20550406 -0.84198335  0.39732059 
 -0.26820613  2.41324245  0.04662674 -2.02847501  
  1.11311765  0.04662674 -0.72835021 -0.26820613
  0.39732059 -0.84198335  0.20550406 -0.06197184]
\end{verbatim}
The corresponding response function is shown in Figure~\ref{fig:example_hamsim-sine}.
\begin{figure}[htpb]
    \centering
    \includegraphics[width=\columnwidth]{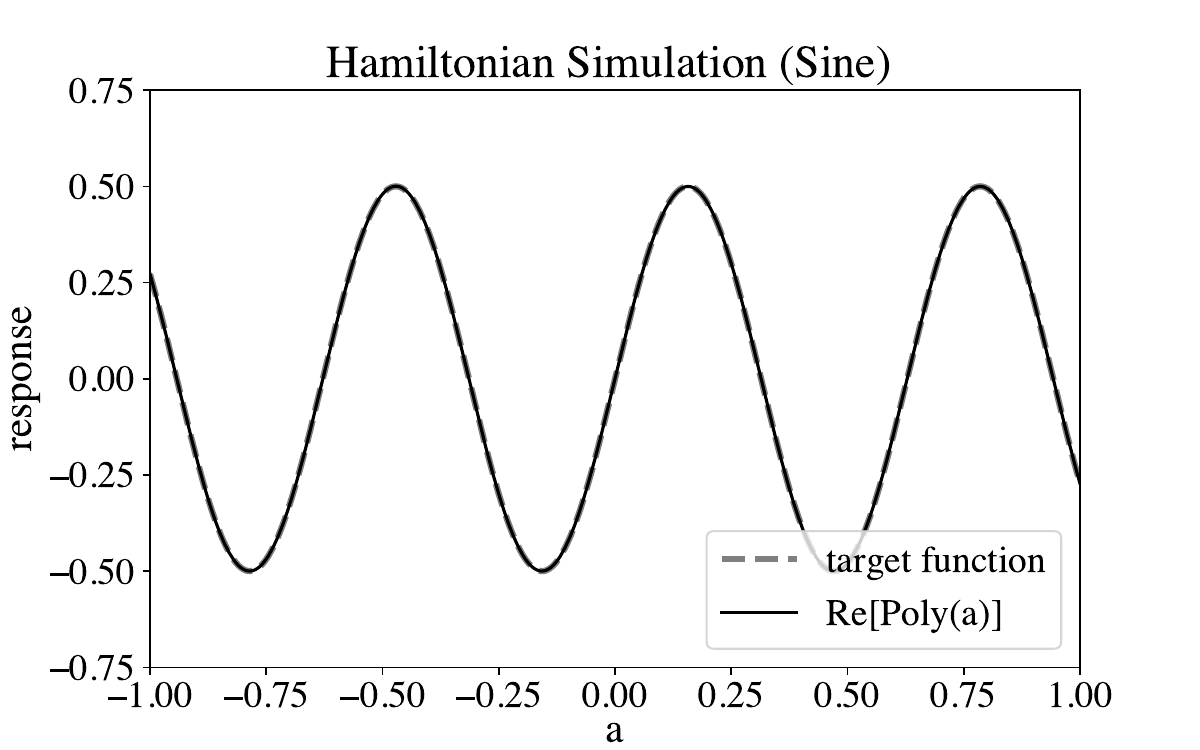}
    \caption{Response function for the polynomial approximation to the sine function with $t=5$ and $\epsilon=0.1$.}
    \label{fig:example_hamsim-sine}
\end{figure}

\subsection{Threshold function}
\label{subsec:thresh_func}

Distinguishing eigenvalues and singular values may be performed using
a step function, which we will take to be  ${\rm step}(a-1/2)$ for illustrative purposes. This may be polynomially approximated using a Taylor series expansion of
\begin{equation}
    \begin{split}
        \rm{step}(a-1/2) &\approx \\
	    \frac{1}{2}\Big( {\rm erf} &\big( k(a+1/2)\big) - {\rm erf}\big( k(a-1/2) \big) \Big),
    \end{split}
\end{equation}
which becomes a good approximation for large $k$.  

For example, with $k=10$ and using a degree $d=18$ Taylor series, a set of QSP phase
angles for this polynomial is:
\begin{verbatim}
pyqsp --plot-real-only --plot-npts=400 
--seqargs=18,10 --plot poly_thresh

[0.73930816 -0.69010006 -0.63972139 -0.47754554  
 0.81797049  0.09205065 -0.87660105  0.13460844  
 0.23892207  1.32216648 -2.90267058  0.13460844
 2.2649916   0.09205065 -2.32362216  2.66404712 
-0.63972139 -0.69010006  2.31010449]
\end{verbatim}
The corresponding response function is shown in Figure~\ref{fig:example_threshold}.
\begin{figure}[htpb]
    \centering
    \includegraphics[width=\columnwidth]{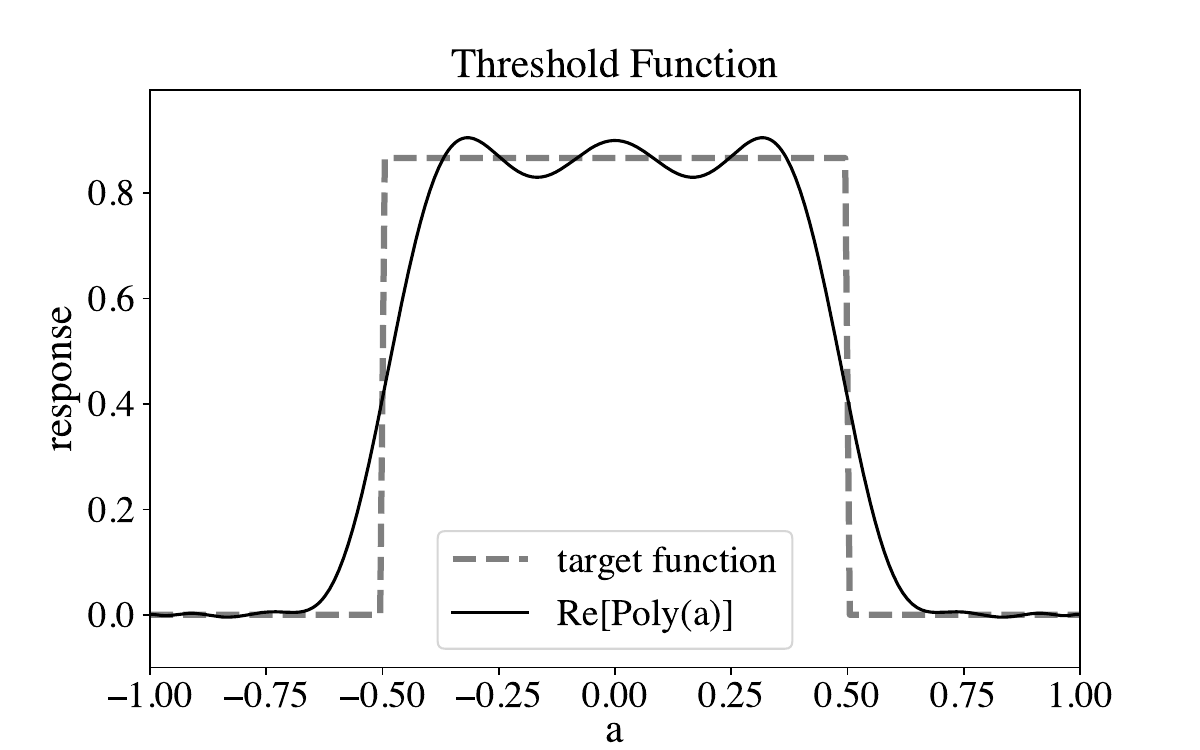}
    \caption{Response function for the polynomial approximation to the threshold function with $k=10$ and $d=18$.}
    \label{fig:example_threshold}
\end{figure}

Note that this is an even function of $a$, but it may be used just in
the region $a\ge 0$, e.g. to distinguish singular values that are
above or below $1/2$.  It can be made as sharp as desired by
increasing $k$ and the degree of the polynomial.  

\subsection{Linear amplitude amplification}
\label{subsec:linear_amplification}

Linear amplitude amplification is a subroutine useful for a number of quantum algorithms including simulation.
The goal is to multiply inputs by a constant $1/2\Gamma$ for $\Gamma \in (0, 1/2]$.
As is usual for QSP, the absolute value of the output must bounded by $1$ and therefore we seek a polynomial approximation that performs the linear amplification only for small inputs.
We can obtain a suitable approximation by truncating the Taylor expansion of
\begin{equation}
    \begin{split}
        \rm{linear\_amplification}(a, \Gamma) &\approx \\
	    \frac{a}{2\Gamma} \times \frac{1}{2}\Big( {\rm erf} &\big( k(a + 2\Gamma)\big) - {\rm erf}\big( k(a - 2\Gamma) \big) \Big),
    \end{split}
\end{equation}
where $k$ is chosen to obtain the desired accuracy within the region $[-\Gamma, \Gamma]$.
This approximation is described in further detail in~\cite{low2017quantum}.

For example, with $\Gamma=0.25$ and using a degree $d=19$ Taylor series, a set of QSP phase angles for this polynomial is:
\begin{verbatim}
pyqsp --plot-real-only --plot-npts=400 
--seqargs=19,0.25 --plot poly_linear_amp

[0.07658557 -0.01961714 -0.10257913  0.08269406  
 0.16291683  0.43552219 -2.62323892  2.61402960  
 0.02001814 -2.21710253  0.92449012  0.02001814
-0.52756304  0.51835372  0.43552219  0.16291683
 0.08269406 -0.10257913  -0.0196171 -1.49421074]

\end{verbatim}
The corresponding response function is shown in Figure~\ref{fig:example_linamp}.
\begin{figure}[htpb]
    \centering
    \includegraphics[width=\columnwidth]{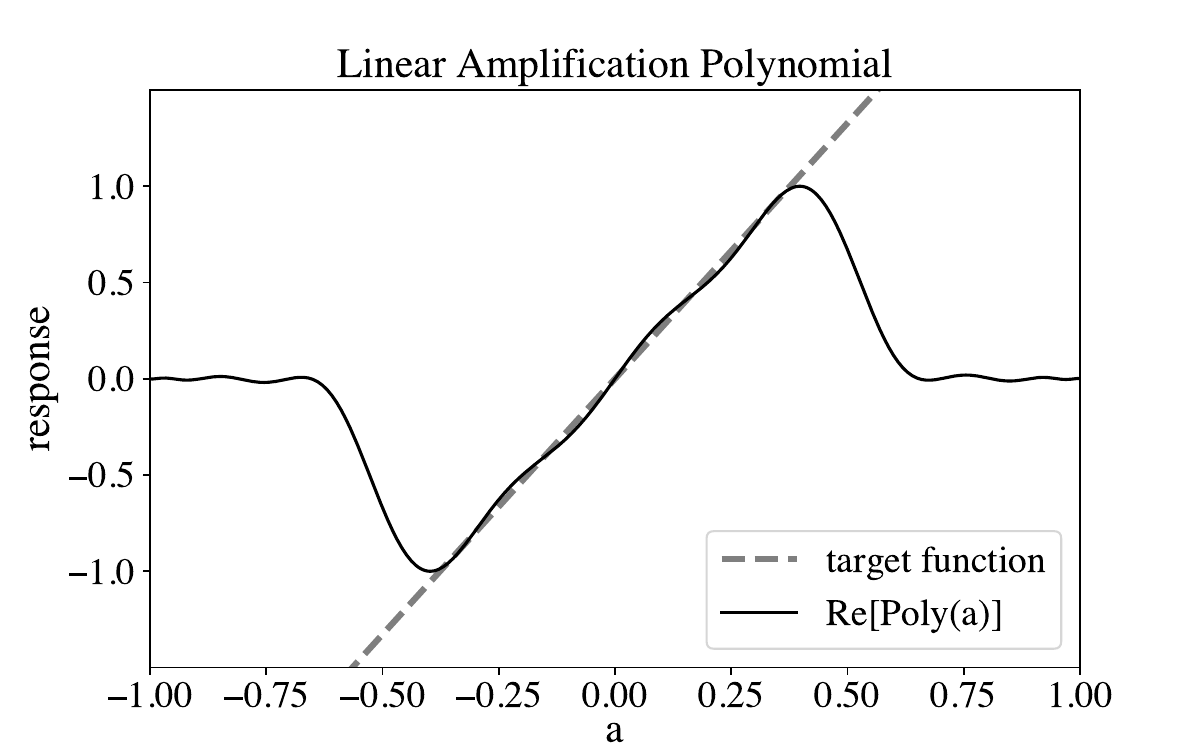}
    \caption{Response function for the linear amplification polynomial with $\Gamma=0.25$ and $d=19$.}
    \label{fig:example_linamp}
\end{figure}

\subsection{Phase estimation polynomial}
\label{subsec:phase_est_func}

Similar to the threshold function is the phase estimation polynomial of Eq.~(\ref{eq:SignFunction_Target}) used in Section~\ref{sec:PhaseEstimation}.

For example, with $\Delta=10$ and using a degree $d=18$ Taylor series, a set of QSP phase
angles for this polynomial in the $(W_x, S_z, \braket{+| \cdot | +})$-QSP convention is:
\begin{verbatim}
pyqsp --plot-real-only --plot-npts=400 
--seqargs=18,10 --plot poly_phase

[-2.69295576  0.92644177 -2.47601161 -2.95790072 
 -3.07710363  2.40352005  2.38432687 -3.0712802 
 -2.54668308 -0.87407521  0.59490957 -3.0712802
 -0.75726578  2.40352005 -3.07710363 -2.95790072 
 -2.47601161  0.92644177  2.01943322]
\end{verbatim}
The corresponding response function is shown in Figure~\ref{fig:example_phase_++}.
\begin{figure}[htpb]
    \centering
    \includegraphics[width=\columnwidth]{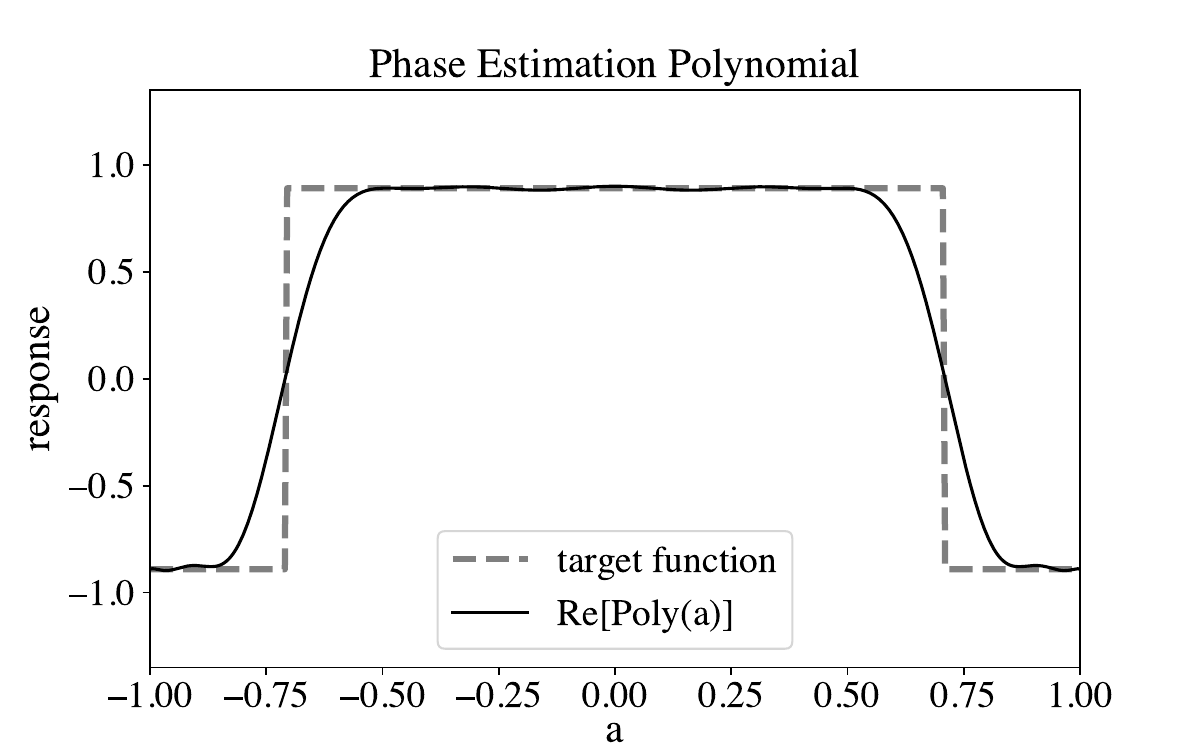}
    \caption{Response function for the real phase estimation polynomial with $\Delta=10$ and $d=18$ in the $(W_x, S_z, \braket{+| \cdot | +})$-QSP convention..}
    \label{fig:example_phase_++}
\end{figure}

A set of QSP phases for this polynomial can also be given in the $(W_x, S_z, \braket{0| \cdot | 0})$-QSP convention; use of this convention clarifies the reduction to the quantum Fourier transform presented in Section~\ref{sec:EmergentQFT}.
\begin{verbatim}
pyqsp --polydeg 16 --measurement="z"
--func="-1+np.sign(1/np.sqrt(2)-x)+
np.sign(1/np.sqrt(2)+x)"  --plot polyfunc

[0.6744825  2.4248297  2.7351234   2.7316442   
 0.0127715  3.915519   2.3178308  -0.00533221
 2.3088486  2.36385    2.3181891   1.585311 
 2.411246   0.4094785 -0.40136954  0.7154387   
 1.8687413]
\end{verbatim}
The corresponding response function is shown in Figure~\ref{fig:example_phase_00}.
\begin{figure}[htpb]
    \centering
    \includegraphics[width=\columnwidth]{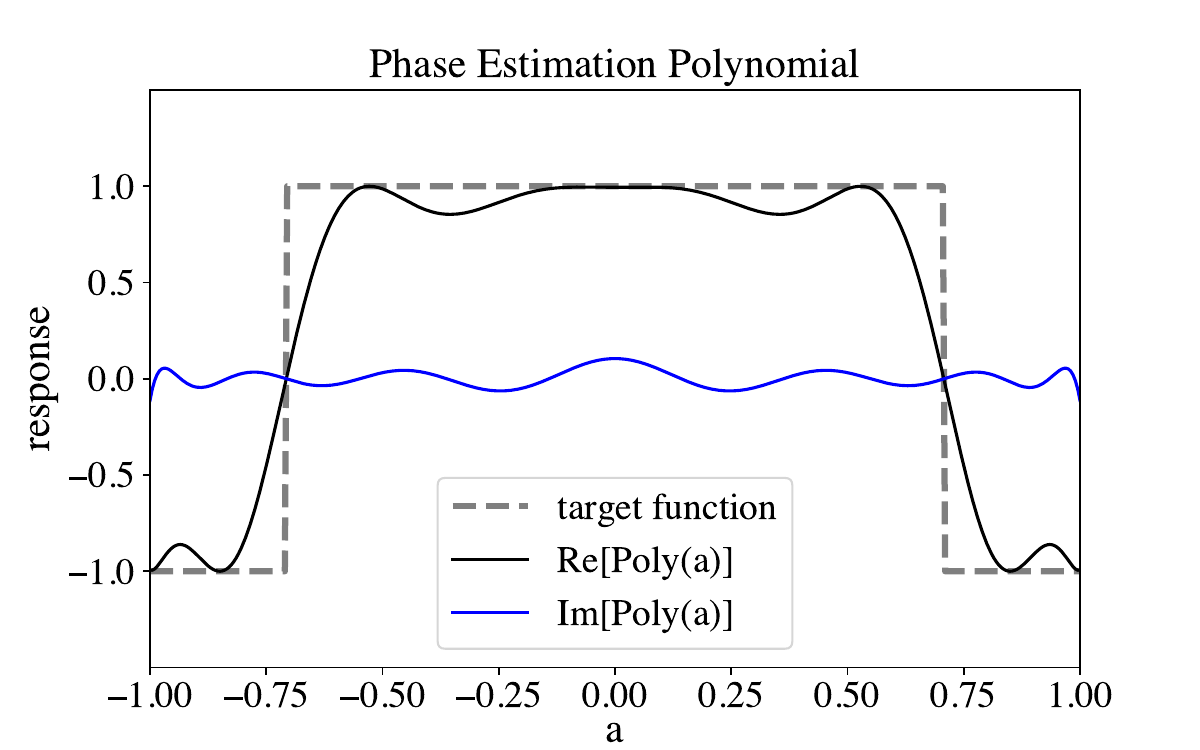}
    \caption{Response function for a degree $18$ polynomial approximation to the phase estimation function in the $(W_x, S_z, \braket{0| \cdot | 0})$ convention.}
    \label{fig:example_phase_00}
\end{figure}
Note that we are no longer using the real polynomial approximation of Eq.~(\ref{eq:SignFunction_Target}) since it does not satisfy the conditions of Theorem~\ref{thm:qsp_Wx-Sz-00}; because of this, there is a small non-zero imaginary response.  The QSP phase angles for this example are generated using an optimization algorithm.

\subsection{Eigenstate filtering}

As in the eigenvalue threshold problem of Section~\ref{sec:Threshold}, 
suppose $\mathcal{H}$ is a Hermitian matrix with an eigenvalue $\lambda$ which
is known to be separated from other eigenvalues by a gap $\Delta_\lambda > 0$,
and the problem is to create, using QSP, a projection operator onto
the eigenspace corresponding to $\lambda$.  Lin and Tong~\cite{Lin_2020} show that the degree $d=2k$ polynomial
\bea
	f_k(x, \Delta_\lambda) = \frac{T_k\left(-1 + 2\frac{x^2-\Delta_\lambda^2}{1-\Delta_\lambda^2} \right)}{T_k\left(-1 + 2\frac{-\Delta_\lambda^2}{1-\Delta_\lambda^2} \right)}
\,,
\eea
known as the ``eigenstate filtering function,'' is an optimal
polynomial for filtering out the unwanted information from all other
eigenstates.

For example, with $\delta=0.3$ and using a degree $d=30$ Taylor series, a set of QSP phase
angles for this polynomial is:
\begin{verbatim}
pyqsp --plot-positive-only --plot-real-only 
--plot-tight-y --seqargs 30,0.3 --plot efilter

[-2.22655153  2.26982696 -0.76378956  0.07418111  
  0.25458387  0.5916072   0.30309002  0.10101557 
 -0.12943648 -1.00141425  0.60368299 -2.2897962
 -0.04337353  0.28364185  2.28161478 -0.61804648 
 -0.85997787  0.28364185 -0.04337353  0.85179646  
  0.60368299 -1.00141425 -0.12943648  0.10101557
  0.30309002  0.5916072   0.25458387  0.07418111 
 -0.76378956 -0.87176569  2.48583745]
\end{verbatim}
which produces this response function for $a>0$:
The corresponding response function is shown in Figure~\ref{fig:example_efilter}.
\begin{figure}[htpb]
    \centering
    \includegraphics[width=\columnwidth]{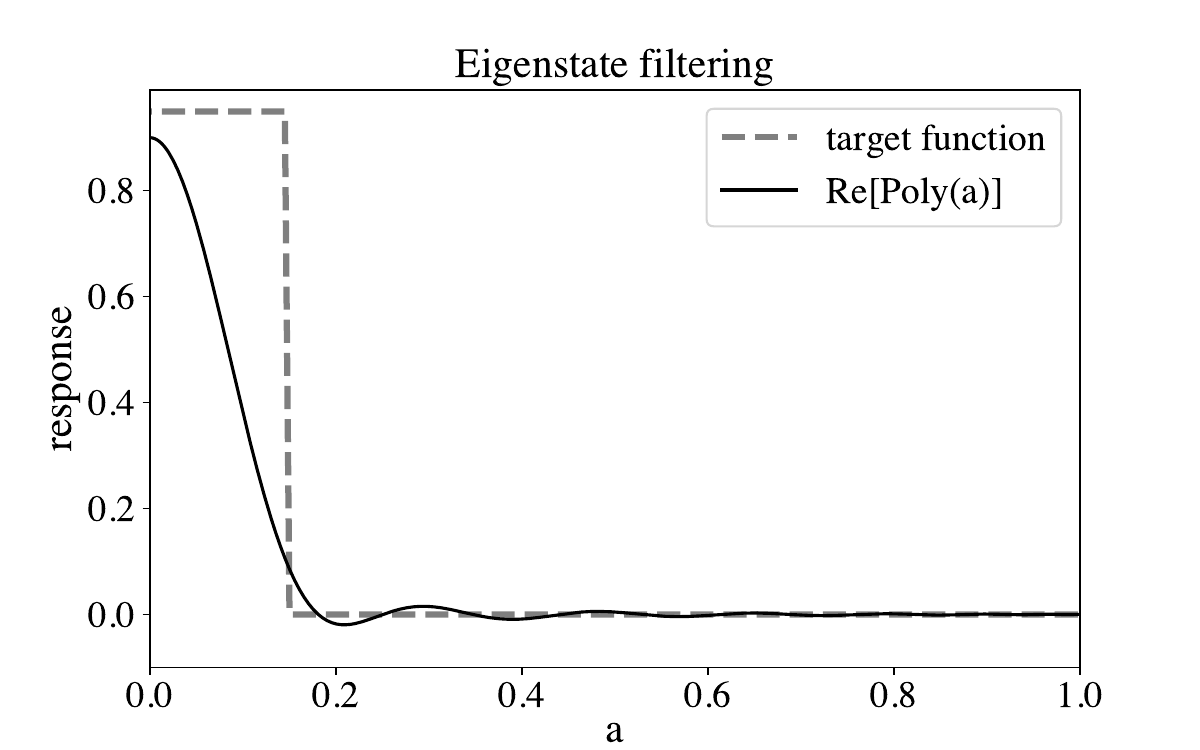}
    \caption{Response function for the polynomial approximation to the eigenstate filtering function with $\delta=0.3$ and $d=30$.}
    \label{fig:example_efilter}
\end{figure}

\noindent This is a better threshold function than the one presented in
section~\ref{subsec:thresh_func}, and the threshold can be located
where desired by changing $\Delta_\lambda$.  

\subsection{Gibbs distribution}

Given $\mathcal{H} = \sum_\lambda \lambda |\lambda\>\<\lambda|$, its corresponding Gibbs distribution state $\rho(\beta) = \frac{1}{Z}e^{-\beta \mathcal{H}}$ is the density matrix
\bea
	\rho(\beta) = \frac{1}{Z}\sum_\lambda e^{-\beta\lambda} |\lambda\>\<\lambda|
\,.
\eea
An approximation to $e^{-\beta a}$ is useful for obtaining $\rho$ using QSP.
To ensure that the function has definite parity, we choose a polynomial approximation to $e^{-\beta |a|}$.

For example, with $\beta=3.5$ and using a degree $d=20$ Taylor series, a set of QSP phase
angles for this polynomial is:
\begin{verbatim}

pyqsp --plot-positive-only --plot-real-only 
--plot-tight-y --seqargs=20,3.5 --plot gibbs

[-0.0311925   0.15173154 -0.42846816  0.59591    
 -0.41539264 -0.16200557  0.12112529 -0.09068282 
 -0.92154011 -0.88213549  1.0199175  -0.88213549
  2.22005254 -0.09068282  0.12112529 -0.16200557 
 -0.41539264  0.59591    -0.42846816 -2.98986111  
  1.53960383]
\end{verbatim}
The corresponding response function is shown in Figure~\ref{fig:example_gibbs}.
\begin{figure}[htpb]
    \centering
    \includegraphics[width=\columnwidth]{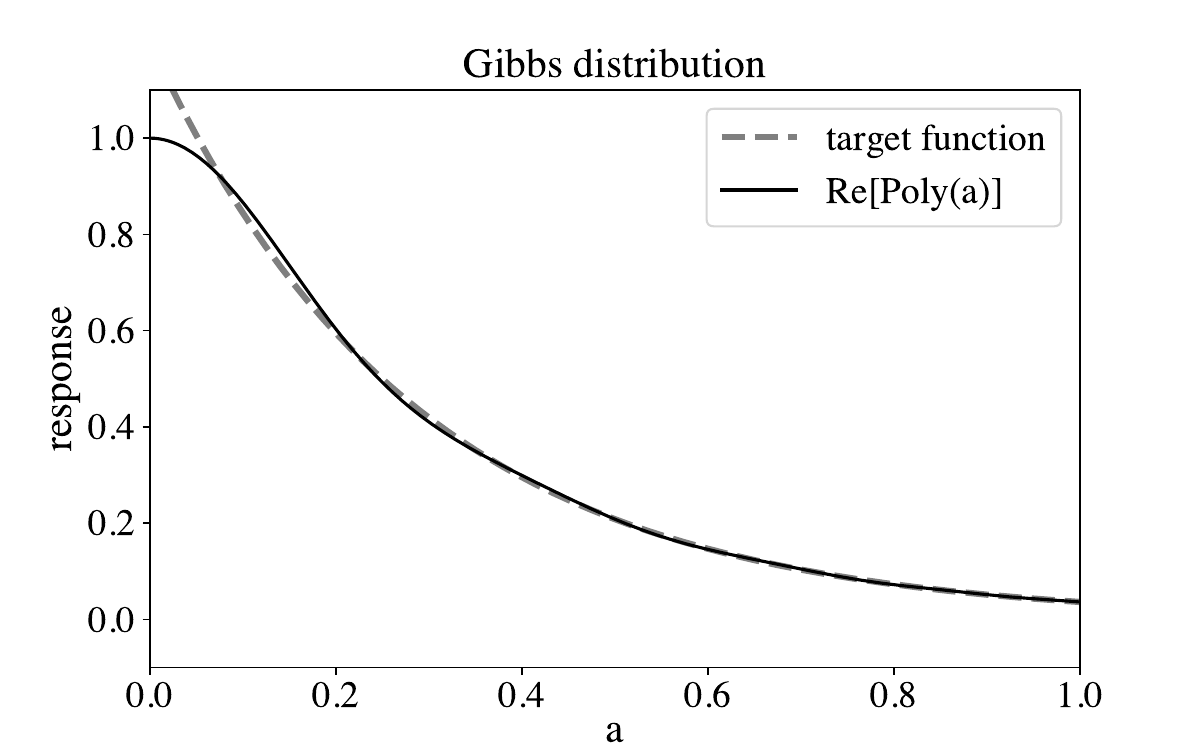}
    \caption{Response function for the polynomial approximation to the Gibbs distribution for $a > 0$ with $\beta=3.5$ using a degree $d=20$ polynomial. The response is scaled to maximize ${\rm Poly}(a)$. Note that the approximation deviates from the target near $a = 0$ as the symmetrized version of the Gibbs distribution is non-analytic about that point.}
    \label{fig:example_gibbs}
\end{figure}

\subsection{ReLU}

The ``rectified linear unit'' activation function, $\text{ReLU}(x) := \text{max}(0,x)$, is popular in
machine learning, and QSP is a natural framework to employ for
realizing such nonlinear activation functions for quantum machine
learning.  A common differentiable approximation of the ReLU function
is the softplus function, which is made into an even function in this
version:
\be
	f(a) = \frac{\ln \lpL 1+e^{\Delta(|a|-\delta)} \rp}{\Delta}
\ee
where $\Delta$ parameterizes the steepness, and $\delta$ the offset of the threshold from $0$.

For example, with $\delta=0.6$ and $\Delta=15$ and using a degree $20$ Taylor series, a set of QSP phase
angles for this polynomial is:
\begin{verbatim}
pyqsp --plot-real-only --plot-tight-y
--seqargs=20,0.6,15 --plot relu

[0.5830891  -0.50867554  0.45797704 -1.83149903  
 0.20084092 -0.11936587  0.97960177  0.53415547 
-0.9957325  -0.9362886   1.24987957 -0.9362886
-0.9957325   0.53415547  0.97960177 -0.11936587  
 0.20084092  1.31009363 -2.68361561  2.63291712 
-0.98770723]
\end{verbatim}
The corresponding response function is shown in Figure~\ref{fig:example_relu}.
\begin{figure}[htpb]
    \centering
    \includegraphics[width=\columnwidth]{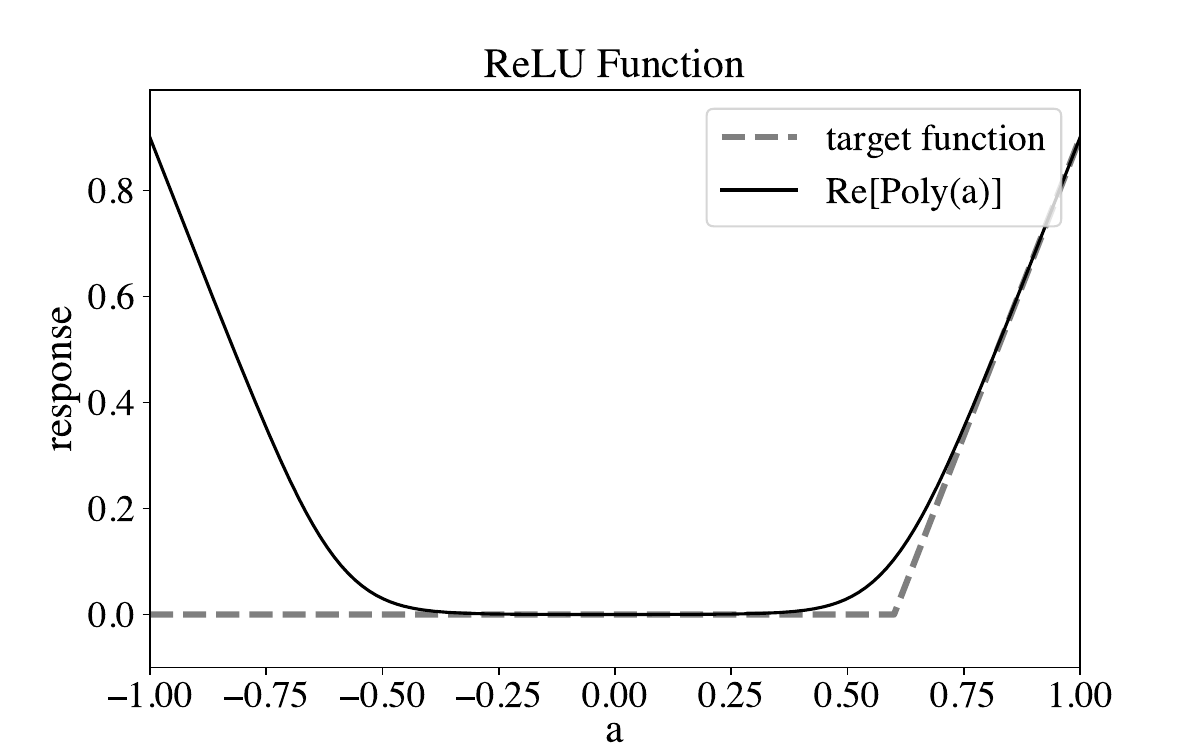}
    \caption{Response function for the polynomial approximation to the ReLU function with $\delta=0.6$ and $\Delta=15$.}
    \label{fig:example_relu}
\end{figure}


\clearpage

\nocite{*}
\bibliographystyle{apsrev4-1}
\bibliography{Refs}

\end{document}